\documentclass[acmsmall]{acmart}

\setcopyright{cc}
\setcctype{by}
\acmDOI{10.1145/3776674}
\acmYear{2026}
\acmJournal{PACMPL}
\acmVolume{10}
\acmNumber{POPL}
\acmArticle{32}
\acmMonth{1}
\received{2025-07-10}
\received[accepted]{2025-11-06}

\usepackage{booktabs}   %
\usepackage{subcaption} %

\usepackage[utf8]{inputenc}            %
\usepackage[T1]{fontenc}               %
\usepackage{amsmath,mathtools,amsthm}  %
\usepackage{graphicx}
\usepackage{thmtools,thm-restate}      %
\usepackage{mathpartir}                %
\usepackage{stmaryrd}                  %
\usepackage{xspace}                    %

\usepackage{url}
\usepackage{xcolor}                    %
\usepackage{tikz}                      %
\usetikzlibrary{trees}
\usepackage[ruled, noend]{algorithm2e}
\usepackage[normalem]{ulem}
\usepackage{listings}
\usepackage{enumerate}                 %
\usepackage{enumitem}
\usepackage{placeins}         %
\usepackage{threeparttable}
\usepackage{fontawesome}               %
\usepackage{pifont}                    %
\usepackage[scaled=0.85]{beramono}
\usepackage[old]{old-arrows}
\usepackage{trimclip}
\DeclareMathAlphabet{\mathbbm}{U}{bbm}{m}{n}
\usepackage{hyperref}
\usepackage{cleveref}                  %
\makeatletter
\AddToHook{cmd/appendix/before}{\crefalias{section}{appendix}}
\AddToHook{cmd/appendix/before}{\crefalias{subsection}{appendix}}
\makeatother
\usepackage{xr-hyper}                  %
\usepackage{trimclip}

\usepackage{quiver}
\usepackage[hyphenbreaks]{breakurl}

\DeclareMathAlphabet{\mathbsf}{\encodingdefault}{\sfdefault}{bx}{n}

\makeatletter
\newcommand{\showlineheight}{{\color{blue}Current line height: \the\baselineskip}}
\newcommand{\showfontsize}{{\color{blue} [Current font size: \f@size pt]}}
\newcommand{\showall}{{\color{blue} [font size: \f@size pt, line height: \the\baselineskip]}}
\makeatother

\newcommand{\LangName}[1]{\textsf{{#1}}}

\newcommand{\Met}{\textsf{\textsc{Met}}\xspace}

\newcommand{\Metp}[1]{\textsf{\textsc{Met}}(\ensuremath{#1})\xspace}
\newcommand{\Metpt}{\textsf{\textsc{Met}}(\ensuremath{\mcT})\xspace}

\newcommand{\Capless}{\textsf{{Capless}}\xspace}

\newcommand{\Feff}{\Feps}

\newcommand{\LambdaEA}{\ensuremath{\lambda_{\mathrm{EA}}}\xspace}

\newcommand{\FepsBold}{\LangName{System~\ensuremath{\textsf{F}^{\boldsymbol{\epsilon}}}}\xspace}
\newcommand{\Feps}{\LangName{System~\ensuremath{\textsf{F}^{\epsilon}}}\xspace}
\newcommand{\Fepss}{\LangName{\ensuremath{\textsf{F}^{\epsilon}}}\xspace}
\newcommand{\Fepssn}{\LangName{System~$\textsf{F}^{\epsilon+\meta{sn}}$}\xspace}
\newcommand{\Fepssns}{\LangName{$\textsf{F}^{\epsilon+\meta{sn}}$}\xspace}

\newcommand{\SystemC}{\LangName{System~C}\xspace}
\newcommand{\SystemXi}{\LangName{System~$\Xi$}\xspace}

\newcommand{\Haskell}{\LangName{Haskell}\xspace}

\newcommand{\Scala}{\LangName{Scala}\xspace}
\newcommand{\Links}{\LangName{Links}\xspace}

\newcommand{\Eff}{\LangName{Eff}\xspace}
\newcommand{\Koka}{\LangName{Koka}\xspace}
\newcommand{\Effekt}{\LangName{Effekt}\xspace}

\newcommand{\CaptureCalculus}{\ensuremath{\LangName{CC}_{<:\BoxSym}}\xspace}

\newcommand{\Frank}{\LangName{Frank}\xspace}
\newcommand{\Helium}{\LangName{Helium}\xspace}

\newcommand{\codesize}{\fontsize{8.6}{10.4}}

\newcommand{\rulesize}{\fontsize{9.5}{12}\selectfont}

\newenvironment{prog}{
  \rulesize\begin{displaymath}
}{
  \end{displaymath}\ignorespacesafterend
}

\definecolor{codegreen}{rgb}{0,0.6,0}
\definecolor{codeblue}{rgb}{0,0,0.8}
\definecolor{codegray}{rgb}{0.5,0.5,0.5}
\definecolor{codepurple}{rgb}{0.58,0,0.82}
\definecolor{codeone}{HTML}{9B26B6}
\definecolor{codetwo}{HTML}{414C87}
\definecolor{backcolour}{rgb}{0.95,0.95,0.92}

\definecolor{myred}{HTML}{BB0000}
\definecolor{myblue}{HTML}{003399}
\definecolor{dred}{HTML}{EB212E}
\definecolor{dblue}{HTML}{2E67F8}
\definecolor{dorange}{HTML}{EFBE7D}
\definecolor{dgreen}{HTML}{00C0A3}
\definecolor{dgray}{HTML}{8A8D8F}
\definecolor{dpurple}{HTML}{9063CD}

\newcommand{\red}[1]{{\color{myred}{#1}}}

\newcommand{\dblue}[1]{{\color{dblue}{#1}}}
\newcommand{\dred}[1]{{\color{dred}{#1}}}
\newcommand{\gray}[1]{{\color{gray}{#1}}}
\newcommand{\black}[1]{{\color{black}{#1}}}

\colorlet{hlcolor}{lightgray}
\newcommand{\highlightwithstyle}[2]{
  {\setlength{\fboxsep}{2pt} \pgfsetfillopacity{0.4} \colorbox{hlcolor}{\pgfsetfillopacity{1}$#1#2$}}
}
\newcommand{\hl}[1]{\mathpalette\highlightwithstyle{#1}}

\newcommand{\xl}[1]{\hl{#1}}

\newcommand{\hle}[1]{{\color{dgray} #1}}

\newcommand{\Modhl}[1]{\hle{\Mod_{#1}}}
\newcommand{\Letmhl}[2]{\Let\hle{\Mod_{#2}}\;}

\makeatletter
\newcommand{\sigb}[2]{
  \@ifmtarg{#1}{:}{:^{#1}} #2
}
\makeatother

\newcommand{\hforall}{\myforall}

\newcommand{\shat}[1]{{#1}^\ast}

\newcommand{\tmcap}[1]{\keyw{cap}_{#1}}

\makeatletter
\newcommand{\tyev}[2]{
  \meta{ev}\,
  \@ifmtarg{#1}{#2}{{#1}^{#2}}
}
\makeatother
\newcommand{\tmev}[1]{\keyw{ev}_{#1}}

\newcommand{\Instctx}{\Omega}

\newcommand{\myforall}{{\text{\rotatebox[origin=c]{180}{$\mathbb{A}$}}}}

\makeatletter
\newcommand{\ogeneric}[2][0.7]{%
  \vphantom{{\oplus}}\mathpalette\o@generic{{#1}{#2}}%
}
\newcommand{\o@generic}[2]{\o@@generic#1#2}
\newcommand{\o@@generic}[3]{%
  \begingroup
  \sbox\z@{$\m@th#1\oplus$}%
  \dimen@=\dimexpr\ht\z@+\dp\z@\relax
  \savebox\tw@[\totalheight]{$\m@th#1\bigcirc$}%
  \makebox[\wd\z@]{%
    \ooalign{%
      $#1\vcenter{\hbox{\resizebox{\dimen@}{!}{\usebox\tw@}}}$\cr
      \hidewidth
      $#1\vcenter{\hbox{\resizebox{#2\dimen@}{!}{$#1\vphantom{\oplus}{#3}$}}}$%
      \hidewidth
      \cr
    }%
  }%
  \endgroup
}
\makeatother

\newcommand{\meta}[1]{\mathsf{#1}}

\newcommand{\one}{\mathbbm{1}}

\newcommand{\transl}[1]{\ensuremath{\llbracket #1 \rrbracket}}

\newcommand{\bigtransl}[1]{\ensuremath{
  \Bigl\llbracket
    \vcenter{\bl #1 \el}
  \Bigr\rrbracket
}}

\let\BoxSym\Box
\newcommand{\squareop}[1]{%
  {\mathrel{\ooalign{\hss\raise-0.1ex\hbox{\scalebox{1.1}{$\BoxSym$}}\hss\cr%
  \kern0.3ex\raise0.12ex\hbox{\scalebox{0.65}{$#1$}}}}}
}

\makeatletter
\DeclareFontEncoding{LS2}{}{\noaccents@}
\DeclareFontSubstitution{LS2}{stix2}{m}{n}
\DeclareSymbolFont{arrows3}{LS2}{stixtt}{m}{n}
\DeclareMathSymbol{\squarelrblackbin}{\mathord}{arrows3}{"89}

\DeclareFontEncoding{LS1}{}{}
\DeclareFontSubstitution{LS1}{stix}{m}{n}
\DeclareSymbolFont{symbols2}{LS1}{stixfrak}{m}{n}
\DeclareMathSymbol{\typecolon}{\mathbin}{symbols2}{"25}
\makeatother

\makeatletter
\newcommand{\Scale}[2][1]{\scalebox{#1}{$\m@th#2$}}
\makeatother

\newcommand{\boxwith}[1]{{\dblue{#1}}}
\newcommand{\mind}[2]{#1_{#2}}
\newcommand{\updlock}[2]{{\llparenthesis #1 \rrparenthesis}_{#2}}

\newcommand{\mR}{\mathbb{R}}

\newcommand{\mcL}{\mathcal{L}}

\newcommand{\mcT}{\mathcal{X}}

\newcommand{\mcS}{\mathcal{S}}
\newcommand{\mcR}{\mathcal{R}}

\newcommand{\emtScp}{{\mcR_{\textsf{scp}}}}
\newcommand{\emtSimp}{{\mcR_{\textsf{simp}}}}

\newcommand{\lock}{\text{\faLock}}

\newcommand{\mlock}{\lock}
\newcommand{\lockwith}[1]{\mlock_{\dblue{#1}}}

\makeatletter
\newcommand{\nameb}[3]{
  {#1 \@ifmtarg{#2}{:}{:_{#2}} #3}
}
\newcommand{\varb}[2]{
  \@ifmtarg{#1}{:}{:_{\dblue{#1}}} #2
}

\newcommand{\varbnd}[3]{
  : \dred{#1} \@ifmtarg{#1}{}{.}
  \@ifmtarg{#2}{}{(\dblue{#2}\mid}
  #3
  \@ifmtarg{#2}{}{)}
}
\newcommand{\Letmn}[3]{
  \@ifmtarg{#1}{\keyw{let}}{\keyw{let}_{#1}} \; \keyw{mod}^{#3}_{#2}\;
}
\newcommand{\Letm}[2]{
  \@ifmtarg{#1}{\keyw{let}}{\keyw{let}_{#1}} \; \keyw{mod}_{#2}\;
}
\newcommand{\Letmf}[2]{
  \@ifmtarg{#1}{\keyw{let}}{\keyw{let}_{#1}} \; #2\;
}

\newcommand{\Casem}[1]{
  \@ifmtarg{#1}{\keyw{case}}{\keyw{case}_{#1}} \;
}
\makeatother

\renewcommand{\Box}{\keyw{mod}}

\newcommand{\amk}[1]{{\langle {#1}{\mkern 1mu \mid\!}\rangle}}
\newcommand{\aex}[1]{{\langle #1\rangle}}
\newcommand{\adj}[2]{{\langle {#1} {\mkern 1mu \mid\mkern 1mu} #2 \rangle}}

\newcommand{\aid}{{\langle\rangle}}
\newcommand{\aeq}[1]{[#1]}

\newcommand{\aconst}[1]{\red{dontuseme}}

\newcommand{\aremove}[1]{\red{dontuseme}}
\newcommand{\act}[2]{#1(#2)}

\makeatletter
\newcommand{\superimpose}[2]{{%
  \ooalign{%
    \hfil$\m@th#1\@firstoftwo#2$\hfil\cr
    \hfil$\m@th#1\@secondoftwo#2$\hfil\cr
  }%
}}
\makeatother

\newcommand{\Mod}{\keyw{mod}}

\newcommand{\locks}[1]{\meta{locks}(#1)}

\newcommand{\To}{\Rightarrow}

\newcommand{\typm}[3]{#1 \vdash #2 \;\gray{@}\, \dblue{#3}}
\newcommand{\typmet}[4]{#1 \vdash #2 : #3 \;\gray{@}\, \dblue{#4}}
\newcommand{\atmode}[1]{\;\gray{@}\, \dblue{#1}}

\makeatletter
\newcommand{\oset}[3][0ex]{%
  \mathrel{\mathop{#3}\limits^{
    \vbox to#1{\kern-2\ex@
    \hbox{$\scriptstyle#2$}\vss}}}}
\makeatother

\newcommand{\stransl}[2]{\ensuremath{\llbracket #1 \rrbracket}_{#2}}

\newcommand{\tr}[1]{\transl{#1}}
\newcommand{\str}[2]{\stransl{#1}{#2}}

\makeatletter
\newcommand{\inc}[3]{
  \@ifmtarg{#1}{}{#1 \mathrel{\typecolon}}
  #2 \sqsubseteq #3
}
\makeatother

\newcommand{\Row}{\meta{Row}}

\newcommand{\Scope}{\meta{Scope}}
\newcommand{\Effect}{\meta{Effect}}
\newcommand{\KindEffect}{\meta{Effect}}

\newcommand{\Pure}{\meta{Abs}}
\newcommand{\Any}{\meta{Any}}

\newcommand{\Abs}{-}

\newcommand{\earr}[3]{#1 \to^{#3} #2}
\newcommand{\elambda}[1]{\lambda^{#1}}

\makeatletter
\DeclareRobustCommand{\Circle}{%
  \mathbin{\mathpalette\on@ntimes\relax}%
}
\newcommand{\on@ntimes}[2]{%
  \vcenter{\hbox{%
    \sbox0{\m@th$#1\otimes$}%
    \setlength\unitlength{\wd0}%
    \begin{picture}(1,1)
    \linethickness{0.35pt}
    \put(.5,.5){\circle{.8}}
    \end{picture}%
  }}%
}
\makeatother

\newcommand{\txmark}{\text{\textcolor{black}{\ding{51}}{\textcolor{black}{\kern-0.7em\ding{55}}}}}

\newcommand{\Def}{\keyw{def}\;}

\newcommand{\At}{\;\keyw{at}\;}
\newcommand{\Try}{\keyw{try}}
\makeatletter
\newcommand{\barr}[4]{
  \@ifmtarg{#2}{(#1)}{(#1,#2)}
  \To #3 \eslash #4
}
\newcommand{\carr}[3]{
  (
  \@ifmtarg{#1}{#2}{#1
    \@ifmtarg{#2}{}{, #2}
  }
  )
  \To #3
}
\newcommand{\carrsingle}[3]{
  \@ifmtarg{#2}{}{(}
  \@ifmtarg{#1}{#2}{#1
    \@ifmtarg{#2}{}{, #2}
  }
  \@ifmtarg{#2}{}{)}
  \To #3
}
\newcommand{\block}[3]{
  \{
  (
  \@ifmtarg{#1}{#2}{#1
    \@ifmtarg{#2}{}{, #2}
  }
  )
  \To #3 \}
}
\newcommand{\blocksingle}[3]{
  \{
  (\@ifmtarg{#1}{#2}{#1
    \@ifmtarg{#2}{}{, #2}}
  )
  \To #3 \}
}
\newcommand{\typxi}[4]{
  #1 \vdash #2 : #3
  \@ifmtarg{#4}{}{\mid {#4}}
}
\newcommand{\typf}[4]{
  #1 \vdash #2 : #3
  \@ifmtarg{#4}{}{\mid {#4}}
}
\newcommand{\typfe}[5]{
  #1 \vdash_{#5} #2 : #3
  \@ifmtarg{#4}{}{\mid {#4}}
}
\newcommand{\typfev}[5]{
  #1 \vdash_{#5} #2 : #3
  \@ifmtarg{#4}{}{\;\dred{\mid #4}}
}
\makeatother
\newcommand{\eslash}{\mathop{\slash}}
\newcommand{\CBox}{\keyw{box}\;}
\newcommand{\CUnbox}{\keyw{unbox}\;}

\newcommand{\subtype}{\leqslant}

\newcommand{\ftv}[1]{\meta{ftv}(#1)}

\newcommand{\fv}[1]{\meta{fv}(#1)}
\newcommand{\dom}[1]{\meta{dom}(#1)}

\newcommand{\BL}[1]{\mathsf{bl}(#1)}

\newcommand{\ol}[1]{\overline{#1}}

\newcommand{\evar}{\varepsilon}

\newcommand{\Fork}{\keyw{fork}}
\NewDocumentCommand{\ForkC}{ m m O{dual} }{\Fork^{#2}_{#1}} %

\newcommand{\refa}[1]{{\color{red!80} \renewcommand{\color}[1]{}{#1}\ifthenelse{\equal{#1}{}}{}{\,}(1)}}
\newcommand{\refb}[1]{{\color{blue!80} \renewcommand{\color}[1]{}{#1}\ifthenelse{\equal{#1}{}}{}{\,}(2)}}
\newcommand{\refc}[1]{{\color{violet} \renewcommand{\color}[1]{}{#1}\ifthenelse{\equal{#1}{}}{}{\,}(3)}}
\newcommand{\refd}[1]{{\color{purple} \renewcommand{\color}[1]{}{#1}\ifthenelse{\equal{#1}{}}{}{\,}(4)}}
\newcommand{\refdno}[1]{{\color{purple} \renewcommand{\color}[1]{}{#1}\ifthenelse{\equal{#1}{}}{}{\,}}}
\newcommand{\refe}[1]{{\color{cyan}   \renewcommand{\color}[1]{}{#1}\ifthenelse{\equal{#1}{}}{}{\,}(5)}}
\newcommand{\reff}[1]{{\color{magenta}\renewcommand{\color}[1]{}{#1}\ifthenelse{\equal{#1}{}}{}{\,}(6)}}
\newcommand{\refg}[1]{{\color{brown}  \renewcommand{\color}[1]{}{#1}\ifthenelse{\equal{#1}{}}{}{\,}(7)}}
\newcommand{\refh}[1]{{\color{orange} \renewcommand{\color}[1]{}{#1}\ifthenelse{\equal{#1}{}}{}{\,}(8)}}

\renewcommand{\descriptionlabel}[1]{\hspace{\labelsep}\textrm{\color{gray!85!black}#1}}
\setlist[description]{itemsep=1ex}

\newcommand{\focus}[1]{{\color{gray}$\boxed{#1}$}}
\newcommand{\LHS}{\mathrm{LHS}}
\newcommand{\RHS}{\mathrm{RHS}}

\newcommand{\notsmall}{}

\newcommand{\tylabi}[2]{\text{T-\scshape{#1}\gray{-#2}}}
\newcommand{\semlabi}[2]{\text{E-\scshape{#1}\gray{-#2}}}

\newcommand{\slab}[1]{\textrm{#1}}
\newcommand{\semlab}[1]{\text{\scshape{E-#1}}}

\newcommand{\tylab}[1]{\text{\scshape{T-#1}}}

\newcommand{\mtylab}[1]{\text{\scshape{MT-#1}}}

\newcommand{\var}[1]{\mathit{#1}}

\newcommand{\keyw}[1]{{{\mathbsf{#1}}}}
\newcommand{\force}[1]{\text{$#1$}}

\newcommand{\Localeffect}[1]{\keyw{local}\;#1\;\keyw{in}\;}
\newcommand{\Handle}{\keyw{handle}}
\newcommand{\Handler}{\keyw{handler}}
\newcommand{\NHandler}{\keyw{nhandler}}

\newcommand{\With}{\;\keyw{with}\;}

\newcommand{\Let}{\keyw{let}\;}

\newcommand{\In}{\;\keyw{in}\;}

\newcommand{\Do}{\keyw{do}\;}
\newcommand{\Doy}{\keyw{do}}
\newcommand{\Ret}{\keyw{return}\;}

\newcommand{\Pre}[1]{\mathsf{Pre}(#1)}

\newcommand{\Value}{\mathsf{Value}}

\makeatletter
\DeclareRobustCommand{\circbullet}{\mathbin{\vphantom{\circ}\text{\circbullet@}}}
\newcommand{\circbullet@}{%
  \check@mathfonts
  \m@th\ooalign{%
    \clipbox{0 0 0 {\dimexpr\height-\fontdimen22\textfont2}}{$\bullet$}\cr
    $\circ$\cr
  }%
}
\DeclareRobustCommand{\bulletcirc}{\mathbin{\text{\bulletcirc@}}}
\newcommand{\bulletcirc@}{%
  \check@mathfonts
  \m@th\ooalign{%
    \raisebox{\fontdimen22\textfont2}{\clipbox{0 {\fontdimen22\textfont2} 0 0}{$\bullet$}}\cr
    $\circ$\cr
  }%
}
\makeatother

\newcommand{\Unit}{{()}}

\newcommand{\code}[1]{{\codesize\texttt{#1}}}

\newcommand{\Yield}{{\code{yield}}}

\newcommand{\namei}[2]{{\color{dred}\text{$\text{$\mathit{#1}$}_{#2}$}}}

\newcommand{\defeq}{\;\dred{\doteq}\;\;}
\newcommand{\treq}{\;\;\dred{=}\;\;}

\newcommand{\Hvdash}{\;\;\gray{\vdash}\;\;}
\newcommand{\Hnotvdash}{\;\;\gray{\not\vdash}\;\;}
\newcommand{\Hcolon}{\;\;\gray{\colon}\;}
\newcommand{\Hmid}{\;\;\gray{\mid}\;\;}
\newcommand{\Hatmode}[1]{\;\;\gray{@}\;\; \dblue{#1}}

\newcommand{\TUnit}{\code{1}}
\newcommand{\Int}{\code{Int}}

\newcommand{\sto}{\twoheadrightarrow}

\newcommand{\reducesto}{\mathrel{\leadsto}}
\makeatletter
\newcommand{\nreducesto}[2]{
  \mathrel{{}_{\@ifmtarg{#1}{}{{#1}}}{\leadsto}_{{\@ifmtarg{#2}{}{{#2}}}}}
}
\makeatother

\newcommand{\ba}{\begin{array}}
\newcommand{\ea}{\end{array}}

\newcommand{\bl}{\ba[t]{@{}l@{}}}
\newcommand{\el}{\ea}

\renewenvironment{displaymath}{\notsmall\[}{\]\normalsize\ignorespacesafterend}
\newenvironment{syntax}{\begin{displaymath}\ba{@{}l@{\quad}r@{~}c@{~}l@{}}}{\ea\end{displaymath}\ignorespacesafterend}

\newenvironment{reductions}{\begin{displaymath}\ba{@{}l@{\quad}@{}r@{~}c@{~}l@{}}}{\ea\end{displaymath}\ignorespacesafterend}
\newenvironment{nreductions}{\begin{displaymath}\ba{@{}l@{\quad}@{}r@{~}c@{~}l@{}}}{\ea\end{displaymath}\ignorespacesafterend}

\newcommand{\bs}{\ba[t]{@{}l@{\quad}r@{~}c@{~}l@{}}}
\newcommand{\es}{\ea}

\newcommand{\EC}{\mathcal{E}}

\newenvironment{mathparshrink}
{\begin{@empty}%
    \begin{mathpar}}
      {\end{mathpar}\end{@empty}}

\lstdefinestyle{ran}{
    commentstyle=\color{codegray},
    numberstyle=\tiny,
    stringstyle=\ttfamily,
    basicstyle=\codesize\selectfont\ttfamily,
    breakatwhitespace=false,
    breaklines=true,
    captionpos=b,
    keepspaces=true,
    numbers=none,
    numbersep=5pt,
    showspaces=false,
    showstringspaces=false,
    showtabs=false,
    tabsize=2,
    columns=fullflexible,
    aboveskip=.9\medskipamount,
    belowskip=.9\medskipamount,
    keywords=[1]{
      do, mask, handle, with, let, in, resume, return, eff, type,
      fst, snd, if, then, else, case, of, fun, raise, maska, data, ev
    },
    keywordstyle=[1]\bfseries,
    keywords=[2]{
      Int, List, Maybe, String, Bool, Pure, Any, Proc, Abs,
      true, false, nil, cons, just, nothing, proc, N
    },
    keywordstyle=[2]\color{codegreen},
    keywords=[3]{
      choose, fail, get, put, ufork, fork, yield, suspend, abort, log, throw, ask, foo, leak,
      bar, baz, Yield, Gen, State, Fork, Coop, Queue, UCoop, Read
    },
    keywordstyle=[3]\color{dblue},
    morecomment = [l]{\#},
    literate={
      {->}{{$\to$}}{2}
      {=>}{{$\Rightarrow$}}{2}
      {\#>}{{$\sharparrow$}}{2}
      {<=}{{$\subtype$}}{1}
      {|-}{{$\vdash$}}{1}
      {|/-}{{$\nvdash$}}{1}
      {@}{{$\gray{\texttt{@}}$}}{1}
      {forall}{{$\forall$}}{1}
      {hforall}{{$\hforall$}}{1}
      {Pi}{{$\Pi$}}{1}
      {Unit}{{{\color{codegreen}1}}}{1}
      {Gamma}{{{$\Gamma$}}}{1}
      {earr}{{$\xrightarrow{\texttt{e}}$}}{2}
      {earrg}{{$\xrightarrow{\texttt{\leff{Gen} \ltype{Int}, e}}$}}{2}
    },
    escapeinside={<@}{@>}
}

\newcommand{\sharparrow}{\mathrel{\mkern3mu\raisebox{-.1ex}{\scalebox{1}[1]{\#}}\mkern-17mu\To}}

\newcommand{\leff}[1]{{\color{dblue}#1}}
\newcommand{\ltype}[1]{{\color{codegreen}#1}}

\lstdefinestyle{koka}{
    commentstyle=\color{codegray},
    numberstyle=\tiny,
    stringstyle=\ttfamily\small,
    basicstyle=\codesize\selectfont\ttfamily,
    breakatwhitespace=false,
    breaklines=true,
    captionpos=b,
    keepspaces=true,
    numbers=none,
    numbersep=5pt,
    showspaces=false,
    showstringspaces=false,
    showtabs=false,
    tabsize=2,
    columns=fullflexible,
    aboveskip=.9\medskipamount,
    belowskip=.9\medskipamount,
    keywords=[1]{
      do, mask, handle, with, let, in, resume, return, effect,
      fst, snd, if, then, else, case, of
    },
    keywordstyle=[1]\bfseries,
    keywords=[2]{
      List, Int, Unit, list
    },
    keywordstyle=[2]\color{codegreen},
    keywords=[3]{
      choose, fail, get, put, ufork, fork, yield, suspend, abort, log, throw, ask, foo, leak,
      bar, baz, Yield, Gen, State, Fork, Coop
    },
    keywordstyle=[3]\color{dblue},
    morecomment = [l]{\#},
    literate={
      {=>}{{$\Rightarrow$}}{2}
      {->}{{$\rightarrow$}}{2}
    }
}

\newif\ifnoappendix

\ifnoappendix
\fi

\begin{document}

\title{Rows and Capabilities as Modal Effects} %

\author{Wenhao Tang}
\orcid{0009-0000-6589-3821}
\email{wenhao.tang@ed.ac.uk}
\affiliation{%
  \institution{The University of Edinburgh}
  \country{UK}}

\author{Sam Lindley}
\orcid{0000-0002-1360-4714}
\email{sam.lindley@ed.ac.uk}
\affiliation{%
  \institution{The University of Edinburgh}
  \country{UK}}

\begin{abstract}
  Effect handlers allow programmers to model and compose computational
  effects modularly.
  Effect systems statically guarantee that all effects are handled.
  Several recent practical effect systems are based on either row
  polymorphism or capabilities.
  However, there remains a gap in understanding the precise
  relationship between effect systems with such disparate foundations.
  The main difficulty is that in both row-based and capability-based
  systems, effect tracking is typically \emph{entangled} with other
  features such as functions.

  We propose a uniform framework for encoding, analysing, and
  comparing effect systems.
  Our framework exploits and generalises modal effect types, a recent
  novel effect system which \emph{decouples} effect tracking from
  functions via modalities.
  Modalities offer fine-grained control over when and how effects are
  tracked, enabling us to express different strategies for effect
  tracking.
  We give encodings as macro translations from existing row-based and
  capability-based effect systems into our framework and show that
  these encodings preserve types and semantics.
  Our encodings reveal the essence of effect tracking mechanisms in
  different effect systems, enable a direct analysis on their
  differences, and provide practical insights on language design.

\end{abstract}

\begin{CCSXML}
<ccs2012>
   <concept>
       <concept_id>10003752.10010124.10010125.10010130</concept_id>
       <concept_desc>Theory of computation~Type structures</concept_desc>
       <concept_significance>500</concept_significance>
       </concept>
   <concept>
       <concept_id>10003752.10003790.10011740</concept_id>
       <concept_desc>Theory of computation~Type theory</concept_desc>
       <concept_significance>500</concept_significance>
       </concept>
   <concept>
       <concept_id>10003752.10010124.10010125.10010126</concept_id>
       <concept_desc>Theory of computation~Control primitives</concept_desc>
       <concept_significance>500</concept_significance>
       </concept>
 </ccs2012>
\end{CCSXML}

\ccsdesc[500]{Theory of computation~Type structures}
\ccsdesc[500]{Theory of computation~Type theory}
\ccsdesc[500]{Theory of computation~Control primitives}

\keywords{effect handlers, effect types, modal types}

\maketitle

\section{Introduction}
\label{sec:introduction}

Effect handlers~\citep{PlotkinP13} provide a powerful abstraction to
define and compose computational effects including state, concurrency,
and probability. %
Effect systems statically ensure that all effects used in a program
are handled.
The literature includes much work on effect systems for effect
handlers
based on a range of different theoretical foundations. Two of the most
popular and well-studied approaches are row-based effect
systems~\citep{linksrow,koka,frank} and capability-based effect
systems~\citep{BrachthauserSO20,BrachthauserSLB22,BoruchGruszeckiOLLB23}.

Row-based effect systems, as in the languages
\Koka~\citep{koka,XieCIL22}, \Links~\citep{linksrow}, and
\Frank~\citep{frank}, follow the traditional monadic reading of
effects: effects are what computations do when they run.
They treat effect types as a row of effects and annotate each function
arrow with an effect row.
For modularity, they implement \emph{parametric effect polymorphism}
via row polymorphism~\citep{Remy94,Leijen05}.
For example, a standard application function in
\Feps~\citep{XieBHSL20}, a core calculus of \Koka, has type:
\[
  \forall\evar . (\Int\to^\evar\TUnit) \to \Int \to^\evar \TUnit
\]
It is polymorphic in its effects $\evar$, which must agree with the
effect performed by its first argument.

Capability-based effect systems, as in the language
\Effekt~\citep{BrachthauserSO20,BrachthauserSLB22} and an extension to
\Scala 3~\citep{BoruchGruszeckiOLLB23}, adopt a contextual reading of
effects: effects are capabilities provided by the context.
Treating effects as capabilities enables a notion of \emph{contextual
effect polymorphism}~\citep{BrachthauserSO20} which allows
effect-polymorphic reuse of functions without effect variables.
For example, an uncurried application function in
\SystemC~\citep{BrachthauserSLB22}, a core calculus of \Effekt, has
type:
\[
  \carr{f:\carrsingle{\Int}{}{\TUnit}}{\Int}{\TUnit}
\]
The argument $f$ is a capability.
It is a second-class function that cannot be returned as a value.
It can use any capabilities the context provides.
We write $\To$ for second-class functions.
For a curried application function, which requires returning a
function, we must capture capabilities in types:
\[
  \carr{}{f:\carrsingle{\Int}{}{\TUnit}}{(\carrsingle{\Int}{}{\TUnit}\At\{f\})}
\]
Its result has type $(\carrsingle{\Int}{}{\TUnit}\At\{f\})$.
As well as specifying an argument and result types as usual, this type
also includes a \emph{capture set} $\{f\}$ which records that the
returned function may use the capability $f$ bound by the argument
type $(f:\carrsingle{\Int}{}{\TUnit})$.

Though row-based and capability-based effect systems are both
well-studied, their relationship is not.
In this paper, we aim to bridge this gap in the literature by
encoding both styles of effect systems into a uniform framework.
\citet{YoshiokaSI24} propose a parameterised calculus which can be
instantiated to various row-based effect systems, but they point out
that it is challenging future work to extend their approach to
capability-based effect systems.
Row-based and capability-based effect systems differ significantly in
both theoretical foundations and interpretations of effects.
Moreover, their mechanisms for tracking effects are \emph{entangled}
with other features such as functions.
For instance, as we have seen above, a function arrow in \Feps is not
only a standard function type but also provides effect annotations.
Similarly, a function arrow in \SystemC may bind capabilities.
The entanglement of effect tracking with such features is the central
challenge in analysing the differences between such effect systems.

An alternative foundation for effect systems has recently emerged in
the form of \emph{modal effect types} (\Met)~\citep{TangWDHLL25}, a
novel approach to effect systems based on multimodal type
theory~\citep{GratzerKNB21,Gratzer23,KavvosG23}.
\Met \emph{decouples} effect tracking from standard type and term
constructs via modalities.
For instance, an application function in \Met has a plain function
type $(\Int\to\TUnit)\to\Int\to\TUnit$.
This type imposes no restriction on how effects from the context may
be used.
To control the use of effects, we can add modalities to the type.
For example, the type
$\boxwith{\aeq{\code{yield}}}(\Int\to\TUnit)\to\Int\to\TUnit$
restricts the argument function to use only the operation
$\code{yield}$ by wrapping it with the \emph{absolute modality}
$\aeq{\code{yield}}$ (modalities have higher precedence than function
arrows); the type
$\boxwith{\aeq{}}(\Int\to\TUnit)\to\boxwith{\aeq{}}(\Int\to\TUnit)$
restricts both the argument and result functions to be pure.

\citet{TangWDHLL25} focus on the pragmatics of \Met, especially how
modalities enable concise type signatures of higher-order functions
without losing modularity.
In this paper we exploit the observation that the decoupling of effect
tracking via modalities leads to a tangible increase in flexibility
and expressivity compared to typical effect systems whose effect
tracking is entangled with other features.
Such decoupling allows us to encode a range of effect systems,
including those based on rows and capabilities, in a uniform
framework.

\medskip

We introduce \Metp{\mcT}, a System F-style core calculus with modal
effect types parameterised by an \emph{effect structure} $\mcT$.
The effect structure is our main extension to \Met~\citep{TangWDHLL25}.
An effect structure, inspired by prior work on abstracting row and effect
types~\citep{rose,HubersM23,YoshiokaSI24}, defines the structure of
effect collections.
\Met hardwires the underlying effect structure to
scoped rows~\citep{Leijen05}.
In contrast, \Metp{\mcT} allows us to smoothly account for the
different treatments of effect collections adopted by different effect
systems, such as sets~\citep{BauerP13,BrachthauserSLB22}, simple
rows~\citep{rose}, %
and scoped rows~\citep{Leijen05,koka,frank}.
Parameterising by the effect structure enables us to separate the
bureaucracy of managing effect collections from our main concern which
is how to use modalities to encode different effect tracking
mechanisms.

\Metp{\mcT} has two further extensions to \Met.
The first extension is \emph{modality-parameterised handlers}.
This is a natural generalisation of effect handlers to be
parameterised by a modality which is used to wrap continuations. This
extension is crucial for the encodings of \Feps and \SystemC as we
will see in \Cref{sec:overview-handlers}.
The second extension is \emph{local labels}, a minimal extension which
allows us to dynamically generate operation
labels~\citep{VilhenaP23}. This extension is crucial for encoding
{named handlers}~\citep{BiernackiPPS20,ZhangM19,BrachthauserSO20}
(also called lexically-scoped handlers) as adopted in some languages,
especially those with capability-based effect systems like \Effekt.

As the main novelty of this paper,
we encode, as \emph{macro translations}~\citep{Felleisen91}, various
effect systems based on rows and capabilities into our uniform
framework \Metp{\mcT}.
We prove that our encodings preserve typing and operational semantics.
Our encodings do not heavily alter the structure of programs but
mostly merely insert terms for manipulating modalities;
our semantics preservation theorems establish a strong correspondence
between the behaviours of source calculi and their translations.
Our primary case studies are encodings of \Feps~\citep{XieBHSL20}, a
core calculus of \Koka with a row-based effect system, and of
\SystemC~\citep{BrachthauserSLB22}, a core calculus of \Effekt with a
capability-based effect system.
By encoding effect systems into a uniform framework, we can directly
reason about the differences the effect tracking mechanisms of
different effect systems
(\Cref{sec:overview-comparison,sec:overview-handler-comparison}).
Our encodings also offer practical insights for language designers
(\Cref{sec:discussion}).

Beyond analysing differences between effect systems, \Metp{\mcT} opens
up interesting future research directions. First, \Metp{\mcT} gives a
uniform intermediate representation for different effect systems which
enables us to design type-directed optimisations without restricting
ourselves to a specific effect system.
Second, \Metp{\mcT} allows us to design a new effect system by
directly giving its encoding into \Metp{\mcT} instead of starting from
scratch. Type soundness and effect safety of \Metp{\mcT} guarantee the
corresponding properties hold for the new effect system.

\medskip

The main contributions of this paper are as follows.
\begin{itemize}
  \item We give a high-level overview of \Metpt and a high-level
    overview of how to encode row-based and capability-based effect
    systems into \Metpt which we use to compare row-based and
    capability-based effect systems (\Cref{sec:overview}).
  \item We formally define \Metp{\mcT} (\Cref{sec:core-calculus})
    including our three extensions to \Met: effect structures,
    modality-parameterised handlers, and local labels. We prove type
    soundness and effect safety of \Metp{\mcT} for any effect structure
    $\mcT$ satisfying certain natural validity conditions.
  \item We formally define the encoding of \Feps, a core calculus with
    a row-based effect system \`a la \Koka, into \Metp{\emtScp} with
    the theory $\emtScp$ for scoped rows
    (\Cref{sec:row-based-effect-systems}).
  We prove the encoding preserves types and semantics.
  \item We formally define the encoding of \SystemC, a core calculus
    with a capability-based effect system \`a la \Effekt, into
    \Metp{\mcS} with the theory $\mcS$ for sets
    (\Cref{sec:capability-based-effect-systems}).
  We prove the encoding preserves types and semantics.
  \item We discuss encodings of further effect systems, practical
    insights for language design provided by our encodings, as well as
    potential extensions to \Metp{\mcT} (\Cref{sec:more-encodings}).
\end{itemize}
\Cref{sec:related-work} discusses related and future work.
The full specifications, proofs, and appendices can
be found in the extended version of the paper~\citep{TangS25arxiv}.

\section{Overview}
\label{sec:overview}

In this section we give a high-level overview of the main ideas of the
paper.
We begin with a brief introduction to modal
effects~\citep{TangWDHLL25} in \Metp{\mcT} and examples of effect
theories $\mcT$.
We briefly describe the row-based effect system of
\Feps~\citep{XieBHSL20} and the capability-based effect system of
\SystemC~\citep{BrachthauserSLB22} along with their encodings into
\Metp{\mcT}.
We use these encodings to directly compare the different systems
in a uniform framework.
We specifically consider encodings of the different kinds of effect
handlers provided by the different systems.
We also briefly discuss the results of encoding other effect systems
in \Metp{\mcT}.

\subsection{Modal Effects and \Metp{\mcT}}
\label{sec:met-introduction}

\Metp{\mcT} is a System F-style core calculus.
Every well-typed term in System F is also well-typed in \Metp{\mcT}.
For example, we may define a higher-order application function as
follows.
\begin{prog}
  \namei{app}{\Metpt}\defeq
  \lambda f^{\Int\to \TUnit} . \lambda x^\Int . f\; x
  \Hcolon (\Int\to\TUnit)\to\Int\to\TUnit
\end{prog}
We use meta-level macros defined by $\doteq$ in red to refer to code
snippets.

\subsubsection{Effect Contexts}
\Metpt adopts a contextual reading of effects.
Effectful operations are ascribed a type signature, either globally or
locally.
For our examples we begin by assuming global operations $\code{yield}
: \Int \sto \TUnit$ and $\code{ask}:\TUnit\sto\Int$.
Typing judgements include an \emph{ambient effect context} which
tracks the operations that may be performed.
Consider the following function.
\begin{prog}
  \Hvdash \namei{gen}{\Metpt}\defeq
  \lambda x^\Int . \Do\code{yield}\;x
  \Hcolon \Int \to \TUnit
  \Hatmode{\code{yield}}
\end{prog}
It has type $\Int \to \TUnit$.
When applied it performs the $\code{yield}$ operation using the $\Doy$
syntax.
The judgement specifies the effect context with the syntax
$\!\!\atmode{\code{yield}}$, which tracks the possibility of
performing $\code{yield}$.
We can now apply \namei{app}{\Metpt} to \namei{gen}{\Metpt} and $42$ as follows.
\begin{prog}
  \Hvdash (\lambda f^{\Int\to \TUnit} . \lambda x^\Int . f\;x)~
    (\lambda x^\Int . \Do\code{yield}\;x)~42
    \Hcolon \TUnit
    \Hatmode{\code{yield}}
\end{prog}
There is a natural notion of subeffecting on effect contexts.
The following judgement is also valid.
\begin{prog}
  \Hvdash
  \lambda x^\Int . \Do\code{yield}\;x
  \Hcolon \Int \to \TUnit
  \Hatmode{\code{yield},\code{ask}}
\end{prog}

\subsubsection{Absolute Modalities}
\label{sec:overview-absolute}
Effect contexts specified by $\!\!\atmode{E}$ belong to typing
judgements instead of types.
As discussed in \Cref{sec:introduction}, \Metpt uses modalities to
track effects in types.
An \emph{absolute modality} $\aeq{E}$ allows us to specify a new
effect context $E$ in types different from the ambient one.
For example, consider the following typing derivation.
\begin{prog}
  \inferrule*
  {
  \lockwith{\aeq{\code{yield}}} \Hvdash
    \lambda x^\Int . \Do\code{yield}\;x
    \Hcolon \Int \to \TUnit
    \Hatmode{\code{yield}}
  }
  {
  \Hvdash \namei{gen'}{\Metpt}\defeq
    \Mod_{\aeq{\code{yield}}}~(\lambda x^\Int . \Do\code{yield}\;x )
    \Hcolon \boxwith{\aeq{\code{yield}}}(\Int \to \TUnit)
    \Hatmode{F}
  }
\end{prog}
This term has type $\boxwith{\aeq{\code{yield}}}(\Int \to \TUnit)$.
We highlight modalities in blue when they appear in types.
The syntax $\Mod_{\aeq{\code{yield}}}$ introduces an absolute modality
$\aeq{\code{yield}}$ which specifies a singleton effect context of
$\code{yield}$ and uses it to override the ambient effect context
$F$.
The typing judgement of the premise uses the new effect context
$\code{yield}$ as its ambient effect context.
The lock $\lockwith{\aeq{\code{yield}}}$ tracks the switch of the
effect context and controls the accessibility of variables on the left
of it.
Only variables that are known not to use effects other than
${\code{yield}}$ may be used.
This is important to ensure effect safety.
For example, consider the following invalid judgement.
\begin{prog}
  {
  f:\Int\to\TUnit \Hnotvdash
  \Mod_{\aeq{\code{yield}}}~(\lambda x^\Int . f\;x)
  \Hcolon \boxwith{\aeq{\code{yield}}}(\Int \to \TUnit)
  \Hatmode{\code{ask}}
  }
\end{prog}
This program is unsafe as $f$ may invoke $\code{ask}$ which we must
not use under effect context $\code{yield}$ specified by the modality
$\aeq{\code{yield}}$.
\Metpt rejects this judgement as it relies on the following invalid
judgement for the inner function.
\begin{prog}
  {
    f:\Int\to\TUnit, \lockwith{\aeq{\code{yield}}} \Hnotvdash
    \lambda x^\Int . f\;x
    \Hcolon \Int \to \TUnit
    \Hatmode{\code{yield}}
  }
\end{prog}
This typing judgement is invalid as the lock
$\lockwith{\aeq{\code{yield}}}$ prevents the use of $f$.
To make it valid, we can annotate the binding of $f$ with the modality
$\aeq{\code{yield}}$ as $f \varb{\aeq{\code{yield}}} \Int \to \TUnit$.
This annotation tracks that the function $f$ may only use the
operation $\code{yield}$.
Such annotated bindings are introduced by modality elimination.
For instance, we can eliminate the modality of \namei{gen'}{\Metpt}
and then apply it via the $\Letm{}{}\!\!$ syntax as follows (where we
elide the typing of the bound term).
\begin{prog}
  \inferrule*
  {
    ... \\
    f\varb{\aeq{\code{yield}}}{\Int\to\TUnit} \Hvdash
    f\;42
    \Hcolon \TUnit
    \Hatmode{\code{yield}}
  }
  {
    \Hvdash
    \Letm{}{\aeq{\code{yield}}} f
      = \Mod_{\aeq{\code{yield}}}~(\lambda x^\Int . \Do\code{yield}\;x)
      \In f\;42
    \Hcolon \TUnit
    \Hatmode{\code{yield}}
  }
\end{prog}
The term $\lambda x^\Int.\Do\code{yield}\;x$ inside the modality
$\aeq{\code{yield}}$ is bound to $f$.
The binding of $f$ is annotated with this absolute modality.
Consequently, the use of $f$ in $f\;42$ requires the ambient effect
context to contain the operation $\code{yield}$.
In general, whether a variable binding $f\varb{\mu} A$ can be used
after a lock $\lockwith{\nu}$ is controlled by a modality
transformation relation which we will introduce in
\Cref{sec:modalities}.

\subsubsection{Relative Modalities}
\label{sec:overview-relative}
As well as being able to specify a fresh effect context from scratch
with an absolute modality, \Metpt also has \emph{relative modalities}
$\aex{D}$ which allow us to extend the ambient effect context with an
\emph{extension} $D$.
For instance, consider the following derivation.
\begin{prog}
  \inferrule*
  {
  \lockwith{\aex{\code{yield}}} \Hvdash
  \lambda x^\Int . \Do\code{yield}\;(\Do\code{ask}\;\Unit)
  \Hcolon \Int \to \TUnit
  \Hatmode{\code{yield},\code{ask}}
  }
  {
  \Hvdash
  \Mod_{\aex{\code{yield}}}~(\lambda x^\Int . \Do\code{yield}\;(\Do\code{ask}\;\Unit) )
  \Hcolon \boxwith{\aex{\code{yield}}}(\Int \to \TUnit)
  \Hatmode{\code{ask}}
  }
\end{prog}
The relative modality $\aex{\code{yield}}$ extends the ambient effect context
$\code{ask}$ with the operation $\code{yield}$.
Consequently, the inside function can use both operations.
Relative modalities are especially useful for giving composable types
to effect handlers.
We refer to \citet{TangWDHLL25} for further details.
We use relative modalities in the encoding of \SystemC as we will see
in \Cref{sec:overview-encoding-caps}.

\subsubsection{Effect Structures}

Improving on \Met, we parameterise \Metp{\mcT} by an effect structure
$\mcT$ which defines the well-formedness relations and equivalence
relations for extensions and effect contexts as well as a subeffecting
relation $E\subtype F$.
In the remainder of the overview, we will use two effect structures:
$\mcS$, which models effect collections as sets of operations, to
encode capability sets in \SystemC, and $\emtScp$, which models effect
collections as scoped rows of operations, to encode effect rows in
\Feps.
Sets are unordered and allow only one occurrence of each label,
whereas scoped rows allow repeated labels and identify rows up to
reordering of non-identical labels.
Both theories support effect variables.
Theory $\mcS$ allows arbitrary numbers of effect variables while
theory $\emtScp$ only allows at most one effect variable in each row
following row polymorphism~\citep{Remy94,Leijen05}.

\subsection{Rows as Modal Effects}
\label{sec:overview-encoding-rows}

\Koka~\citep{KokaLang} has an effect system based on scoped
rows~\citep{Leijen05}.
\Feps~\citep{XieBHSL20} is a core calculus underlying \Koka.
To encode \Feps, we use the effect structure $\emtScp$ of scoped rows.

Function types in \Feps have the form $\earr{A}{B}{E}$, where $E$ is
an effect row that specifies the effects that the function may use.
Effect types in \Feps are entangled with function types.
The key idea of our encoding is to decouple the effect type $E$ from
the function arrow via an absolute modality in \Metp{\emtScp}.
Writing $\tr{-}$ for translations, we translate a function type as
follows.
\begin{prog}
\bl
  \transl{\earr{A}{B}{E}} = \boxwith{\aeq{\tr{E}}}(\tr{A}\to \tr{B})
\el
\end{prog}
An effectful function in \Feps is decomposed into an absolute modality
and a standard function in \Metp{\emtScp}.
For instance, consider the following first-order effectful function in
\Feps which invokes the operation $\code{yield}$ from
\Cref{sec:met-introduction}.
\begin{prog}\bl
  \namei{gen}{\Fepss} \defeq
  \lambda^{\code{yield}} x^\Int . \Do\Yield\;x
  \Hcolon \Int\to^{\code{yield}}\TUnit
\el\end{prog}
(Each $\lambda$-abstraction in \Feps is annotated with an effect row.)
The translation of \namei{gen}{\Fepss} is exactly the function
\namei{gen'}{\Metpt} defined in \Cref{sec:overview-absolute}. We repeat
its definition here for easy reference.
\begin{prog}\bl
  \tr{\namei{gen}{\Fepss}}
  \treq
    \Modhl{\aeq{\code{yield}}}~(\lambda x^\Int . \Do\code{yield}\;x )
    \Hcolon \boxwith{\aeq{\code{yield}}}(\Int \to \TUnit)
\el\end{prog}
On the term level, we insert a modality introduction
$\Mod_{\aeq{\code{yield}}}$ for the $\lambda$-abstraction,
corresponding to the type-level modality $\aeq{\code{yield}}$.
We colour ${\Mod}$ in grey in the translations.
The black parts remain terms with valid syntax
and provide intuitions on the translation.
Remember that the modality ${\aeq{\code{yield}}}$ is a first-class
type constructor and not part of the function type.

As a more non-trivial example including both higher-order functions
and function application, consider the effect-polymorphic application
function in \Feps from \Cref{sec:introduction}.
\begin{prog}\bl
  \namei{app}{\Fepss} \defeq
  \Lambda\evar^\Effect . \lambda f^{\Int\to^\evar\TUnit} . \lambda^{\evar} x^{\Int} .  f\;x
  \Hcolon \forall\evar . (\Int\to^\evar\TUnit) \to \Int \to^\evar \TUnit
\el\end{prog}
This function abstracts over an effect variable $\evar$ which stands
for the effects performed by the argument $f$.
Both $f$ and the inner $\lambda$-abstraction are annotated with
$\evar$ as $f$ is invoked so the effects must match up.
The outer $\lambda$-abstraction is pure as partial application is
pure.
The encoding of \namei{app}{\Fepss} in \Metp{\emtScp} is as follows.
\begin{prog}\ba{r@{}c@{}l}
  \transl{\namei{app}{\Fepss}}
  &\treq&
  \force{\Lambda\evar^\Effect} . \Modhl{\aeq{}}\;(\lambda f^{\boxwith{\aeq{\evar}}(\Int\to\TUnit)} .
    \Modhl{\aeq{\evar}}\;\text{$(\lambda x^\Int . \Letmhl{}{\aeq{\evar}}\; f' = f \In f'\;x))$} \\
  &\Hcolon& \forall\evar^\Effect.\boxwith{\aeq{}}(\text{$\boxwith{\aeq{\evar}}(\Int\to\TUnit)\to \boxwith{\aeq{\evar}}(\Int \to \TUnit)$})
\ea\end{prog}
Each function arrow is associated with an absolute modality reflecting
the effects performed by that function.
For the pure function arrow in the middle, we use the empty absolute
modality $\boxwith{\aeq{}}$.
The type abstraction $\Lambda\evar$ and quantifier $\forall\evar$ are
preserved.
We omit kinds when obvious.
In the term, in addition to modality introduction, we also insert a
modality elimination for $f$ before applying it to $x$.
The use of $f'$ requires that the effect variable $\evar$ is present
in the effect context.

Our term translation from \Feps to \Metp{\emtScp} explicitly decouples
the effect tracking mechanism of \Feps from function abstraction and
application.
This reveals the essence of effect tracking in \Feps.
Each $\lambda$-abstraction $\lambda^E x. M$ in \Feps is encoded in
\Metp{\emtScp} by inserting a modality introduction
$\Mod_{\aeq{\tr{E}}}$.
This demonstrates that a function in \Feps carries its effects.
Each function application $V\;W$ in \Feps is encoded by inserting a
modality elimination $\Letm{}{\aeq{\tr{E}}} f = \tr{V} \In f\;\tr{W}$
for function $V$ of type $\earr{A}{B}{E}$.
This demonstrates that when a function is invoked in \Feps, we need to
provide all effects it may use, as the elimination of $\aeq{\tr{E}}$
and use of $f$ together require $\tr{E}$ to be present in the effect
context.

We give the full encoding of \Feps into \Metp{\emtScp} in
\Cref{sec:row-based-effect-systems}.

\subsection{Capabilities as Modal Effects}
\label{sec:overview-encoding-caps}

\Effekt~\citep{EffektLang} has an effect system based on capabilities.
\SystemC~\citep{BrachthauserSLB22} is a core calculus underlying
\Effekt.
Since \SystemC tracks capabilities as sets, we use the effect structure
$\mcS$ of sets to encode it.

Functions in \SystemC are called \emph{blocks}.
Blocks are second-class in that they must be fully applied and cannot
be returned.
Capabilities are introduced as block variables.
Unlike row-based effect systems which have a separate notion of
operation labels, \SystemC interprets effects as capabilities provided
by the context.
A capability can only be used if it is in scope.

\newcommand{\yld}{{\var{y}}}

\subsubsection{First-Order Blocks}
\label{sec:overview-encoding-caps-blocks}
Let us start with a simple example. Supposing we have a capability
$\var{\yld}:\carrsingle{\Int}{}{\TUnit}$ (for yielding integers) in
the context, we can construct the following block.
\begin{prog}\bl
\var{\yld}:^\ast\carrsingle{\Int}{}{\TUnit} \Hvdash
  \namei{gen}{C} \defeq \blocksingle{x:\Int}{}{\var{\yld}(x)}
  \Hcolon \carrsingle{\Int}{}{\TUnit} \Hmid \{\var{\yld}\} \\
\el\end{prog}
The star $\ast$ on the binding of $\var{\yld}$ indicates that this
block variable is a capability.
Braces delimit blocks. Arguments are wrapped in parentheses. Double
arrows emphasise that blocks are second-class.
The block applies the capability $\var{\yld}$ from the context to the
argument $x$.
The typing judgement tracks a capability set $\{\var{\yld}\}$, which
contains all capabilities that the block may use.
The block arrow itself has no capability annotation.
The above block is simply encoded as a $\lambda$-abstraction in
\Metp{\mcS}.%
\footnote{%
  If we strictly follow the encoding of \SystemC in
  \Cref{sec:encoding-caps}, there would be an extra identity modality
  for the translated function. This modality is crucial for keeping
  the encoding systematic. We omit such identity modalities in the
  overview.}
\begin{prog}
\bl
\shat{\var{\yld}} : \Effect,
\var{\yld} : \boxwith{\aeq{\shat{\yld}}}(\Int \to \TUnit),
\hat{\var{\yld}}\varb{\aeq{\shat{\yld}}}\Int \to \TUnit
\Hvdash
  \tr{\namei{gen}{C}}
  \treq \lambda x^\Int . \hat{\var{\yld}}\;x
  \Hcolon \Int \to \TUnit \Hatmode{\shat{\var{\yld}}} \\
\el
\end{prog}
The most interesting aspect of the encoding is how we encode the
capability $\yld:\carrsingle{\Int}{}{\TUnit}$ in the context.
A capability $\var{\yld}$ in \SystemC can appear as both a type and a
term.
We introduce an effect variable $\shat{\yld}$ of kind $\Effect$ to
represent it at the type level.
We omit kinds in the context when obvious.
We encode the capability $\var{\yld}$ itself as a term variable of
type $\boxwith{\aeq{\shat{\yld}}}(\Int\to\TUnit)$, where the absolute
modality makes sure that whenever $\yld$ is invoked the effect
variable $\shat{y}$ must be present in the effect context.
To avoid repeatedly writing modality eliminations, the modality of
$\yld$ is immediately eliminated and bound to $\hat{\yld}$ after
$\yld$ is introduced.
The translation of the block body directly applies $\hat{\yld}$ to
$x$.
The effect variable $\shat{\yld}$ must be in the effect context
specified by $\!\!\atmode{\shat{\yld}}$ because $\hat{\yld}$ is used.

\subsubsection{Boxes}
In \SystemC we can turn a second-class block into a first-class value
by boxing it.
\begin{prog}\bl
{\yld}:^\ast\carrsingle{\Int}{}{\TUnit} \Hvdash
  \namei{gen'}{C} \defeq \text{$\CBox$}\blocksingle{x:\Int}{}{{\yld}(x)}
  \Hcolon \carrsingle{\Int}{}{\TUnit} \text{$\At$} \{\yld\} \\
\el\end{prog}
This typing judgement has no capability set as it is for values which
are always pure in \SystemC.
The value has type $\carrsingle{\Int}{}{\TUnit} \text{$\At$}\{\yld\}$,
which means it is a boxed block of type $\carrsingle{\Int}{}{\TUnit}$
with capability set $\{\yld\}$.
The block may only use the capability $\yld$.
We can unbox a boxed block $V$ via $\CUnbox V$ which gives back a
second-class block.
We simply encode boxing and unboxing as modality introduction and
elimination in \Metp{\mcS}.
For instance, we encode \namei{gen'}{C} as follows.
\begin{prog}
\bl
\shat{\var{\yld}},
\var{\yld} : \boxwith{\aeq{\shat{\yld}}}(\Int \to \TUnit),
\hat{\var{\yld}}\varb{\aeq{\shat{\yld}}}\Int \to \TUnit
\Hvdash
  \tr{\namei{gen'}{C}}
  \treq \force{\Modhl{\aeq{\shat{\yld}}}}\;(\lambda x^\Int . \hat{\var{\yld}}\;x)
  \Hcolon \boxwith{\aeq{\shat{\yld}}}(\Int \to \TUnit) \Hatmode{\cdot} \\
\el
\end{prog}
The capability set annotation $\!\!\At\{\yld\}$ in the type is encoded
as the absolute modality $\aeq{\shat{\yld}}$.
The encoding shows the connection between boxes of \SystemC and
modalities, supporting the claim of \citet{BrachthauserSLB22} that
boxes of \SystemC are inspired by modal
connectives~\citep{ChoudhuryK20}.

\subsubsection{Higher-Order Blocks}
\label{sec:overview-encoding-caps-higher-order-block}
The situation become more involved when we consider higher-order
blocks that take other blocks as arguments.
This is because \SystemC entangles the introduction and tracking of
capabilities with blocks, especially their construction and
application.

Let us consider the uncurried and curried application functions
(blocks) introduced in \Cref{sec:introduction}.
\begin{prog}
\ba{l@{}l@{}c@{}l@{}c@{}l}
&\namei{app}{C} &\defeq& \block{x:\Int}{f:{\carrsingle{\Int}{}{\TUnit}}}{f(x)}
&\Hcolon& \carr{\Int}{f : \carrsingle{\Int}{}{\TUnit}}{\TUnit} \\%[.5ex]
&\namei{app'}{C} &\defeq&
  \blocksingle{}{f:{\carrsingle{\Int}{}{\TUnit}}}{
  \CBox\blocksingle{x:\Int}{}{f(x)}}
  &\Hcolon& \carr{}{f:\carrsingle{\Int}{}{\TUnit}}{(\carrsingle{\Int}{}{\TUnit}\At\{f\})} \\
\ea
\end{prog}
These are block constructions.
The first block \namei{app}{C} binds the integer parameter $x$ first
because \SystemC requires value parameters like $x$ to appear before
blocks parameters like $f$ in a parameter list.
In addition to behaving like standard $\lambda$-abstractions, block
constructions also play an important role in capability tracking.
Specifically:
\begin{enumerate}[topsep=2pt]
  \item Both \namei{app}{C} and \namei{app'}{C} bind a capability
    ${f:\carrsingle{\Int}{}{\TUnit}}$ for their block bodies. This
    capability $f$ can also be used in the type as shown in the type
    of \namei{app'}{C}.
  \item For soundness, \SystemC assumes that this new capability $f$
    is called directly at least once in the block body even if $f$ may
    actually not be used. (The capability $f$ is indeed called
    directly in \namei{app}{C} but not so in \namei{app}{C'} as being
    boxed.) Consequently, the capability $f$ is always added to the
    capability set of the block body tracked by the typing judgement.
  \item In addition to the new capability $f$, both \namei{app}{C} and
  \namei{app'}{C} allow any capability from the context to be called
  as well.
\end{enumerate}
Our encoding of block constructions in \Metp{\mcS} takes account of
these three constraints and exposes them explicitly via modalities.
For instance, \namei{app}{C} is encoded as follows.
\begin{prog}\ba{r@{}c@{}l}
\transl{\namei{app}{C}} &\treq&
  \Lambda \shat{f} . \Modhl{\aex{\shat{f}}}~(
     \lambda x^\Int . \lambda f^{\boxwith{\aeq{\shat{f}}}(\Int\to \TUnit)} .
    {\Letmhl{}{\aeq{\shat{f}}} \hat{f} = f \In}
    \hat{f}\;x) \\[.2ex]
  &\Hcolon& \forall\shat{f} . \boxwith{\aex{\shat{f}}}(
     \Int \to \boxwith{\aeq{\shat{f}}}(\Int\to \TUnit) \to \TUnit
  )
\ea\end{prog}
For (1), in order to allow the term variable $f$ to appear in types,
we introduce an effect variable $\shat{f}$ and wrap the type
$\Int\to\TUnit$ of the argument $f$ with an absolute modality
$\aeq{\shat{f}}$.
The effect variable $\shat{f}$ represents the term variable $f$ at the
level of types.
Additionally, we immediately eliminate the modality of $f$ to
$\hat{f}$.
As a result, in the context of the application $\hat{f}\;x$ we have
three bindings of $\shat{f}$, $f$, and $\hat{f}$, consistent with the
translation of capability $\yld$ as shown in
\Cref{sec:overview-encoding-caps-blocks}.
For (2) and (3), we use a relative modality $\aex{{\shat{f}}}$ to
wrap the whole function type.
The relative modality adds the effect variable $\shat{f}$ to the
ambient effect context for the function to use, in accordance with
(2).
The relative modality also still allows the function to use effects
from the ambient effect context as we have seen in
\Cref{sec:overview-relative}, in accordance with (3).

The translation of \namei{app'}{C} is similar.
\begin{prog}\ba{r@{~}c@{~}l}
\transl{\namei{app'}{C}} &\treq&
  \Lambda \shat{f} . \Modhl{\aex{\shat{f}}}~(\lambda f^{\boxwith{\aeq{\shat{f}}}(\Int\to \TUnit)} .
    {\Letmhl{}{\aeq{\shat{f}}} \hat{f} = f \In}
    \Modhl{\aeq{\shat{f}}}~(\lambda x . \hat{f}\;x)) \\[.2ex]
  &:& \forall\shat{f} . \boxwith{\aex{\shat{f}}}(
    \boxwith{\aeq{\shat{f}}}(\Int\to \TUnit) \to \boxwith{\aeq{\shat{f}}}(\Int\to \TUnit)
  )
\ea\end{prog}

In general, the translation of block types from \SystemC to
\Metp{\mcS} is as follows, where we let $A$ and $B$ range over value
types and let $T$ range over block types.
\begin{prog}
\transl{\carr{\ol{A}}{\ol{f:T}}{B}} \;=\;
\forall\ol{\shat{f}} . \boxwith{\aex{\ol{\shat{f}}}}(\ol{\transl{A}}\to \ol{\boxwith{\aeq{\shat{f}}}\transl{T}}\to \transl{B})
\end{prog}
A block type is decomposed into a standard function type with extra
modalities and type quantifiers, which makes explicit exactly how
\SystemC introduces and tracks capabilities.

\subsubsection{Block Calls}
\label{sec:overview-encoding-caps-block-calls}
Blocks must be fully applied.
Assuming we have a capability $\yld : \carrsingle{\Int}{}{\TUnit}$ in
the context, we can apply the blocks \namei{app}{C} and
\namei{app'}{C} to the block \namei{gen}{C} as follows.
\begin{prog}\ba{l@{}l@{}c@{}l@{}c@{}l}
  \yld :^\ast \carrsingle{\Int}{}{\TUnit} \Hvdash& \namei{app}{C}(\namei{gen}{C}, 42)
  &\Hcolon& \TUnit &\Hmid& \{\yld\} \\[.5ex]
  \yld :^\ast \carrsingle{\Int}{}{\TUnit} \Hvdash& \namei{app'}{C}(\namei{gen}{C})
  &\Hcolon& \carrsingle{\Int}{}{\TUnit} \text{$\At$} \{\yld\} &\Hmid& \{\yld\}
\ea\end{prog}
(As blocks must be fully applied, we must additionally pass an integer
to \namei{app}{C} --- in this case $42$.)
These are block calls.
Similar to block constructions, block calls in \SystemC not only pass
arguments to a block but also play an important role in capability
tracking.
Specifically:
\begin{enumerate}[topsep=2pt]
  \item Recall that both \namei{app}{C} and \namei{app'}{C} bind a
    capability $f$. Consequently, when calling them with
    \namei{gen}{C}, \SystemC substitutes $f$ with the capability set
    $\{\yld\}$ of \namei{gen}{C} in types. This is reflected by
    $\!\!\At\{\yld\}$ in the type of calling \namei{app'}{C} (before
    substitution it was $\!\!\At\{f\}$).
  \item Recall that \SystemC assumes the capability $f$ bound by
  \namei{app}{C} and \namei{app'}{C} is called directly.
  Consequently, the capability set of the whole block call must be
  extended with the capability set $\{\yld\}$ of the argument
  \namei{gen}{C}.
  This is reflected by the fact that both typing judgements track the
  capability sets $\{\yld\}$ even though the application of
  \namei{app'}{C} does not call $\yld$ directly.
\end{enumerate}
Our encoding of block calls in \Metp{\mcS} takes account of these two
constraints and exposes them explicitly via modalities.
For instance, our example application of \namei{app}{C} is encoded as
follows.
\begin{prog}\ba{l}
  \shat{\var{\yld}},
  \var{\yld} : \boxwith{\aeq{\shat{\yld}}}(\Int \to \TUnit),
  \hat{\var{\yld}}\varb{\aeq{\shat{\yld}}}\Int \to \TUnit
  \Hvdash
  \\
  \hspace{6em}
  \text{$\Letmhl{}{\aex{\shat{\yld}}}$} f = \tr{\namei{app}{C}}\;\shat{\yld}
  \In f\;42\;(\Modhl{\aeq{\shat{\yld}}}\;\tr{\namei{gen}{C}})
  \Hcolon \TUnit
  \Hatmode{\shat{\yld}}
\ea\end{prog}
For (1), recall that in the translation $\tr{\namei{app}{C}}$ we bind
an effect variable $\shat{f}$ to represent the capability $f$ and wrap
the argument type with an absolute modality $\aeq{\shat{f}}$.
Thus for the application of $\tr{\namei{app}{C}}$, we instantiate the
effect variable $\shat{f}$ with $\shat{\yld}$ and box the argument
$\tr{\namei{gen}{C}}$
with the absolute modality $\aeq{\shat{\yld}}$.
For (2), the elimination of the relative modality $\aex{\shat{\yld}}$
of $\tr{\namei{app}{C}}\;\shat{\yld}$ and the use of $f$ ensure that
$\shat{\yld}$ must be present in the effect context.

The translation of the call of \namei{app'}{C} is similar.
\begin{prog}\bl
  \shat{\var{\yld}},
  \var{\yld} : \boxwith{\aeq{\shat{\yld}}}(\Int \to \TUnit),
  \hat{\var{\yld}}\varb{\aeq{\shat{\yld}}}\Int \to \TUnit
  \Hvdash \\
  \hspace{6em}
  \text{$\Letmhl{}{\aex{\shat{\yld}}}$} f = \tr{\namei{app'}{C}}\;\shat{\yld}
  \In f\;(\Modhl{\aeq{\shat{\yld}}}\;\tr{\namei{gen}{C}})
  \Hcolon \boxwith{\aeq{\shat{\yld}}}(\Int \to \TUnit)
  \Hatmode{\shat{\yld}}
\el\end{prog}

As with the encoding of \Cref{sec:overview-encoding-rows}, the
encoding of \SystemC in \Metp{\mcS} helps elucidate exactly how the
capability tracking of \SystemC is entangled with constructs like
block constructions and calls.
Modality introduction and elimination reveal the hidden mechanisms.

We give the full encoding of \SystemC into \Metp{\mcS} in
\Cref{sec:capability-based-effect-systems}.

\subsection{Comparing Rows and Capabilities}
\label{sec:overview-comparison}

As a uniform framework, \Metpt allows us to directly compare how
effect tracking differs in different effect systems without dealing
with the subtleties in their typing and reduction rules.

For instance, let us compare the encoding of function types and
polymorphic types in \Feps with the encoding of block types and box
types in \SystemC.
\begin{prog}\ba{l@{~}r@{~}c@{~}l}
\text{\Feps to \Metp{\emtScp} : } &
  \transl{\earr{A}{B}{E}} &=& \boxwith{\aeq{\tr{E}}}(\tr{A}\to \tr{B}) \\[.5ex]
  & \transl{\forall\evar . A} &=& \forall\evar . \tr{A} \\[.5ex]
\text{\SystemC to \Metp{\mcS} : } &
\transl{\carr{\ol{A}}{\ol{f:T}}{B}} &=&
\forall\ol{\shat{f}} . \boxwith{\aex{\ol{\shat{f}}}}(\ol{\transl{A}}\to \ol{\boxwith{\aeq{\shat{f}}}\transl{T}}\to \transl{B}) \\[.5ex]
&\transl{T \At C} &=& \boxwith{\aeq{\tr{C}}}\tr{T}
\ea\end{prog}
From the encodings we can immediately observe two key differences
of \Feps and \SystemC.
\begin{enumerate}[leftmargin=4ex]
\item
The encoding of function types in \Feps is wrapped with an absolute
modality, whereas the encoding of a block type in \SystemC is wrapped
with a relative modality.
The encoding of box types in \SystemC is wrapped with an absolute
modality.
The different modalities reveal a fundamental difference between the
meanings of functions in \Feps and blocks in \SystemC: functions in
\Feps can only use those effects specified in their types, whereas
blocks in \SystemC can use arbitrary effects from the context unless
they are boxed.
\item
The encoding of block types in \SystemC binds a list of effect
variables and wraps each block argument type with an absolute modality
of the corresponding effect variable, whereas the encoding of a
function type in \Feps is much less involved.
Only the encoding of polymorphic types in \Feps binds effect
variables.
The difference in the treatment of argument types reveals that
capabilities in \SystemC act as an implicit form of parametric
polymorphism, abstracting the capabilities used by each block
variable.
This explains why capability-based effect systems do not require
explicit effect variables in many cases where row-based effect systems
do.
\end{enumerate}

\subsection{Encoding Effect Handlers}
\label{sec:overview-handlers}

We have seen how effectful functions in \Feps and \SystemC are encoded
in \Metpt.
These are the most important parts of our encodings, as most effect
systems track effects by giving different intepretations to functions.
Though all effect systems discussed in this paper support effect
handlers, the same ideas apply equally to traditional effect systems
for languages with only built-in effects.
Nonetheless, the encodings of effect handlers in \Feps and \SystemC
are interesting and reveal fundamental differences between the typing
and semantics of effect handlers in these two calculi.
In this section, we first briefly review what effect handlers are and
then show how effect handlers in \Feps and \SystemC are encoded.

\subsubsection{Effect Handlers in \Metpt}
\label{sec:overview-handler-metp}
Effect handlers allow us to customise how to handle effectful
operations.
For instance, we can write a handler to handle the $\code{yield}$
operation defined in \Cref{sec:met-introduction} by summing up all
yielded integers as follows.
\begin{prog}\bl
\namei{sum}{\Metpt} \defeq
  \Handle\; (\Do\code{yield}\;42;\Do\Yield\;37;0) \With \{
    \code{yield}\;p\;r \mapsto p + r\;\Unit
  \}
\el\end{prog}
The computation $\Do\code{yield}\;42;\Do\code{yield}\;37;0$ is handled by
the handler $\{\code{yield}\;p\;r \mapsto p + r\;\Unit\}$.
The handler consists of one operation clause for the operation
$\code{yield}$.
In this operation clause, the variable $p$ of type $\Int$ is bound to
the parameter of the $\code{yield}$ operation, and the variable $r$ of
type $\TUnit\to\Int$ is bound to its {recursively-handled}
continuation.
(This kind of recursive handling is known as \emph{deep
  handlers}~\citep{KammarLO13} in the literature.)
For instance, when the first $\code{yield}$ operation is handled, $p$
is $42$ and $r$ is the continuation $\lambda y^\TUnit .
\Handle\;(\Do\code{yield}\;37;0)\With \{\code{yield}\;p\;r\mapsto
p+r\;\Unit\}$.
The handler clause adds the yielded integer $p$ to the result of the
continuation $r$, thus returning the sum of all handled operations.
The above program reduces to $79$.
Effect handlers also have a return clause which we omit here, but
describe in \Cref{sec:core-calculus}.

\subsubsection{Encoding Effect Handlers in \Feps}
\label{sec:overview-handler-feps}
\Feps does not use the $\Handle \With\!\!$ syntax.
Instead, a handler in \Feps is defined as a handler value, which is a
function that takes an argument to handle.
Consider the following polymorphic handler for the $\code{yield}$
operation.
\begin{prog}
\bl
\namei{sum}{\Fepss} \defeq
  \Lambda \evar . \force{\Handler}\; \{\code{yield}\;p\;r\mapsto p + r\;\Unit\}\;
  \Hcolon \forall\evar . (\TUnit \to^{\code{yield},\evar} \Int) \to^\evar \Int
\el
\end{prog}
The term \namei{sum}{\Fepss} is polymorphic over other effects
$\evar$ that it does not handle.
The $\Handler$ syntax defines a handler, which is a function that
takes an argument of type $\TUnit \to^{\code{yield},\evar} \Int$,
calls this argument with unit and handles the $\code{yield}$
operation.
The continuation $r$ has type $\earr{\TUnit}{\Int}{\evar}$ as it may
use effects abstracted by $\evar$.
For instance, we can apply \namei{sum}{\Fepss} as follows which
reduces to $79$.
\begin{prog}
\namei{sum}{\Fepss}\;E\;(\lambda x^{\TUnit}. \Do\code{yield}\;42;\Do\code{yield}\;37;0)
\end{prog}

We can easily encode \namei{sum}{\Fepss} in \Metp{\emtScp} as a
polymorphic function whose body uses the $\Handle\With\!\!$
syntax to handle the argument.
The main difficulty is that for the handler clause, the continuation
$r$ should have type $\tr{\earr{\TUnit}{\Int}{\evar}} =
\boxwith{\aeq{\evar}}(\TUnit\to\Int)$ following the translation of
function types in \Cref{sec:overview-encoding-rows}.
However, the typing rule of handlers in \citet{TangWDHLL25} only
allows us to give a function type to $r$ with no modality.
To solve this problem, we introduce \emph{modality-parameterised
handlers}.
In \Metp{\mcT}, the handler syntax is annotated with a modality $\mu$
as $\Handle^\mu\;M\With H$.
The continuation $r$ in the handler clause of $H$ now has type
$\boxwith{\mu}(A\to B)$ for some types $A$ and $B$.
With the modality-parameterised handler, we can translate
\namei{sum}{\Fepss} as follows, omitting the details of the
translation of the handler clause, which we name $H'$.
\begin{prog}
\ba{r@{}c@{}l}
\transl{\namei{sum}{\Fepss}} &\treq&
\Lambda\evar . \Modhl{\aeq{\evar}}~(
\lambda f^{\boxwith{\aeq{\code{yield},\evar}}(\TUnit \to \Int)} .
    \Handle^{\aeq{\evar}}\;
 (\Letmhl{}{\aeq{\code{yield},\evar}} f' = f \In
  f'\;\Unit) \With H')\\
&\Hcolon& \forall\evar.\boxwith{\aeq{\evar}}(
      \boxwith{\aeq{\code{yield},\evar}}(\TUnit \to \Int)
      \to \Int)
\ea
\end{prog}
We eliminate the modality of the argument $f$ before applying and
handling it.
The type translation follows the translation given in
\Cref{sec:overview-encoding-rows}.
We give full details of our modality-parameterised handlers in
\Cref{sec:typing-metp} and formally define the translation of handlers
in \Cref{sec:encoding-rows}.

\subsubsection{Encoding Effect Handlers in \SystemC}
\label{sec:overview-handler-systemc}
\SystemC adopts named handlers.
Instead of using operation labels to identify which operation we want
to invoke and handle, in \SystemC each handler binds a fresh
capability in the scope of the handler and handles the use of this
capability.
For instance, we can define a named handler and use it to handle a
computation as follows.
\begin{prog}
\Hvdash \namei{sum}{C} \defeq \force{\Try}\;\{{\yld}^{\Int\To\TUnit}\To \yld(42);\yld(37);0\}
  \With \{p\;r\mapsto p+r(\Unit)\}
\Hcolon \Int \Hmid C
\end{prog}
Handlers in \SystemC use the $\Try\With\!\!$ syntax.
This handler introduces a capability $\yld$ of type
$\carrsingle{\Int}{}{\TUnit}$ in the scope between $\Try$ and
$\!\!\With\!\!$.
We use the capability $\yld$ to yield integers $42$ and $37$.
These two uses of $y$ are handled by the handler, whose operation
clause is similar to what we have seen before, except it uses a
capability in place of an operation label.

The semantics of named handlers in \SystemC differs from that of the
standard effect handlers of \citet{PlotkinP13}.
Named handlers have a generative semantics~\citep{BiernackiPPS20}
which dynamically generates a fresh runtime label for each capability
introduced by a handler.
Dynamic generation guarantees the uniqueness of runtime labels, which
ensures that all uses of a capability must be handled by the handler
that introduces the capability.

To encode the named handlers of \SystemC into \Metp{\mcT}, we need to
resolve this semantic gap.
Adding named handlers to \Metpt would work but is rather heavyweight.
We observe that the essence of named handlers is actually a way to
dynamically generate labels.
We introduce \emph{local labels} to \Metp{\mcT}, which decouple
dynamic generation of labels from named handlers.
This extension is inspired by the local effects of
\citet{BiernackiPPS19}, dynamic labels of \citet{VilhenaP23}, and
fresh labels of the \Links
language~\citep{links-generative-labels}.
The syntax $\Localeffect{\ell:A\sto B} M$ introduces a local label
in the scope of $M$.
The type system ensures the local label $\ell$ cannot escape from $M$.
The semantics generates a fresh label to replace the local label
$\ell$.
We provide the details in \Cref{sec:core-calculus}.
With local labels, we can encode \namei{sum}{C} as follows, omitting
the handler $H'$, which contains an operation clause for $\ell_\yld$
translated from the handler of \namei{sum}{C}.
\begin{prog}
\hspace{-.5em}
\ba{l@{}r@{}c@{}l}
\gray{\vdash}\;& \transl{\namei{sum}{C}} &\treq&
\Localeffect{\ell_\yld : \Int\sto\TUnit}
\Handle^{\aeq{\tr{C}}}\; \\[.5ex]
 &\span\span \quad (\lambda \force{\yld}^{\boxwith{\aeq{\force{\ell_\yld}}}(\Int\to\TUnit)} .
    \Letmhl{}{\aeq{\ell_\yld}} \hat{\yld} = \yld \In
    \hat{\yld}\;42 ; \hat{\yld}\;37 ; 0
  )\;(\Modhl{\aeq{\ell_\yld}}\;(\lambda x^\Int . \Do \ell_\yld\;x))
  \With H'
   \;\gray{:}\; \Int \;\gray{@}\; \dblue{\tr{C}}
\ea
\end{prog}

We introduce a local label $\ell_\yld$ for the handler.
We use the term $\Modhl{\aeq{\ell_\yld}}\;(\lambda x^\Int . \Do
\ell_\yld\;x)$ which invokes the operation $\ell_\yld$ to simulate the
capability introduced by the named handler in \namei{sum}{C}.
The translation of the handled computation binds this function to
$\yld$, eliminates the modality of $\yld$ to $\hat{\yld}$, and uses
$\hat{\yld}$ to yield integers $42$ and $37$.
As in the encoding of effect handlers in \Feps, we also use our
modality-parameterised handlers here and annotate $\Handle$ with the modality
$\aeq{\tr{C}}$.

Our translation $\tr{\namei{sum}{C}}$ is simplified for clarity; it is
actually the result of reducing the full translation of \namei{sum}{C}
by a few steps.
We give the full translation in \Cref{sec:encoding-caps}.

\subsubsection{Comparing Encodings of Effect Handlers}
\label{sec:overview-handler-comparison}

Our encodings of \Feps and \SystemC effect handlers elucidate how
effect handlers differ in these two effect systems.
\begin{enumerate}[leftmargin=4ex]
\item
The \SystemC encoding requires local labels, whereas the \Feps
encoding does not, which reveals the syntactic difference that
capabilities in \SystemC have scopes whereas operation labels in \Feps
do not, and the semantic difference that \SystemC generates fresh
runtime labels for effect handlers, whereas \Feps does not.
\item
The \Feps encoding performs operations directly, whereas the
\SystemC encoding wraps operation invocations into a function (such as
the term $\Modhl{\aeq{\ell_\yld}}(\lambda x^\Int . \Do \ell_\yld\;x)$
in $\tr{\namei{sum}{C}}$) and passes this function to the handled
computation.
This difference shows how in a capability-based effect system such as
\SystemC operations are not directly invoked via their labels but are
instead invoked and passed around as blocks explicitly at the term
level.
\end{enumerate}

\subsection{More Encodings}
\label{sec:overview-more-encodings}

The encodings of \Feps and \SystemC illustrate the core idea of using
modalities to encode and compare effect systems with different
foundations in \Metp{\mcT}.
However, \Metp{\mcT} can be used for much more than encoding
these two effect systems.
In \Cref{sec:more-encodings}, we will discuss two more representative
encodings of effect systems into \Metp{\mcT}, including
\begin{itemize}%
  \item \SystemXi~\citep{BrachthauserSO20}, an early core calculus for \Effekt
  based on capabilities, and
  \item \Fepssn~\citep{XieCIL22}, a core calculus formalising
  scope-safe named handlers of \Koka.
\end{itemize}
These results further demonstrate the expressiveness of \Metp{\mcT} as
a general framework to encode, compare, and analyse effect systems.
We further discuss practical language design insights arising from our
encodings in \Cref{sec:discussion}.

\section{The Core Calculus \Metp{\mcT}}
\label{sec:core-calculus}

\Metp{\mcT} is a System F-style call-by-value core calculus with modal
effect types parameterised by an effect structure $\mcT$.
In addition to the effect structure, \Metp{\mcT} also extends \Met
with local labels and modality-parameterised handlers.
We aim to be self-contained about modal effect types in this paper and
refer to \citet{TangWDHLL25} for a more complete introduction.

\subsection{Syntax}

The syntax of \Metp{\mcT} is as follows.
We highlight syntax relevant to modal effect types and our extensions
of local labels and modality-parameterised handlers in grey.
{
\begin{prog}\small
\bs
  \slab{Types}    &A,B  &::= & \TUnit \mid A\to B \mid
                      \hl{\boxwith{\mu} A} \\
                    & & \mid &\alpha\mid {\forall \alpha^K.A} \\
  \slab{Modalities}     &\mu,\nu &::= & \hl{\aeq{E}} \mid \hl{\aex{D}} \\
  \slab{Extensions} &D &::= & \cdot \mid \ell,D \mid \evar,D \\
  \slab{Effect Contexts}\hspace{-1.5ex} &E,F &::= & \cdot \mid \ell,E \mid \evar,E \\
  \slab{Kinds}          &K &::= &  {\Pure} \mid {\Any} \mid \Effect \\
  \slab{Contexts}      &\Gamma &::=& \cdot  \mid \Gamma, \alpha:K
                                            \mid \Gamma, \hl{\lockwith{{\mu}_{E}}} \\
                                      & &\mid& \Gamma, \hl{x\varb{\mu_E}{A}}
                                            \mid \Gamma, {\ell : A\sto B}
                                            \\
  \slab{Label Contexts}\hspace{-1ex} &\Sigma &::= & \cdot \mid \Sigma,{\ell : A\sto B} \\
\es
\hfill
\hspace{2em}
\bs
  \slab{Terms}\hspace{-1em}  &M,N  &::= & \Unit \mid x \mid \lambda x^A.M \mid M\,N \\
                              & &\mid& \Lambda \alpha^K.V \mid M\;A \mid \hl{\Mod_\mu\,V} \\
                              & &\mid& \hl{\Letm{\nu}{\mu} x = V\In M} \\
                        &     &\mid& \Do\ell\; M \mid \xl{\Localeffect{\ell:A\sto B} M}  \\
                        &     &\mid& \xl{\Handle^\mu\;M\With H} \\
  \slab{Values}\hspace{-1em}  &V,W  &::= & \Unit \mid x \mid \lambda x^A.M \mid \Lambda\alpha^K.V
                                    \mid {\Mod_\mu\, V} \\
                                    & &\mid&
                                    V\;A\mid
                                     \Letm{\nu}{\mu} x = V\In W \\
  \slab{Handlers}\hspace{-1em}      &H    &::= & \{ \Ret x \mapsto N, \ell\; p \; r \mapsto M \} \\
\es
\end{prog}
}

We have two kinds $\Pure$ and $\Any$ for value types and one kind
$\Effect$ for extensions and effect contexts.
By convention, we usually write $\alpha$ for type variables of values
and $\evar$ for effect variables.
We omit kinds when obvious.
We let $A$ range over both value types $A$ and effect contexts $E$,
and let $\alpha$ range over type variables for them in type
abstraction $\Lambda\alpha^K.V$ and type application $M\;A$.

Unlike \citet{TangWDHLL25}, we omit masking, as it is not used by our
encodings.
We discuss future extensions to \Metp{\mcT}, including masking, in
\Cref{sec:future-extensions}.

For simplicity, we assume that each handler only handles one
operation, and fix a global context $\Sigma$ which associates each
global operation label with its type.
An entry $\ell : A \sto B$ indicates that the operation $\ell$ takes
an argument of type $A$ and returns a value of type $B$.
We also support local labels which are introduced by
$\Localeffect{\ell:A\sto B} M$ and maintained in the context $\Gamma$.
We do not distinguish between local and global labels syntactically.

Values include type application and modality elimination whose
subterms are restricted to be values, following the notion of
\emph{complex values} in call-by-push-value~\citep{Levy2004}.
Such complex values are convenient as we adopt a value
restriction~\citep{Wright95} for type abstraction and modality
introduction.

\subsection{Effect Structures}
\label{sec:effect-mode-theory}

An effect structure defines the structure of effect collections, that is,
extensions and effect contexts in \Metpt.
Extensions $D$ and effect contexts $E$ are both syntactically defined
as lists of labels and effect variables.
We overload commas for list concatenation, e.g., $D,E$ and $E,F$ are
both list concatenation.
As usual, list concatenation is associative but not commutative.
The kinding, equivalence, and subtyping (or subeffecting) relations
for them are determined by an effect structure $\mcT$.

\begin{definition}[Effect structure]
An effect structure $\mcT$ is a tuple $\langle :, \equiv \rangle$
of two relations.
\begin{itemize}
  \item $\Gamma \vdash D : \Effect$ is a kinding relation which
  defines well-formed extensions and is preserved by concatenation
  $D,D'$. That is, if $\Gamma \vdash D : \Effect$ and $\Gamma \vdash
  D' : \Effect$, then $\Gamma \vdash D,D' : \Effect$.
  \item $\Gamma \vdash D \equiv D'$ is an equivalence relation for
  well-formed extensions. %
\end{itemize}
\end{definition}

Our definition of an effect structure $\mcT$ is minimal and only includes
the definitions of kinding and equivalence relations for extensions
$D$. We can naturally derive the kinding relation $\Gamma\vdash
E:\Effect$, equivalence relation $\Gamma\vdash E\equiv E'$, and
subeffecting relation $\Gamma\vdash E\subtype E'$ for effect contexts
as follows.
{
\begin{prog} \ba{l@{\qquad\quad }l@{\qquad\quad }l@{\qquad\quad }l}
\inferrule*
{ }
{\Gamma \vdash \cdot : \Effect }

&
\inferrule*
{\Gamma \ni \evar : \Effect}
{\Gamma \vdash \evar : \Effect}

&
\inferrule*
{
  \Gamma \vdash D : \Effect \\
  \Gamma \vdash E : \Effect \\
}
{\Gamma \vdash D,E : \Effect}
\\[.5ex]

\inferrule*
{ }
{\Gamma\vdash \cdot \equiv \cdot}

&
\inferrule*
{ }
{\Gamma\vdash \evar \equiv \evar }
\qquad

&
\inferrule*
{
  \Gamma\vdash D_1 \equiv D_2 \\
  \Gamma\vdash E_1 \equiv E_2
}
{\Gamma\vdash D_1,E_1 \equiv D_2,E_2}

&
\inferrule*
{
  \Gamma \vdash E,E' \equiv F\\
}
{\Gamma\vdash E \subtype F}
\el\end{prog}
}

The kinding and equivalence relations for effect contexts are defined
inductively.
The subeffecting relation is more interesting.
We have $E\subtype F$ if there exists an effect context $E'$ such that
$E,E'$ is well-formed and $E,E'\equiv F$.
It is easy to verify that this subeffecting relation is a preorder.
We often write $:_\mcT$, $\subtype_\mcT$, and $\equiv_\mcT$ to denote
which specific effect structure we refer to.
We sometimes omit the context $\Gamma$ for the equivalence and
subeffecting for brevity.

We give three examples of effect structures, among which $\emtScp$
and $\mcS$ are used for the encoding of \Feps and \SystemC in
\Cref{sec:encoding-rows,sec:encoding-caps}, respectively.

\begin{definition}[Simple Rows]
  $\emtSimp = \langle :_\emtSimp,\equiv_\emtSimp\rangle$ defines
  effect collections as simple rows~\citep{rose} of operation labels.
  Well-formed extensions consist of distinct labels without any effect
  variable.
  $D\equiv D'$ if $D$ is identical to $D'$ modulo reordering of
  labels.
\end{definition}

\begin{definition}[Scoped Rows]
  $\emtScp = \langle :_\emtScp,\equiv_\emtScp\rangle$ defines
  effect collections as scoped rows~\citep{Leijen05} of operation
  labels.
  Well-formed extensions comprise potentially duplicated labels
  without any effect variable.
  $D\equiv D'$ if $D$ is identical to $D'$ modulo reordering of
  distinct labels.
\end{definition}

\begin{definition}[Sets]
  $\mcS = \langle :_\mcS, \equiv_\mcS\rangle$ defines effect
  collections as sets.
  Well-formed extensions are sets of labels and effect variables.
  The equivalence relation is set equivalence.
\end{definition}

Full formal definitions of these effect structures are given in
\Cref{app:effect-mode-theories}.
The effect structure $\emtScp$ corresponds to the treatment of effect
collections as scoped rows used in \Met, modulo the fact that \Met has
presence types for labels in effect contexts, whereas we choose not
to for simplicity.
We discuss extending \Metp{\mcT} with presence types and richer effect
kinds in \Cref{sec:future-extensions}.

Following \citet{YoshiokaSI24}, an effect structure that intuitively
characterises the notion of a collection of effects should satisfy the
following validity conditions.
\begin{definition}[Validity Conditions]
  \label{def:sanity-conditions}
  Validity conditions for an effect structure $\mcT$ are
  \begin{enumerate}
    \item if $E\subtype_\mcT \cdot$ then $E = \cdot$, and
    \item if $\ell\subtype_\mcT \ell',E$ and $\ell\neq\ell'$ then
    $\ell\subtype_\mcT E$.
  \end{enumerate}
\end{definition}
The validity conditions together ensure that if a label $\ell$ is a
subtype of an effect context $E$, then it must syntactically appear in
the effect context $E$.
The first condition prevents us from claiming that some label is
contained in the empty effect context.
The second condition prevents us from identifying two syntactically
different label as the same one.
All effect structures given above satisfy the validity conditions.
Our type soundness and effect safety theorems in
\Cref{sec:type-soundness} are parameterised by any effect structure
satisfying the validity conditions.

\subsection{Modalities}
\label{sec:modalities}

Modalities manipulate effect contexts as follows. %
{
\vspace{-.2\baselineskip}
\begin{prog}
  \act{\aeq{E}}{F} ~=~ E \qquad\qquad\qquad \act{\aex{D}}{F} ~=~ D,F
\vspace{-.2\baselineskip}
\end{prog}
}
The absolute modality $\aeq{E}$ completely replaces the effect context
$F$ with $E$.
The extension modality $\aex{D}$ extends the effect context $F$ with $D$.
Following \Met~\citep{TangWDHLL25}, we write $\mu_F$ as a meta-level
notation for the pair of modality $\mu$ and effect context $F$ where
$F$ is the effect context that $\mu$ manipulates.

\paragraph{Modality Composition}

We define the composition of modalities as follows.
\begin{prog}\bl
\mu \circ {\aeq{E}}  =  \aeq{E}
\qquad
\qquad
\aeq{E} \circ {\aex{D}}  =  \aeq{D,E}
\qquad
\qquad
\aex{D_1} \circ \aex{D_2}  =  \aex{D_2,D_1}
\el\end{prog}
Composition is from left to right, for consistency with \Met.
First, an absolute modality fully determines the new effect context
$E$ no matter what $\mu$ does before.
Second, setting the effect context to $E$ followed by extending $E$
with ${D}$ is equivalent to directly setting the effect context to
$D,E$.
Third, relative modalities can be composed into one by combining
the extensions.
Composition is well-defined as we have $\act{(\mu\circ\nu)}{E} =
\act{\nu}{\act{\mu}{E}}$.
We also have associativity $(\mu\circ\nu)\circ\xi =
\mu\circ(\nu\circ\xi)$ and identity $\aid \circ \mu = \mu \circ \aid =
\mu$. %
All of these properties are independent of the effect structure $\mcT$.

\paragraph{Modality Transformation}
We define a modality transformation judgement, which determines the
coercion of modalities, controlling the accessibility of variables as
mentioned in \Cref{sec:overview-absolute} where we disallow the usage
of the variable $f:\Int \to\TUnit$.
Given a variable binding $f\varb{\mu_F} A$ (which means $f$ is
introduced by eliminating the modality $\mu$ of some value of type
$\boxwith{\mu}A$ at effect context $F$), we can access it after a lock
$\lockwith{\nu_F}$ if the modality
transformation relation $\Gamma\vdash \mu\To\nu\atmode{F}$ holds.
The modality transformation judgement is defined as follows.%
\begin{prog}\small\bl
  \inferrule*[left=\mtylab{Abs}]
  {
    \Gamma\vdash E\subtype \mu(F) \\
  }
  {\Gamma\vdash \aeq{E}\To {\mu} \atmode{F}}
  \qquad\qquad\qquad

  \inferrule*[left=\mtylab{Extend}]
  {
    \Gamma \vdash D_1,F \subtype D_2,F \text{ for all } E\subtype F
  }
  {\Gamma\vdash \aex{D_1} \To \aex{D_2} \atmode{E}}
\el\end{prog}

Both rules make sure that we do not lose any effects after
transformation.
Rule \mtylab{Abs} allows us to transform an absolute modality
$\aeq{E}$ to any other modality $\mu$ as long as $E\subtype \mu(F)$.
Rule \mtylab{Extend} allows us to transform an extension modality
$\aex{D_1}$ to another extension modality $\aex{D_2}$ as long as for
any effect context $F$ larger than $E$, we have $D_1,F\subtype D_2,F$.
We need to quantify over all effect contexts $F$ which are larger than
the ambient effect context $E$ because the new effect context that a
relative modality gives us depends on the ambient effect context.
For instance, consider the following judgement which coerces the
modality $\aex{D_1}$ of $V$ to $\aex{D_2}$.
\begin{prog}\bl
\Gamma \vdash \Letm{}{\aex{D_1}} x = V \In \Mod_{\aex{D_2}}\;x
: \boxwith{\aex{D_2}} (\Int\to\Int) \atmode{E}
\el\end{prog}
In its derivation tree we need the transformation relation
$\aex{D_1}\To\aex{D_2} \atmode{E}$.
To preserve the judgement after upcasting $E$ to a larger $F$, the
transformation requires $D_1,F \subtype D_2,F$ for any $E \subtype F$.

Our \mtylab{Extend} rule is suitable for any effect structure, while
the corresponding rule \mtylab{Upcast} in \citet{TangWDHLL25} is
specific to the treatment of effect collections as scoped rows in
\Met.
Given a specific effect structure, we can usually find an
easier-to-compute representation of \mtylab{Extend} without universal
quantification (as is the case for the \mtylab{Upcast} rule in \Met).

\subsection{Kinds and Contexts}

The kinding relations for extensions and effect contexts are provided
by the effect structure $\mcT$ in
\Cref{sec:effect-mode-theory}.
For value types, we have two kinds where $\Pure$ is a subkind of
$\Any$.
A type has kind $\Pure$ if all function types appearing as syntactic
subterms of the type are wrapped in absolute modalities.
For example, $(\TUnit \to \TUnit)\to\TUnit$ does not have kind $\Pure$
whereas $\boxwith{\aeq{}}((\TUnit\to\TUnit)\to\TUnit)$ does.
Intuitively, values whose types have kind $\Pure$ do not depend on the
ambient effect context.
For any operation $\ell : A\sto B$, types $A$ and $B$ should have kind
$\Pure$ to avoid effect leakage following \citet{TangWDHLL25}.
The kinding and type equivalence rules of \Metpt are given in
\Cref{app:full-spec}.

Contexts are ordered.
We write $\Gamma\atmode{E}$ when context $\Gamma$ is well-formed at
effect context $E$, that is, the types of the variables are
well-kinded, and the variables and locks are compatible with $E$.
For instance,
the following context is well-formed at effect context $E$.
\begin{prog}
  {x\varb{\mu_F}{A_1},y\varb{\nu_F}{A_2}},~
  {\lockwith{\aeq{E}_F}},~
  {z\varb{\xi_E}{A_3}, w : A_4} \atmode{E}
\end{prog}
Let us read from right to left.
Variable $w$ is at effect context $E$ (it is technically tagged with
an identity modality $\aid_E$ which is omitted).
Variable $z$ is tagged with modality $\xi_E$, which means it is not at
effect context $E$ but actually at effect context $\xi(E)$.
Lock $\lockwith{\aeq{E}_F}$ changes the effect context to $E$ from
$F$.
Variables $y$ and $x$ are at effect contexts $\nu(F)$ and $\mu(F)$,
respectively.
Each modality in the context carries an index of the effect context it
manipulates, making switching of effect contexts explicit.
We frequently omit the index when it is clear what it must be.
Formal definitions of kinding and context well-formedness rules are in
\Cref{app:rules-metp}.
We define $\locks{-}$ to compose all the modalities on
the locks in a context.
\begin{prog}\bl
\locks{\cdot} = {\aid} \qquad
\locks{\Gamma,\lockwith{\mind{\mu}{E}}} = \locks{\Gamma}\circ {\mu} \qquad
\locks{\Gamma,x\varb{\mu_E}{A}} = \locks{\Gamma} \\
\el\end{prog}
We identify contexts up to the following two equations.
\begin{prog}\bl
\Gamma,\lockwith{\aid_F}, \Gamma' \atmode{E} = \Gamma, \Gamma' \atmode{E} \qquad\qquad
\Gamma,\lockwith{{\mu}_F},\lockwith{{\nu}_{F'}}, \Gamma' \atmode{E}
  = \Gamma,\lockwith{({\mu}\circ{\nu})_{F}}, \Gamma'\atmode{E} \\
\el\end{prog}

\subsection{Typing}
\label{sec:typing-metp}

\Cref{fig:typing-metp} gives the typing rules for \Metp{\mcT}.
As before, we highlight rules relevant to modal effect types and our
extensions in grey.
The typing judgement $\typm{\Gamma}{M:A}{E}$ means that the term $M$
has type $A$ under context $\Gamma$ and effect context $E$ with
well-formedness condition $\Gamma\atmode{E}$.

\begin{figure}[tb]
{\small
\raggedright
\boxed{\Gamma\vdash (\mu,A)\To\nu \atmode{F}}
\hfill
\vspace{-1.5\baselineskip}
\begin{mathpar}
\inferrule*
{
  \Gamma \vdash A:\Pure
}
{\Gamma\vdash (\mu,A)\To\nu \atmode{F}}

\inferrule*
{
  \Gamma \vdash \mu \To \nu \atmode{F}
}
{\Gamma\vdash (\mu,A)\To\nu \atmode{F}}
\vspace{-.5\baselineskip}
\end{mathpar}
\raggedright
\boxed{\typm{\Gamma}{M : A}{E}\vphantom{\mu}}
\hfill
\vspace{-\baselineskip}
\begin{mathparshrink}
\hl{
\inferrule*[Lab=\tylab{Var}]
{
  {\Gamma \vdash (\mu,A) \To \locks{\Gamma'} \atmode{F} }
}
{\typm{\Gamma,x\varb{\mu_F}{A},\Gamma'}{x:A}{E}}
}

\hl{
\inferrule*[Lab=\tylab{Mod}]
{
  \typm{\Gamma,\lockwith{\mind{\mu}{F}}}{V:A}{\mu(F)}
}
{\typm{\Gamma}{\Mod_\mu\,V : \boxwith{\mu} A}{F}}
}

\hl{
\inferrule*[Lab=\tylab{Letmod}]
{
  \typm{\Gamma,\lockwith{\nu_F}}{V : \boxwith{\mu} A}{\nu(F)} \\\\
  \typm{\Gamma,x\varb{(\nu\circ\mu)_F}{A}}{M:B}{F}
}
{\typm{\Gamma}{\Letm{\nu}{\mu} x = V \In M : B}{F}}
}

\inferrule*[Lab=\tylab{Abs}]
{
  \typm{\Gamma, x : A}{M : B}{E}
}
{\typm{\Gamma}{\lambda x^A .M : A \to B}{E}}
\quad\
\inferrule*[Lab=\tylab{App}]
{
  \typm{\Gamma}{M : A \to B}{E} \\\\
  \typm{\Gamma}{N : A}{E}
}
{\typm{\Gamma}{M\; N: B}{E}}
\quad\
\inferrule*[Lab=\tylab{TAbs}]
{
  \typm{\Gamma, \alpha:K}{V : A}{E}
}
{\typm{\Gamma}{\Lambda \alpha^K. V : \forall \alpha^K . A}{E}}
\quad\
\inferrule*[Lab=\tylab{TApp}]
{
  \typm{\Gamma}{M : \forall \alpha^K.A}{E} \\\\
  \Gamma \vdash B : K
}
{\typm{\Gamma}{M\; B: A[B/\alpha]}{E}}

\inferrule*[Lab=\tylab{Unit}]
{ }
{\typm{\Gamma}{\Unit : \TUnit}{E}}

{
\inferrule*[Lab=\tylab{Do}]
{
  \Sigma,\Gamma \ni \ell : A\sto B \\
  \typm{\Gamma}{N : A}{\ell, E} \\
}
{\typm{\Gamma}{\Do\ell\; N : B}{\ell, E}}
}

\xl{
\inferrule*[Lab=\tylab{LocalEffect}]
{
  \typm{\Gamma,\ell:A\sto B}{M : A'}{E} \\
}
{\typm{\Gamma}{\Localeffect{\ell:A\sto B} M : A'}{E}}
}

\xl{
\inferrule*[Lab=\tylab{Handle}]
{
  \mu(F) = E \\
  \Gamma \vdash \mu \To \aid \atmode{F} \\
  \Gamma \vdash \mu \To \mu\circ\mu \atmode{F} \\
  \typm{\Gamma, \lockwith{\mu_F}, \lockwith{\aex{\ell}_{E}}}{M : A}{\ell,E} \\\\
  \Sigma,\Gamma \ni \ell : A'\sto B' \quad
  \typm{\Gamma, \lockwith{\mu_F}, x : \boxwith{(\mu\circ\aex{\ell})} A}{N : B}{E} \quad
  \typm{\Gamma, \lockwith{\mu_F}, p : A', r: \boxwith{\mu}(B' \to B)}{N' : B}{E} \\
}
{\typm{\Gamma}{\Handle^\mu\;M\With
  \{\Ret x \mapsto N, \ell\;p\;r \mapsto N' \}
 : B}{F}}
}
\end{mathparshrink}
}
\vspace{-.2\baselineskip}
\caption{Typing rules and auxiliary rules of \Metp{\mcT}.}
\label{fig:typing-metp}
\vspace{-.1\baselineskip}
\end{figure}

\paragraph{Modality Introduction and Elimination}

Rule \tylab{Mod} introduces a modality $\mu$ to the conclusion, puts a
lock into the context of the premise, and changes the effect context.
Rule \tylab{Letmod} eliminates a modality $\mu$ and moves it to the
variable binding.
We have seen examples that rely on these rules in
\Cref{sec:met-introduction}.
There is another modality $\nu$ in \tylab{Letmod} which is needed for
technical reasons to support sequential elimination.
For instance, given a variable $x : \boxwith{\nu}\boxwith{\mu} A$ with
two modalities, to eliminate both $\nu$ and $\mu$, we can first
eliminate $\nu$ to $y\varb{\nu}{\boxwith{\mu} A}$ and then to
$z\varb{\nu\circ\mu}{A}$ as follows.
\begin{prog}
\Letm{}{\nu} y = x \In
\Letm{\nu}{\mu} z = y \In M
\end{prog}
We restrict $\Mod_\mu$ and $\Letm{\nu}{\mu}\!\!$ to values to avoid
effect leakage, as in \Met~\citep{TangWDHLL25,LorenzenWDEL24}.
Otherwise, for example, if we were to allow a computation
$\Mod_{\aeq{\Yield}}\; (\Do \Yield\; 42)$, this term would be
well-typed in the empty effect context but get stuck as $\Yield$ is
not handled (note that we do not want $\Mod$ to suspend computations
as it could be confusing to programmers).

\paragraph{Accessing Variables}

Locks control the accessibility of variables as we have shown in
\Cref{sec:met-introduction}.
Rule \tylab{Var} uses the auxiliary judgement
$\Gamma\vdash(\mu,A)\To\locks{\Gamma'}\atmode{F}$ (also defined in
\Cref{fig:typing-metp}) to check whether we can access a variable
$x\varb{\mu_F}{A}$ given all locks in $\Gamma'$.
When $A$ has kind $\Pure$, we can always use $x$ as it does not depend
on the effect context.
Otherwise we need to make sure the coercion from $\mu$ to
$\locks{\Gamma'}$ is safe by checking the modality transformation
relation $\Gamma\vdash\mu\To\locks{\Gamma'}\atmode{F}$ where
$\locks{\Gamma'}$ composes the modalities on locks in $\Gamma'$.
We have seen an example in \Cref{sec:overview-absolute} that the
variable $f:\Int\to\TUnit$ cannot be used while
$f\varb{\aeq{\Yield}}{\Int\to\TUnit}$ can.
As another example, $x\varb{\aex{\ell}}{\TUnit\to\TUnit},
\lockwith{\aeq{\ell'}}\vdash x:\TUnit\to\TUnit \atmode{\ell'}$ is
ill-typed since we cannot transform the modality $\aex{\ell}$ to
$\aeq{\ell'}$.
It would be well-typed if $x$ had type $\Int$.

\paragraph{Local Labels}
Rule \tylab{LocalEffect} binds a fresh local label $\ell$ with type
signature $A\sto B$ (we adopt the Barendregt convention for local
labels.)
Well-formedness of type $A'$ and effect context $E$ under $\Gamma$
ensures that $\ell$ cannot appear in $A'$ or $E$.
Rule \tylab{Do} may use any label from $\Sigma$ and $\Gamma$.
The operational semantics (\Cref{sec:semantics-metp}) generates
runtime labels to substitute local labels.

\paragraph{Modality-Parameterised Handlers}
Rule \tylab{Handle} defines a handler and uses it to handle a
computation $M$.
Let us first ignore all occurrences of the modality $\mu$.
A handler of operation $\ell$ extends the effect context with $\ell$
as indicated by the lock $\lockwith{\aex{\ell}_E}$ in the typing
judgement of $M$.
The return value of $M$ is bound to the variable $x$ in the return
clause $\Ret x \mapsto N$.
The type of $x$ also has the modality $\aex{\ell}$ since $x$ may use
the operation $\ell$, e.g., when $M$ returns a function $\lambda
x.\Do\ell\;x$.

We generalise the handlers of~\citet{TangWDHLL25} to be parameterised
by a modality $\mu$.
The modality $\mu$ transforms the effect context $F$ to $\mu(F) = E$
for the whole term as witnessed by the addition of the lock
$\lockwith{\mu_F}$ to the context of each premise.
Since both the handled computation and the handler clauses are
well-typed under the lock $\lockwith{\mu_F}$, we can wrap the
continuation $r$, which captures the handled computation and the
handler, into the modality $\mu$.
The return value $x$ is also wrapped in the modality $\mu$ as it is
returned from $M$ whose context contains the lock $\lockwith{\mu_F}$.
This is in contrast to the handler rule of \citet{TangWDHLL25}, as
shown below, which just gives $r$ the function type $B' \to B$.
\begin{prog}
{\fontsize{9}{12}\selectfont
\inferrule*%
{
  \Sigma,\Gamma \ni \ell : A'\sto B' \quad\
  \typm{\Gamma, \lockwith{{\aex{\ell}}_{E}}}{M : A}{\ell,E} \quad\
  \typm{\Gamma, x : \boxwith{\aex{\ell}} A}{N : B}{E} \quad\
  \typm{\Gamma, p : A', r: B' \to B}{N' : B}{E}
}
{\typm{\Gamma}{\Handle\;M\With
  \{\Ret x \mapsto N, \ell\;p\;r \mapsto N' \}
 : B}{E}}
}
\end{prog}
To recover the original handler construct of \citet{TangWDHLL25}, we
just need to instantiate the modality $\mu$ to the identity modality
$\aid$ as shown by the following syntactic sugar.
\begin{prog}\ba{rcl}
& &\Handle\;M \With \{\Ret x\mapsto N,\ell\;p\;r\mapsto N'\} \\
&\doteq& \Handle^{\aid}\;M \With \{\Ret x\mapsto N,\ell\;p\;r\mapsto \Letm{}{\aid} r = r \In N'\} \\
\ea\end{prog}

Having a modality $\mu$ for the continuation $r$ allows us to have
more fine-grained control over effect tracking for the continuation.
As discussed in \Cref{sec:overview-handlers}, the extra expressiveness
provided by this rule is especially useful for a unified framework to
encode other effect systems with support for effect handlers, as
different encodings typically require translating a function type into
a type with some modalities.
We give an example of an effect handler annotated with the empty
absolute modality $\aeq{}$ in \Metpt based on the handler
\namei{sum}{\Metpt} in \Cref{sec:overview-handler-metp}.
\begin{prog}\bl
  \Handle^{\aeq{}}\; (\Do\code{yield}\;42;\Do\Yield\;37;0) \With \{
    \bl\Ret x \mapsto \Letm{}{\aeq{\code{yield}}} x' = x \In x', \\
    \code{yield}\;p\;r \mapsto \Letm{}{\aeq{}} r' = r \In  p + r'\;\Unit
  \}\el
\el\end{prog}
As a result of the annotation $\aeq{}$, the continuation $r$ has type
$\boxwith{\aeq{}}(\TUnit\to\Int)$ instead of $\TUnit\to\Int$.
In the return clause we eliminate the modality
$\aeq{}\circ\aex{\code{yield}} = \aeq{\code{yield}}$ of $x$.
In contrast, the omitted return clause of \namei{sum}{\Metpt} is $\Ret
x \mapsto \Letm{}{\aex{\code{yield}}} x' = x \In x'$.

The new handler rule requires the modality $\mu$ to have a comonadic
structure as specified by the conditions $\Gamma \vdash \mu\To\aid
\atmode{F}$ and $\Gamma \vdash \mu\To\mu\circ\mu \atmode{F}$.
These conditions are important because semantically a handler for
operation $\ell$ may not be used (when $\ell$ is not invoked) or be
used multiple times (when $\ell$ is invoked multiple times).
Intuitively, each use of the handler consumes one modality $\mu$.
The condition $\mu\To\aid \atmode{F}$ makes sure that when the handler
is not used, we can transform away the modality $\mu$ at effect
context $F$.
The condition $\mu\To\mu\circ\mu \atmode{F}$ makes sure that when the
handler is used multiple times, we can duplicate the modality $\mu$
each time the handlers is used.
For example, the identity modality $\aid$ trivially satisfies the
comonadic structure, and the absolute modality $\aeq{E}$ satisfies the
comonadic structure at $F$ with $E\subtype F$.

\subsection{Operational Semantics}
\label{sec:semantics-metp}

We adopt the generative semantics of \citet{BiernackiPPS20} for local
labels.
Each local label $\ell$ introduced by $\Localeffect{\ell:A\sto B} M$
is replaced by a fresh label generated at runtime.
We manage these labels in a context defined as $\Instctx ::= \cdot
\mid \Instctx, \ell : A\sto B$.
We do not syntactically distinguish runtime generated labels from
static labels; runtime labels are tracked in $\Instctx$.
We define value normal forms $U$ which cannot reduce further.
The definitions for all new syntax and the operational semantics are
given in~\Cref{fig:semantics-metp}.
The reduction relation has the form $M\mid\Instctx\reducesto
N\mid\Instctx'$.
We omit $\Instctx$ when it is unchanged. Only \semlab{Gen} extends
$\Instctx$.
We do not restrict $M$ and $N$ to be closed terms.
All judgements defined previously are also straightforwardly extended
with $\Instctx$.
For instance, typing judgements are of form
$\typm{\Instctx\mid\Gamma}{M:A}{E}$ for runtime terms.

\begin{figure}[tb]
\small
{
\begin{syntax}
\slab{Value normal forms}         &U& ::= &
x
\mid \lambda x^A.M
\mid \Lambda \alpha^K.V
\mid \Mod_\mu\;U
\\
\slab{Evaluation Contexts} &  \EC &::= & [~]
  \mid \EC\; N \mid U\;\EC \mid \EC\; A \mid \Mod_\mu\;\EC
  \mid \Letm{\nu}{\mu} x = \EC \In M
   \\
  & &\mid & \Do\ell\;\EC %
  \mid \Handle^\mu\;\EC\With H
  \\
\end{syntax}
}

\begin{nreductions}
\semlab{App}   & (\lambda x^A.M)\,U &\reducesto& M[U/x] \\
\semlab{TApp}  & (\Lambda \alpha^K.U)\,A &\reducesto& U[A/\alpha] \\
\semlab{Letmod} &
  \Letm{\nu}{\mu} x = \Mod_\mu\,U \In M
  &\reducesto& M[U/x] \\
\semlab{Gen} &\Localeffect{\ell:A\sto B} M \mid \Instctx
  &\reducesto &
  M[\ell'/\ell] \mid \Instctx, \ell':A\sto B
    \quad\hfill\text{where } \ell'\text{ fresh in $\Instctx$ and $\Sigma$}
\\
\semlab{Ret} &
  \Handle^\mu\; U \With H &\reducesto& N[(\Box_{(\mu\circ\aex{\ell})}\,U)/x], \\
  & & &\multicolumn{1}{r}{
  \text{where } H = \{ \Ret x \mapsto N, \ell\;p\;r \mapsto N' \}
  }
\\
\semlab{Op} &
  \Handle^\mu\; \EC[\Do\ell \; U] \With H
    &\reducesto&
    N[U/p, (\Mod_{\mu}\;(\lambda y.\Handle^\mu\; \EC[y] \With H))/r]\\
  & \multicolumn{3}{r}{
    \text{where } \ell\notin\BL{\EC} \text{ and } H \ni (\ell \; p \; r \mapsto N)
  }
\\
\semlab{Lift} &
  \EC[M] &\reducesto& \EC[N]  \hfill\text{if } M \reducesto N \\
\end{nreductions}
\vspace{-4pt}
\caption{Operational semantics of \Metpt.}
\vspace{-.2\baselineskip}
\label{fig:semantics-metp}
\vspace{-.3\baselineskip}
\end{figure}

The operational semantics mostly follows \Met.
Rule \semlab{Gen} is new and generates a fresh runtime label for a
local label binding.
Moreover, since we generalise the handler of \Met, rules \semlab{ret}
and \semlab{Op} are also generalised.
Rule \semlab{Ret} wraps the return value with the modality
$\mu\circ\aex{\ell}$.
Rule \semlab{Op} wraps the continuation with the modality $\mu$.
The modalities in both rules are consistent with the typing rule
\tylab{Handle} in \Cref{sec:typing-metp}.
The function $\BL{\EC}$ gives the set of bound operation labels which
have handlers installed in the evaluation context $\EC$.
The condition $\ell\notin\BL{\EC}$ makes sure each operation $\ell$ is
handled by the dynamically innermost handler of $\ell$.

\subsection{Type Soundness and Effect Safety}
\label{sec:type-soundness}

To state syntactic type soundness, we first define normal forms.

\begin{definition}[Normal Forms]
  \label{def:normal-forms}
  We say that term $M$ is in normal form with respect to effect
  context $E$, if it is either in value normal form $M = U$ or of the
  form $M = \EC[\Do\ell\;U]$ for $\ell\subtype E$.
\end{definition}

The following theorems together give type soundness and effect safety.
They hold for any effect structure $\mcT$ satisfying the validity
conditions of \Cref{def:sanity-conditions}.

\begin{restatable}[Progress]{theorem}{progress}
  \label{thm:progress}
  In \Metp{\mcT} where $\mcT$ satisfies the validity conditions, if
  $\,\typm{\Instctx\mid\cdot}{M:A}{E}$, then either
  $M\mid\Instctx\reducesto N\mid\Instctx'$ for some $N$ and
  $\Instctx'$, or $M$ is in a normal form with respect to $E$.
\end{restatable}
\vspace{-.5\baselineskip}
\begin{restatable}[Subject Reduction]{theorem}{subjectReduction}
  \label{thm:subjrect-reduction}
  In \Metp{\mcT} where $\mcT$ satisfies the validity conditions, if
  $\,\typm{\Instctx\mid\Gamma}{M:A}{E}$ and $M\mid\Instctx \reducesto
  N\mid\Instctx'$, then $\typm{\Instctx'\mid\Gamma}{N:A}{E}$.
\end{restatable}

The proofs are given in \Cref{app:proofs-metp}.

\section{Encoding a Row-Based Effect System \`a la \Koka}
\label{sec:row-based-effect-systems}

In this section, we briefly present \Feps~\citep{XieBHSL20}, a System
F-style core calculus formalising the row-based effect system of
\Koka~\citep{KokaLang}, and show how to encode it into \Metp{\emtScp}.
We refer to \citet{XieBHSL20} for a complete introduction to \Feps.

\subsection{\FepsBold}
\label{sec:feps}

The syntax of \Feff is as follows.
\begin{prog}
\bs
  \slab{Value Types}\hspace{-.8em}  &A,B  &::= & \TUnit \mid \alpha
  \mid \earr{A}{B}{E} \mid \forall \alpha^K.A \\
  \slab{Effect Rows}\hspace{-1em}  &E,F &::= & \cdot \mid \evar \mid \ell, E \\
  \slab{Kind} &K & ::= & \Effect \mid \Value \\
  \slab{Contexts} &\Gamma  &::= & \cdot
                  \mid \Gamma,x : A
                  \mid \Gamma,\alpha : K
                  \\
  \slab{Label Contexts}\hspace{-1em} &\Sigma &::= & \cdot \mid \Sigma,{\ell : A\sto B} \\
\es
\quad
\bs
  \slab{Values}  &V,W & ::= & \Unit \mid x \mid \elambda{E} x^A . M \\
                    & & \mid & \Lambda \alpha^K . V \mid
                        {\Handler\;H}
                        \\
  \slab{Computations}\hspace{-1ex} &M,N  &::= &
                        \Ret V
                        \mid V\;W \mid V\,A \\
                    & &\mid & {\Do \ell\;V} \mid \Let x = M \In N \\
  \slab{Handlers} & H & ::= & \{\ell\;p\;r \mapsto N\} \\
\es
\end{prog}
Different from \citet{XieCIL22}, our version of \Feps is fine-grain
call-by-value~\citep{LevyPT03}.
Effect rows $E$ are scoped rows~\citep{Leijen05} with an optional tail
effect variable $\evar$.
As in \Metp{\mcT}, we assume a fixed global label context $\Sigma$.
By convention we write $\evar$ for effect variables and $\alpha$ for
value type variables.
We omit their kinds, $\Effect$ and $\Value$, when obvious.
In type abstraction and application, we let $A$ range over both value
types and effect rows, and let $\alpha$ range over their type
variables.

Typing judgements in \Feps include $\typf{\Gamma}{V}{A}{}$ for values
and $\typf{\Gamma}{M}{A}{E}$ for computations, where the latter tracks
effects $E$.
The typing rules and operational semantics of \Feps are standard for a
System F-style calculus with effect handlers and a row-based effect
system~\citep{linksrow,koka}.
We provide the full rules in \Cref{app:semantics-feps} and
show three representative typing rules here.
\begin{prog}
\inferrule*[Lab=\tylab{Abs}]
{
  \typf{\Gamma,
  x:A}{M}{B}{E}
}
{\typf{\Gamma}{\lambda^E x^A.M}{\earr{A}{B}{E}}{}}
\qquad \qquad
\inferrule*[Lab=\tylab{Do}]
{
  \Sigma \ni \ell : A\sto B \\\\
  \typf{\Gamma}{V}{A}{} \\
}
{\typf{\Gamma}{\Do \ell\;V}{B}{\ell,E}}
\qquad \qquad
\inferrule*[Lab=\tylab{Handler}]
{
  H = \{\ell\;p\;r \mapsto N \}\\
  \Sigma \ni \ell : A' \sto B' \\\\
  \typf{\Gamma,
  p : A', r: \earr{B'}{A}{E}}{N}{A}{E} \\
}
{\typf{\Gamma}{\text{$\Handler$}\; H}{\earr{(\earr{\TUnit}{A}{\ell,E})}{A}{E}}{}}
\end{prog}

Rule \tylab{Abs} introduces a $\lambda$-abstraction.
Rule \tylab{Do} invokes an operation $\ell$.
Rule \tylab{Handler} introduces a handler as a function that takes an
argument function of type $\earr{\TUnit}{A}{\ell,E}$ as in
\Cref{sec:overview-handler-feps}.
\citet{XieBHSL20} do not include a return clause in handlers for
\Feps.

\subsection{Encoding \Feps into \Metp{\emtScp}}
\label{sec:encoding-rows}

\begin{figure}[tb]
{\small
\renewcommand{\gray}[1]{#1}
\renewcommand{\arraystretch}{0.95}
\[
\hspace{-1em}
\ba[t]{r@{\ \ }c@{\ \ }l}
  \transl{-} &:& \text{Kind} \to \text{Kind}\\
  \transl{\Effect} &=& \Effect \\
  \transl{\Value} &=& \Pure \\[2ex]
  \transl{-} &:& \text{Effect Row} \to \text{Effect Context}\\
  \transl{\cdot} &=& \cdot \\
  \transl{\evar} &=& \evar \\
  \transl{\ell,E}   &=& \ell,\transl{E} \\[2ex]
  \transl{-} &:& \text{Label Context} \to \text{Label Context}\\
  \transl{\cdot} &=& \cdot \\
  \transl{\Sigma,\ell:A\sto B} &=& \tr{\Sigma},\ell:\tr{A}\sto\tr{B} \\[2ex]
  \transl{-} &:& \text{Value / Handler} \to \text{Term} \\
  \transl{\Unit} &=& \Unit \\
  \transl{x} &=& x \\[.5ex]
  \transl{\Lambda\alpha^K.V} &=& \Lambda\alpha^{\tr{K}}.\transl{V} \\[.5ex]
  \transl{\elambda{E} x^A . M} &=& \hle{\Mod_{\aeq{\transl{E}}}}\;(\lambda x^{\transl{A}}.\transl{M})
  \\[.5ex]
  \Bigl\llbracket
    {\vcenter{\bl\Handler\; H:\\ %
    \gray{\text{\small$\earr{(\earr{\TUnit}{A}{\ell,E})}{A}{E}$}}\el}}
  \Bigr\rrbracket
  &=&
  \hle{\Mod_{\aeq{\transl{E}}}}\;(
  \lambda f .
    \Handle^{\aeq{\tr{E}}}\;
    (\Let \hle{\Mod_{\aeq{\transl{\ell,E}}}}\; f' = f \In
      f'\;\Unit) \With \tr{H^{\gray{E}}}) \\[.5ex]
  \transl{\{\ell\;p\;r\mapsto N\}^{\gray{E}}} &=&
      \{
        \Ret x \mapsto \Let\hle{\Mod_{\aeq{\transl{\ell,E}}}}\; x' = x \In x',
        \ell\;p\;r\mapsto \transl{N}
      \}
\ea
\hfill
\hspace{-.43\textwidth}
\ba[t]{r@{\ \ }c@{\ \ }l}
  \transl{-} &:& \text{Type} \to \text{Type}\\
  \transl{\TUnit} &=& \TUnit \\
  \transl{\alpha} &=& \alpha \\
  \transl{\earr{A}{B}{E}} &=& \boxwith{\aeq{\transl{E}}}(\transl{A} \to \transl{B}) \\[.5ex]
  \transl{\forall\alpha^K.A} &=& \forall\alpha^{\tr{K}} . \transl{A} \\[2ex]
  \stransl{-}{} &:& \text{Context} \to \text{Context}\\
  \stransl{\cdot}{} &=& \cdot \\
  \stransl{\Gamma,x:A}{} &=& \stransl{\Gamma}{},x:\transl{A} \\
  \stransl{\Gamma,\alpha:K}{} &=& \stransl{\Gamma}{},\alpha:\tr{K} \\[2ex]
  \transl{-} &:& \text{Computation} \to \text{Term} \\
  \transl{\Ret V} &=& \tr{V} \\
  \transl{\Let x = M \In N} &=& \Let x = \tr{M} \In \tr{N} \\
  \transl{V\;A} &=& \transl{V}\;\tr{A} \\
  \transl{(V:\gray{A\to^E B})\;W} &=&
    \Let \hle{\Mod_{\aeq{\transl{E}}}}\; x = \transl{V} \In x\;\transl{W} \\
  \transl{\Do \ell\; V} &=& \Do\ell\;\tr{V} \\
\ea
\]
}
\caption{An encoding of \Feps in \Metp{\emtScp}.}
\vspace{-.2\baselineskip}
\label{fig:feps-to-metp}
\vspace{-.3\baselineskip}
\end{figure}

\Cref{fig:feps-to-metp} encodes \Feps in \Metp{\emtScp}.
The translation is mostly straightforward.

For kinds, we translate effect kind $\Effect$ to effect kind $\Effect$
and value kind $\Value$ to the kind $\Pure$.
We always translate values in \Feps into values of kind $\Pure$ in
\Metp{\emtScp}.

For types, we decouple effects from function types in \Feps by
translating an effectful function type $A\to^E B$ into a function type
with an absolute modality $\boxwith{\aeq{\tr{E}}}(\tr{A}\to\tr{B})$.

For contexts, we homomorphically translate each entry.

For terms, the translation is type-directed and essentially defined on
typing judgements.
We annotate components of a term with their types as necessary.
We highlight modality-relevant syntax of the term translation in
\gray{grey}.
The grey parts show how modalities decouple effect tracking.
The black parts themselves remain valid programs after type erasure.

The translation of lambda abstraction $\lambda^E x^A . M$ introduces
an absolute modality by $\Mod_{\aeq{\tr{E}}}$, and the translation of
function application $V\;W$ first eliminates the modality of $\tr{V}$
by $\Letm{}{\aeq{\tr{E}}} x = \tr{V}$ before applying it.
Examples for translations of lambda abstraction and application can be
found in \Cref{sec:overview-encoding-rows} as
$\tr{\namei{gen}{\Fepss}}$ and $\tr{\namei{app}{\Fepss}}$.

Translations of type abstraction, type application, operation
invocation, and let-binding are homomorphic.
Let-binding in \Metpt is syntactic sugar defined in the standard way
as $\Let x = M \In N \doteq (\lambda x . N)\; M$.
The translation of $\Ret V$ is simply $\tr{V}$.

A handler value $\Handler\;H$ of type
$\earr{(\earr{\TUnit}{A}{\ell,E})}{A}{E}$ is translated to a
higher-order function that handles the application of its function
argument $f$.
We eliminate the modality of $f$ before applying it to $\Unit$ since
$f$ has type $\boxwith{\aeq{\tr{\ell,E}}}(\TUnit\to\tr{A})$.
We introduce a modality $\Mod_{\aeq{\tr{E}}}$ for the whole translated
function since $\Handler\;H$ is an effectful function with effect $E$.
In the return clause of $\tr{H}$, we must eliminate the modality of
$x$ as shown in the typing rule \tylab{Handle} of \Metp{\emtScp}.
This modality elimination is always possible as the type $\tr{A}$ of
$x$ always has kind $\Pure$.
The operation clause of $\tr{H}$ demonstrates why we must use
modality-parameterised handlers.
Note that rule \tylab{Handler} of \Feps gives the continuation $r$ in
$H$ the type $\earr{B'}{A}{E}$.
By annotating $\Handle$ with $\aeq{\tr{E}}$, the continuation $r$ in
$\tr{H}$ has type $\boxwith{\aeq{\tr{E}}}(\tr{B'}\to\tr{A})$,
which is equal to $\tr{\earr{B'}{A}{E}}$.
We now give the full translation of the handler
$\namei{sum}{\Fepss}$ of \Cref{sec:overview-handler-feps}.
\begin{prog}
\ba{r@{}c@{}l}
\transl{\namei{sum}{\Fepss}} &\treq&
\bl\Lambda\evar . \Modhl{\aeq{\evar}}~(
\lambda f^{\boxwith{\aeq{\code{yield},\evar}}(\TUnit \to \Int)} .
    \Handle^{\aeq{\evar}}\;
 ({\Letmhl{}{\aeq{\code{yield},\evar}} f' = f \In}
  f\;\Unit) \With\\%[.5ex]
  \{\Ret x \mapsto {\Letmhl{}{\aeq{\code{yield},\evar}} x' = x \In} x',
    \ell\;p\;r \mapsto p + ({\Letmhl{}{\aeq{\evar}} r' = r \In} r'\;\Unit ) \}) \\%[.5ex]
  \el
  \\
&:& \forall\evar.\boxwith{\aeq{\evar}}(
      \boxwith{\aeq{\code{yield},\evar}}(\TUnit \to \Int)
      \to \Int)
\ea
\end{prog}

We have the following type and semantics preservation theorems with
proofs in \Cref{app:proof-feps}.

\begin{restatable}[Type Preservation]{theorem}{FepsToMetp}
  \label{lemma:type-preservation-feffn-to-metn}
  If $\,\typf{\Gamma}{M}{A}{E}$ in \Feps, then $\transl{\Gamma}
  \vdash \transl{M} : \transl{A} \atmode{\transl{E}}$ in
  \Metp{\emtScp}.
  Similarly for typing judgements of values.
\end{restatable}
\vspace{-.5\baselineskip}
\begin{restatable}[Semantics Preservation]{theorem}{FepsToMetpSemantics}
  \label{lemma:sematics-preservation-feffn-to-metn}
  If $M$ is well-typed and $M \reducesto N$ in \Feps, then $\transl{M}
  \reducesto^\ast \transl{N}$ in \Metp{\emtScp} where
  $\reducesto^\ast$ denotes the transitive closure of $\reducesto$.
\end{restatable}

\section{Encoding a Capability-Based Effect System \`a la \Effekt}
\label{sec:capability-based-effect-systems}

In this section we briefly present \SystemC~\citep{BrachthauserSLB22},
a core calculus formalising the capability-based effect system of
\Effekt~\citep{EffektLang}, and show how to encode it into
\Metp{\mcS}.
We refer to \citet{BrachthauserSLB22} for a complete introduction to
\SystemC.

\subsection{\SystemC}
\label{sec:systemc}

\Cref{fig:typing-systemc} gives the syntax and typing rules for
\SystemC, which is fine-grain call-by-value~\citep{LevyPT03} and
distinguishes between first-class values $V$, blocks $P$ (second-class
functions), and computations $M$.

\begin{figure}[tb] \small
\[
\bs
  \slab{Value Types}\hspace{-1em} &A,B  &::= & \TUnit \mid {T \At C} \\
  \slab{Block Types}\hspace{-1.2em} &T &::= & \carr{\ol{A}}{{\ol{f:T}}}{B} \\
  \slab{Capability Sets}\hspace{-1.2em} &C &::=&  {\{\ol{f}\}} \\
  \slab{Contexts} &\Gamma  &::= & \cdot
                  \mid \Gamma, x: A
                  \mid \Gamma, {f:^C T}
                  \mid \Gamma, f:^\ast T
                  \\
  \slab{Values}   &V,W  &::= & x \mid \Unit \mid {\CBox P} \\
  \slab{Handlers} & H & ::= & \{p\;r \mapsto N\} \\
\es
\hfill
\bs
  \slab{Blocks} &P,Q &::= & f \mid \block{\ol{x:A}}{\ol{f:T}}{M} \\
                &  &\mid &{\CUnbox V} \vphantom{{\ol{f}}} \\
  \slab{Computations}\hspace{-.5em} &M,N  &::= & \Ret V \mid P(\ol{V},\ol{Q}) \\
                  & &\mid& \Let x = M \In N \\
                  & &\mid&  \Def f = P \In N \\
                  & &\mid&  \Try\;\{f^{(A)\To B}\To M\}\With H \\
\es\]

\raggedright
\boxed{\typxi{\Gamma}{V}{A}{} \vphantom{\mid}}
\boxed{\typxi{\Gamma}{P}{T}{C}}
\hfill
\begin{mathparshrink}
\inferrule*[Lab=\tylab{Unit}]
{ }
{\typxi{\Gamma}{\Unit}{\TUnit}{}}

\inferrule*[Lab=\tylab{Var}]
{
  \Gamma \ni x : A
}
{\typxi{\Gamma}{x}{A}{}}

\inferrule*[Lab=\tylab{Box}]
{
  \typxi{\Gamma
  }{P}{T}{C}
}
{\typxi{\Gamma}{\CBox P}{T\At C}{}}

\inferrule*[Lab=\tylab{Transparent}]
{
  \Gamma \ni f :^C T
}
{\typxi{\Gamma}{f}{T}{C}}

\inferrule*[Lab=\tylab{Tracked}]
{
  \Gamma \ni f :^\ast T
}
{\typxi{\Gamma}{f}{T}{\{f\}}}

\inferrule*[Lab=\tylab{Unbox}]
{
  \typxi{\Gamma}{V}{T\At C}{}
}
{\typxi{\Gamma}{\CUnbox V}{T}{C}}

\inferrule*[Lab=\tylab{Block}]
{
  \typxi{\Gamma
  ,{\ol{x:A},\ol{f:^\ast T}}}{M}{B}{C \cup \{\ol{f}\} }
}
{\typxi{\Gamma}{\block{\ol{x:A}}{\ol{f:T}}{M}}{\carr{\ol{A}}{\ol{f:T}}{B}}{C}}

\inferrule*[Lab=\tylab{BSub}]
{
\typxi{\Gamma}{P}{T}{C'} \\
C' \subseteq C
}
{\typxi{\Gamma}{P}{T}{C}}
\end{mathparshrink}

\raggedright
\boxed{\typxi{\Gamma}{M}{A}{C}}
\hfill
\vspace{-1.2\baselineskip}
\begin{mathpar}
\inferrule*[Lab=\tylab{Value}]
{
  \typxi{\Gamma}{V}{A}{}
}
{\typxi{\Gamma}{\Ret V}{A}{\cdot}}

\inferrule*[Lab=\tylab{Let}]
{
  \typxi{\Gamma}{M}{A}{C} \\\\
  \typxi{\Gamma,x:A}{N}{B}{C'} \\
}
{\typxi{\Gamma}{\Let x = M \In N}{B}{C\cup C'}}

\inferrule*[Lab=\tylab{Call}]
{
  \typxi{\Gamma}{P}{\carr{\ol{A_i}}{\ol{f_j:T_j}}{B}}{C} \\\\
  \ol{\typxi{\Gamma}{V_i}{A_i}{}} \\
  \ol{\typxi{\Gamma
  }{Q_j}{T_j}{C_j}}
}
{\typxi{\Gamma}{P(\ol{V_i},\ol{Q_j})}{B[\ol{C_j/f_j}]}{C\cup\ol{C_j}}}

\inferrule*[Lab=\tylab{Sub}]
{
\typxi{\Gamma}{M}{A}{C'} \\\\
C' \subseteq C
}
{\typxi{\Gamma}{M}{A}{C}}

\inferrule*[Lab=\tylab{Def}]
{
  \typxi{\Gamma
  }{P}{T}{C'} \\\\
  \typxi{\Gamma,f :^{C'} T}{M}{A}{C} \\
}
{\typxi{\Gamma}{\Def f = P \In M}{A}{C}}

\inferrule*[Lab=\tylab{Handle}]
{
  \typxi{\Gamma
  ,{f:^\ast\carr{A'}{}{B'}}}{M}{A}{C\cup\{f\}} \\\\
  \typxi{\Gamma,p:A',r:^C\carr{B'}{}{A}}{N}{A}{C} \\
}
{\typxi{\Gamma}{\Try\; \{f^{\carr{A'}{}{B'}}\To M\} \With \{p\;r\mapsto N\}}{A}{C}}
\end{mathpar}
\caption{Syntax and typing rules for \SystemC. We mostly follow the
syntax of \citet{BrachthauserSLB22}. The main difference is that we
write $\To$ for block types to emphasise they are second-class.}
\label{fig:typing-systemc}
\end{figure}

We have three typing judgements for values, blocks, and computations
individually.
Judgements for blocks $\typxi{\Gamma}{P}{T}{C}$ and computations
$\typxi{\Gamma}{M}{A}{C}$ explicitly track a capability set $C$, which
contains the capabilities in $\Gamma$ that may be used.

The typing rules of \SystemC are much more involved than those of
\Feps as capability tracking is deeply entangled with term constructs
such as block constructions (\tylab{Block}), block calls
(\tylab{Call}), block bindings (\tylab{Def}), and usages of block
variables (\tylab{Transparent} and \tylab{Tracked}).
Due to space constraints, we focus on explaining these key rules.

There are two rules for uses of block variables as there are two forms
of block variable bindings in contexts.
A \emph{tracked} binding $f :^\ast T$ stands for a capability.
Rule \tylab{Tracked} tracks $f$ itself in the singleton capability set
$\{f\}$.
A \emph{transparent} binding $f :^C T$ stands for a user-defined block
whose capability set $C$ is known.
Rule \tylab{Transparent} tracks $C$ as the capability set.

Rules \tylab{Def} and \tylab{Block} both bind block variables.
Rule \tylab{Def} binds a block $P$ as a transparent block
variable $f:^{C'} T$ where $C'$ is the capability set of $P$.
Rule \tylab{Block} binds a list of tracked block variables
(capabilities) $\ol{f:^\ast T}$ whose concrete capability sets are
unknown until called.
The rule \tylab{Block} reflects the roles that block constructions
play for capability tracking as we introduced in
\Cref{sec:overview-encoding-caps-higher-order-block}.
For instance, all capabilities $\ol{f}$ are added to the capability
set of the block body $M$.

Rule \tylab{Call} fully applies a block $P$ to values $\ol{V_i}$ and
blocks $\ol{Q_j}$.
The rule reflects the roles that block calls play for capability
tracking as we introduced in \Cref{sec:overview-encoding-caps-block-calls}.
It substitutes each block variable ${f_j}$ (recall that these
variables are bound as $f_j:^\ast T$ in rule \tylab{Block}) with
the capability set $C_j$ of the block $Q_j$ in type $B$.
The capability set of the call is the union of the capability sets of
$P$ and all its block arguments because all these arguments might be
invoked.

Rule \tylab{Handle} defines a named handler which introduces a
capability $f:\carr{A'}{}{B'}$ to the scope of $M$.
Operation invocation via calling $f$ in $M$ is handled by this
handler.
The capability $f$ is added to the capability set of $M$.
The continuation $r$ is introduced as a transparent binding with
capability set $C$ as it may only use capabilities in $C$ provided by
the context.

\SystemC adopts named handlers and a generative semantics with a
reduction relation $M\mid\Instctx \reducesto N\mid\Instctx'$ where
$\Instctx ::= \cdot \mid \ell:\carr{A}{}{B}$ is a context for
runtime operation labels, similar to \Metpt.
The most interesting reduction rule is \semlab{Gen} which uses a
runtime capability value $\tmcap{\ell}$ with a runtime label $\ell$ to
substitute a capability $f$ introduced by a handler.
{\small
\begin{reductions}
\semlab{Gen} &\Try\; \{f^{\carr{A}{}{B}}\To M\} \With H \mid \Instctx
  &\reducesto &
  \Try_\ell\; M[\tmcap{\ell}/f] \With H \mid \Instctx, \ell:\carr{A}{}{B}
  \quad\text{where } \ell \text{ fresh}
\end{reductions}
}
The full specification of operational semantics can be found in
\Cref{app:semantics-cap}.

\subsection{Encoding \SystemC in \Metp{\mcS}}
\label{sec:encoding-caps}

\Cref{fig:SystemC-to-Metn} encodes \SystemC in \Metp{\mcS}.
The term translation is type-directed and defined on typing
judgements.
We annotate components of a term with their types and capability sets
as necessary.
We highlight syntax relevant to modalities and type abstraction of the
term translation in \hle{grey}.
The grey parts show how modalities decouple capability tracking.
The black parts remain valid programs after type erasure.
The encoding is unavoidably more involved than that of \Feps because
of the deeper entanglement of capability tracking with blocks.
As in \Cref{sec:systemc}, we focus on explaining the encoding of
block-relevant constructs.

\begin{figure}[tb]
{\small
\renewcommand{\gray}[1]{#1}
\[\ba[t]{r@{\ \ }c@{\ \ }l}
  \transl{-} &:& \text{Cap Set} \to \text{Effect Context}\\
  \transl{\{\ol{f}\}} &=& \ol{\shat{f}} \\[2ex]
  \transl{-} &:& \text{Value / Block Type} \to \text{Type}\\
  \transl{\TUnit} &=& \TUnit \\
  \transl{T\At C} &=& \boxwith{\aeq{\transl{C}}} \transl{T} \\[.5ex]
  \transl{\carr{\ol{A}}{\ol{f:T}}{B}} &=&
    \forall\ol{\shat{f}} . \boxwith{\aex{\ol{\shat{f}}}}(
    \ol{\transl{A}}\to \ol{\boxwith{\aeq{\shat{f}}}\transl{T}}\to \transl{B})
    \\[2ex]
  \transl{-} &:& \text{Value} \to \text{Term} \\
  \transl{\Unit} &=& \Unit \\
  \transl{x} &=& x \\
  \transl{\CBox P : \gray{T \At C}} &=& \hle{\Mod_{\aeq{\transl{C}}}}\;\transl{P}
  \\[2ex]
\ea
\hfill
\hspace{-3.2em}
\ba[t]{r@{\ \ }c@{\ \ }l}
  \stransl{-}{} &:& \text{Context} \to \text{Context}\\
  \stransl{\cdot}{} &=& \cdot \\
  \stransl{\Gamma,x:A}{} &=& \stransl{\Gamma}{},x:\transl{A} \\[.5ex]
  \stransl{\Gamma,f:^\ast T}{} &=& \stransl{\Gamma}{},
                                 \shat{f},
                                 {f}:\boxwith{\aeq{\shat{f}}}{\transl{T}},
                                 \hat{f}\varb{\aeq{\shat{f}}}{\transl{T}} \\[.5ex]
  \stransl{\Gamma,f:^C T}{} &=& \stransl{\Gamma}{},
                                 {f}:\boxwith{\aeq{\transl{C}}}{\transl{T}},
                                 \hat{f}\varb{\aeq{\transl{C}}}{\transl{T}} \\[1ex]
  \transl{-} &:& \text{Block} \to \text{Term}\\
  \transl{f} &=& \hat{f} \\
  \transl{
    \{({\ol{x:A}},{\ol{f:T}}) \To {M} \}
  } &=&
    \hle{\Lambda\ol{\shat{f}} .
    \Mod_{\aex{\ol{\shat{f}}}}}\;(
    \lambda\ol{x^{\transl{A}}}\,\ol{f^{\boxwith{\aeq{\shat{f}}}\transl{T}}} .\\
    & & \qquad \ol{\Let\hle{\Mod_{\aeq{\shat{f}}}}\; \hat{f} = f \In}
    \transl{M}) \\
  \transl{\CUnbox V : \gray{T \mid C}} &=&
    \Let \hle{\Mod_{\aeq{\transl{C}}}}\; x = \,\transl{V} \In x
  \\[2ex]
\ea\]
\[\ba[t]{r@{\ \ }c@{\ \ }l}
  \transl{-} &:& \text{Computation / Handler} \to \text{Term}\\
  \transl{\Ret V} &=& \transl{V} \\
  \transl{\Let x=M \In N} &=& \Let x = \transl{M} \In \transl{N} \\
  \transl{\Def f=P:\gray{T\mid C}\In N} &=&
    \Let f = \hle{\Mod_{\aeq{\transl{C}}}}\,\transl{P} \In
    \Let\hle{\Mod_{\aeq{\transl{C}}}\;} \hat{f} = f \In
    \transl{N} \\[1ex]
  \transl{P(\ol{V_i},\ol{Q_j:\gray{T_j\mid C_j}})} &=&
  \Let \hle{\Mod_{\aex{\ol{\transl{C_j}}}}\; \black{x = \transl{P}}\;\ol{\transl{C_j}}}\In
  x\;\ol{\transl{V_i}}\;\ol{\hle{(\Mod_{\aeq{\transl{C_j}}}\,\black{\transl{Q_j}})}}
  \\[.5ex]
  \bigtransl{\Try\; \{f^{\carr{A'}{}{B'}}\To M\}\\ \!\!\With H : \gray{A\mid C}}
  &=& \bl
    {\Localeffect{\ell_f : \tr{A'}\sto\tr{B'}}} \Let\hle{\Mod_{\aex{\ell_f}}}\;g = \\[.5ex]
      \qquad\hle{({\Lambda\shat{f}}.\Mod_{\aex{\shat{f}}}\;(}
        \lambda f.
      \Let\hle{\Mod_{\aeq{\shat{f}}}}\;\hat{f} = f \In \transl{M}\hle{))\;\ell_f} \\[1ex]
    \!\!\In\Handle^{\aeq{\tr{C}}}\;
      (g\; (\Modhl{\aeq{\ell_f}}\;(\Modhl{\aid}\;(\lambda x^{\tr{A'}} . \Do \ell_f\;x))))
    \With \tr{H^{\gray{f,C}}} \el \\[.5ex]
  \transl{\{p\;r\mapsto N\}^{\gray{f,C}}} &=&
    \{\bl
    \Ret x \mapsto \Let \hle{\Mod_{\aeq{\ell_f,\tr{C}}}}\; x' = x \In x', \\
    \ell_f\;p\;r\mapsto \Let \hle{\Mod_{\aeq{\transl{C}}}}\; \hat{r} = r \In \transl{N}
    \}\el
\ea\]
}
\caption{An encoding of \SystemC in \Metp{\mcS}.}
\label{fig:SystemC-to-Metn}
\end{figure}

For block constructions and block calls, we have explained their
encodings in detail in
\Cref{sec:overview-encoding-caps-higher-order-block,sec:overview-encoding-caps-block-calls},
using the constructions and calls of blocks \namei{app}{C} and
\namei{app'}{C} as examples.

A block binding $\Def f = P \In N$ not only binds a block $P$ to $f$
but also annotate the binding $f:^{C'} T$ with the capability set $C'$
of the block $P$ as shown by rule \tylab{Def} in
\Cref{fig:typing-systemc}.
For instance, we can bind the block \namei{gen}{C} in
\Cref{sec:overview-encoding-caps-blocks} to $f$ and apply it to $42$.
Its typing derivation is as follows.
\begin{prog}
\inferrule*
{
  \var{\yld}:^\ast\carrsingle{\Int}{}{\TUnit} \Hvdash
    \namei{gen}{C}
    \Hcolon \carrsingle{\Int}{}{\TUnit} \Hmid \{\var{\yld}\}
    \qquad\qquad
  \var{\yld}:^\ast\carrsingle{\Int}{}{\TUnit},
  {f :^{\{\yld\}} \carrsingle{\Int}{}{\TUnit}} \Hvdash
    f(42) \Hcolon \TUnit \Hmid \{\var{\yld}\} \\
}
{\var{\yld}:^\ast\carrsingle{\Int}{}{\TUnit} \Hvdash
  \force{\Def} f = \namei{gen}{C}
  \force{\In} f(42)
  \Hcolon \TUnit \Hmid \{\var{\yld}\}}
\end{prog}
The binding of $f$ in the second premise is annotated with its
capability set $\{\yld\}$ since \namei{gen}{C} uses the capability
$\yld$.
We cannot simply encode such a transparent binding by ignoring its
annotation of the capability set.
Instead, we use an absolute modality to simulate this annotation.
To encode the binding of $f$, we wrap the translated block
$\namei{gen}{C}$ into the absolute modality $\aeq{\yld}$.
The full translation of the above term is as follows, where we provide
the omitted identity modality in
\Cref{sec:overview-encoding-caps-blocks}.
\begin{prog}
\bl
\force{\Let f = \Modhl{\aeq{\shat{\yld}}}}\;(\force{\Modhl{\aid}}\;(\lambda
x^\Int . \hat{\var{\yld}}\;x))
\In
{\Letmhl{}{\aeq{\shat{\yld}}} \hat{f} = f \In}
\Letmhl{}{\aid} f' = \hat{f}\In f'\;42
\el
\end{prog}
We eliminate the modality $\aeq{\shat{\yld}}$ of $f$ and bind it to
$\hat{f}$, reminiscent of how we translate block arguments bound by
block constructions.
In general, for a transparent block variable binding $f :^{C} T$ in
the context, it is translated to two variable bindings
$f:\boxwith{\aeq{\tr{C}}}\tr{T}$ and $\hat{f} \varb{\aeq{\tr{C}}}
\tr{T}$.

The translation of uses of block variables is simple. We translate
each $f$ to its hat version $\hat{f}$.
The simplicity benefits from the fact that we eagerly eliminate the
modality of each $f$ after it is introduced, e.g., in the translations
of block constructions and block bindings.

The translation of named handlers $\Try\; \{f^{\carr{A'}{}{B'}}\To
M\}\With H$ is different from the translation of \namei{sum}{C} in
\Cref{sec:overview-handler-systemc}.
The full translation of \namei{sum}{C} is as follows, where we provide
the omitted identity modality of the function $\lambda x^\Int.\Do\ell_{\yld}\;x$.
\begin{prog}
\bl
\Localeffect{\ell_\yld : \Int\sto\TUnit}
\Letmhl{}{\aex{\ell_\yld}} g = (\hle{\Lambda \shat{\yld} .} \Modhl{\aex{\shat{\yld}}}\;(
  \lambda\yld
  . \Letmhl{}{\aeq{\shat{\yld}}}
  \hat{\yld} = \yld \In \hat{y}\;42;\hat{y}\;37;0))
  \;\hle{\ell_\yld} \\%[.5ex]
  \!\!\In
  \bl\Handle^{\aeq{\tr{C}}}\; (g\;
  (\Modhl{\aeq{\ell_\yld}}\;(\Modhl{\aid}\;(\lambda x^\Int . \Do \ell_\yld\;x))))
 \\%[.5ex]
  \!\!\force{\With}
  \{\bl\Ret x \mapsto {\Letmhl{}{\aeq{\ell_\yld,\tr{C}}} x' = x \In} x', %
    \ell_\yld\;p\;r\mapsto {\Letmhl{}{\aeq{\tr{C}}} \hat{r} = r \In}
    p + \hat{r}\;\Unit \}
     \el\el
\el
\end{prog}

The main difference is that, instead of directly using the local label
$\ell_\yld$ for the handled computation, we introduce an effect
variable $\shat{\yld}$ first and substitute it with $\ell_\yld$.
This extra layer of abstraction is necessary to keep the translation
systematic, because our translations of types and terms consistently
translate a capability $\yld$ to an effect variable $\shat{\yld}$.
After reducing the type application and substitution of $g$ in the
above translation term, we get the translation of \namei{sum}{C} in
\Cref{sec:overview-handler-systemc}.

In the return clause, we additionally eliminate the modality of $x$.
In the operation clause, we eliminate the modality $\aeq{\tr{C}}$ of
$r$ and bind it to $\hat{r}$ as we use a modality-parameterised
handler.
Using a modality-parameterised handler is important because in
\namei{sum}{C}, the continuation $r$ is a transparent binding of form
$f:^C \TUnit\to\Int$ as shown by the typing rule \tylab{Handle} of
\SystemC in \Cref{sec:systemc}.
We need to wrap the translated continuation $r$ with the absolute
modality $\aeq{\tr{C}}$ to be consistent with the translation of
transparent bindings.

For contexts, we translate each entry.
For a variable binding $x:A$, we translate it homomorphically.
For a transparent binding of a block variable $f:^C T$, we translate
it to two term variables $f$ and $\hat{f}$ as discussed in the
translation of $\Def\!\!$ above.
For a tracked binding of a block variable $f:^\ast T$, we translate it
to an effect variable $\shat{f}$ and two term variables $f$ and
$\hat{f}$ as discussed in \Cref{sec:overview-encoding-rows}.

We have the following type and semantics preservation theorems with
proofs in \Cref{app:proof-systemc}.

\begin{restatable}[Type Preservation]{theorem}{SystemCToMetp}
  \label{lemma:type-preservation-systemc-to-metn}
  If $\,\Gamma \vdash M : A \mid C$ in \SystemC, then
  $\stransl{\Gamma}{} \vdash \transl{M} : \transl{A}
  \atmode{\transl{C}}$ in \Metp{\mcS}.
  Similarly for typing judgements of values and blocks.
\end{restatable}
\vspace{-.5\baselineskip}
\begin{restatable}[Semantics Preservation]{theorem}{SystemCToMetpSemantics}
  \label{lemma:sematics-preservation-systemc-to-metn}
  If $M$ is well-typed and $M\mid\Instctx\reducesto N\mid{\Instctx'}$
  in \SystemC, then $\transl{M} \mid {\transl{\Instctx}}
  \reducesto^\ast \transl{N} \mid {\transl{\Instctx'}}$ in
  \Metp{\mcS}, where $\reducesto^\ast$ denotes the transitive closure
  of $\reducesto$.
\end{restatable}

\section{More Encodings and Discussions}
\label{sec:more-encodings}

In this section, we discuss more encodings of effect systems into
\Metpt, highlight practical language design insights gleaned from our
encodings, and outline potential extensions to \Metp{\mcT}.

\subsection{An Early Version of \Effekt}
\label{sec:systemxi}

\SystemXi~\citep{BrachthauserSO20} is an early core calculus of the
\Effekt language.
\SystemXi is essentially a fragment of \SystemC without boxes.
As a result, in \SystemXi capabilities can never appear in types since
we cannot box a second-class block into a first-class value.
While our encoding of \SystemC in \Cref{sec:encoding-caps} directly
gives an encoding of \SystemXi in \Metp{\mcS}, it introduces
unnecessary complexity.
Since capabilities never appear in types in \SystemXi, we do not need
to introduce an effect variable $\shat{f}$ for each capability $f$ in
the encoding.
It turns out that we can simply encode second-class blocks in
\SystemXi as first-class functions in \Metp{\mcS} without introducing
any extra term constructs.
For instance, a block ${\block{x:A}{f:T}{M}}$ is encoded as a function
$\lambda x^{\tr{A}}\,f^{\tr{T}}.\tr{M}$ by merely changing the notations.
We provide the full encoding of \SystemXi in \Metp{\mcS} in
\Cref{app:systemxi} and prove it preserves types and semantics in
\Cref{app:proof-systemxi}.

\subsection{Named Handlers in \Koka}
\label{sec:fepssn}

\citet{XieCIL22} extend \Koka with named handlers and formalise this
extension in the core calculus \Fepssn, which is based on \Feps.
\Fepssn allows each handler to bind a handler name that can be used to
invoke operations.
A handler name is similar to a capability in \SystemC but it is
a first-class value.
For instance, we can define a named handler in \Fepssn as follows.
\begin{prog}\bl
\namei{sum}{\Fepssns} \defeq
  \Lambda \evar . \force{\NHandler}\; \{\text{$\code{yield}\;p\;r\mapsto p + r\;\Unit$}\}\;
  : \forall\evar . (\forall a . \tyev{\code{yield}}{a} \to^{\code{yield}^a,\evar} \Int) \to^\evar \Int
\el\end{prog}
This handler is similar to the handler \namei{sum}{\Fepss} in
\Cref{sec:overview-handler-feps}.
The main difference is that the argument takes a value of type
$\tyev{\code{yield}}{a}$.
This is a first-class handler name with which we can invoke the
$\code{yield}$ operation.
For example, we can apply \namei{sum}{\Fepssn} as follows.
\begin{prog}\bl
  \namei{sum}{\Fepssns}~E~
  (\Lambda a . \lambda h^{\text{$\tyev{\scalebox{0.9}{\texttt{yield}}}{a}$}}
  . h\;42 ; h\;37 ; 0)
\el\end{prog}
Instead of using the label $\code{yield}$ to invoke the operation as
in application of \namei{sum}{\Fepss} in
\Cref{sec:overview-handler-feps}, we directly apply the handler name
$h$ to arguments.
This is reminiscent of the handler \namei{sum}{\SystemC} in
\Cref{sec:overview-handler-systemc} where we invoke the operation by
calling the capability introduced by the handler.
This program reduces to $79$.
The scope variable $a$ ensure scope safety of the handler name,
similar to the technique used by \lstinline{runST} in
\Haskell~\citep{LaunchburyJ95}.

As with the encoding of named handlers in \SystemC, we can encode a
named handler of \Fepssn by introducing a local label $\ell_a$ and
using the term $\Mod_{\aeq{\ell_a}}\;(\lambda x . \Do\ell_a\;x)$ to
simulate the handler name.
We use the effect structure $\mcS$ instead of $\emtScp$ as there can
never be duplicated handlers with the same name in \Fepssn.
The theory $\mcS$ gives us flexibility to have multiple effect
variables, which we use to encode scope variables.
We give the full encoding of \Fepssn in \Metp{\mcS} in
\Cref{app:fepssn} and prove its type and semantics preservation in
\Cref{app:proof-fepssn}.

\subsection{Insights for Language Design}
\label{sec:discussion}

In \Cref{sec:overview-comparison} and
\Cref{sec:overview-handler-comparison}, we demonstrated how our
encodings provide a direct way to compare the differences of \Feps and
\SystemC.
Moreover, our encodings can also help to inform language design
choices based on the following observations.

\begin{enumerate}[leftmargin=4ex]
\item Our encodings together demonstrate that modal effect types are
  as expressive as the row-based and capability-based effect systems
  we consider.
\item The encoding of \SystemXi (\Cref{sec:systemxi}) implies that we
  need not sacrifice first-class functions in order to obtain the
  benefits of the contextual effect polymorphism of \Effekt.
\item The encodings of \SystemC (\Cref{sec:encoding-caps}), \SystemXi
  (\Cref{sec:systemxi}), and \Fepssn (\Cref{sec:fepssn}) demonstrate
  that we can use local labels, a minimal extension as introduced in
  \Cref{sec:core-calculus}, to simulate the relatively heavyweight
  feature of named handlers in \Effekt and \Koka.
\item The encoding of \Fepssn (\Cref{sec:fepssn}) further demonstrates
  that the first-class handler names of \Koka offer no extra
  expressiveness over the second-class local labels of \Metpt.
\item The encoding of \SystemC (\Cref{sec:encoding-caps}) shows that
  instead of having a built-in form of capabilities which can appear
  at both term and type levels as in \Effekt and
  \Scala~\citep{BoruchGruszeckiOLLB23}, we can simulate it by
  introducing an effect variable for each argument and wrap the
  argument into an absolute modality with the corresponding effect
  variable.
\end{enumerate}

\subsection{Potential Extensions to \Metpt}
\label{sec:future-extensions}

We discuss three potential extensions to \Metp{\mcT} and leave their
full development as future work.

\paragraph{Effect Kinds}
We can extend the effect structure to abstract over effect kinds instead
of having a single kind $\Effect$. %
The augmented definition of effect structure is a triple $\mcT =
\langle\mR, :, \equiv\rangle$ where the new component $\mR$ is a set
of effect kinds.
We must extend the kinding and equivalence relations accordingly.
As an example of this extension, in order to characterise R\'emy-style
row types~\citep{Remy89} which use a kind system to ensure that there
is no duplicated label, we can declare $\mR = \{ \Row_\mcL \mid \mcL
\}$ where $\mcL$ is a label set and denotes all labels that must not
be in the row.
As another example, this extension enables us to combine different
effect structures together by assigning a kind to each theory.
For instance, we can declare two kinds $\meta{Set}$ and $\meta{Row}$
for theories $\mcS$ and $\emtScp$ respectively, and then give local
labels the kind $\meta{Set}$ and global labels the kind $\meta{Row}$.
We can then treat local labels as sets and global labels as scoped
rows.

\paragraph{Presence Types}
We can associate operation labels in extensions and effect contexts
with presence types~\citep{Remy94}.
Furthermore, instead of predefining the operation types for labels, we
can assign operation types to labels in extensions and effect contexts
in the manner of \citet{TangWDHLL25}.
For instance, the syntax of extensions could be extended to
$D::=\cdot\mid\ell:P,D\mid\evar,D$, where $P$ is a presence type
typically defined as $P ::= \Abs \mid \Pre{A\sto B} \mid \theta$.
A label can be absent ($\Abs$), present with a type ($\Pre{A\sto B}$),
or polymorphic over its presence ($\theta$).

\paragraph{Masking}
\Metp{\mcT} does not include the mask operator and the mask modality
$\amk{L}$ of \Met~\citep{TangWDHLL25}.
This enables us to substantially simplify the presentation of the core
calculus, especially the definitions relevant to modalities in
\Cref{sec:modalities}, compared to that of \citet{TangWDHLL25}.
Moreover, the lack of the mask operator does not influence our
encodings as the core calculi of \Effekt and \Koka do not have it.
Masking~\cite{BiernackiPPS18, ConventLMM20} is useful for effect
systems based on scoped rows where duplicated labels indicate nested
handlers for the same operation label.
With the mask operator, we can manually select which handler to use
when nested.
It is interesting future work to extend \Metp{\mcT} with a suitable
notion of abstract mask operator and extend the syntax of relative
modalities to $\adj{L}{D}$ where $L$ is a mask and $D$ is an
extension.
This extension will require extending the effect structure to define the
kinding and equivalence relations of masks.
A form of masking also makes sense for effect structures other than
$\emtScp$.
For instance, masking $\ell$ from a computation in $\mcS$ could be
used to disallow $\ell$ to be performed by the computation.

\section{Related and Future Work}
\label{sec:related-work}

\paragraph{Row-Based Effect Systems}
Row-based effect systems track effects by annotating function arrows
with row types denoting effects.
They have been adopted in research languages such as
\Links~\citep{linksrow}, \Koka~\citep{koka}, and \Frank~\citep{frank}.
\Links uses R{\'e}my-style row types with presence polymorphism
\citep{Remy94}, whereas \Koka and \Frank use scoped rows
\citep{Leijen05}.
\Eff~\citep{BauerP13} and \Helium~\citep{BiernackiPPS20} also track
effects on function arrows but treat effect types as sets.
In this paper we focus on \Koka, but we expect that other row-based
effect systems can be encoded similarly by instantiating the effect
structure appropriately.

\paragraph{Capability-Based Effect Systems}

Capability-based effect systems introduce and track effects as
capabilities.
Different variations diverge on when capability sets appear in types.
\Effekt~\citep{BrachthauserSO20,BrachthauserSLB22} uses second-class
functions and only attaches capability sets to types when boxing
functions.
\CaptureCalculus~\citep{BoruchGruszeckiOLLB23} and
\Capless~\citep{XuBPO25}, the foundations for capture tracking in
\Scala 3,
always annotate every type with its capability set and use subtyping
and syntactic sugar to simplify capability sets.
It is interesting future work to encode them in \Metpt.

\paragraph{Abstracting Effect Systems}
\citet{YoshiokaSI24} study different treatments of effect collections
in row-based effect systems.
They propose a parameterised core calculus, \LambdaEA, whose effect
types can be instantiated to various kinds of sets and rows.
The effect types in \LambdaEA are still entangled with function types.
As a result, \LambdaEA cannot encode capability-based effect systems.
We follow \LambdaEA in parameterising our core calculus \Metp{\mcT}
over different treatments of effect collections.
We make use of modalities to decouple effect tracking from function
types, enabling the encodings of both row-based and capability-based
effect systems.

\paragraph{Encoding into Modal Effect Types}
\citet{TangWDHLL25} consider a restricted row-based effect system in
which each effect type can refer only to the lexically closest effect
variable. This restricted system remains remarkably expressive and
suffices for many practical programs. Nonetheless, they show that it
can be encoded into simply-typed \Met without any effect polymorphism.
Our encodings consider richer source languages, showing that modal
effect types are as expressive as several row-based and
capability-based effect systems in the literature.

\paragraph{Local Effects}
Local labels in \Metpt allow us to introduce fresh effects locally.
They are useful for solving the effect encapsulation and accidental
handling problems~\citep{BiernackiPPS19,ConventLMM20,ZhangM19}.
As discussed in \Cref{sec:overview-handler-systemc}, there are various
local effect formalisms in the
literature~\citep{BiernackiPPS19,VilhenaP23,links-generative-labels};
most are based on dynamic generation of fresh effect names, whereas
the calculus of \citet{BiernackiPPS19} is based on effect coercions.
We conjecture that local labels of \Metpt are as expressive as
these formalisms.
We are interested in studying their relationship by encoding them into
\Metpt.

\paragraph{A Modal Type System for Benign Effects}
\citet{Nanevski04} propose a modal type system for benign effects in
Chapter 4.3. We refer to this system as MTBE.
MTBE supports local effects and indexes the standard necessity
modality $\BoxSym$ with effects for effect tracking.
In MTBE, a type $\BoxSym_E\;A$ means a computation which returns a
value of type $A$ and may perform effects in $E$.
This indexed necessity modality is similar to the absolute modality of
\Metpt.
The key difference between MTBE and \Metpt (and \Met) is that MTBE
has no notion of ambient effect context.
MTBE requires functions to be pure: every function type must specify
all the effects it may perform via a box.
In contrast, \Metpt allows a function to perform any effects from the
ambient effect context.
For example, an application function of type
$(\Int\to\TUnit)\to\Int\to\TUnit$ in \Metpt allows its argument to
perform any effects from the ambient effect context, whereas a
function with such a type in MTBE can only be applied to pure
functions (by default each function type has the empty $\BoxSym$).
In order to apply an application function to effectful arguments in
MTBE we must specify what effects may be performed in the type.
This requires parametric effect polymorphism in order to support
arbitrary effectful arguments.
Moreover, MTBE does not support relative modalities as relative
modalities are intimately tied to the notion of ambient effect
contexts.
Our encoding of \SystemC in \Cref{sec:encoding-caps} relies on the
notion of ambient effect contexts and relative modalities.
As a result, MTBE cannot serve as a general framework for encoding
various effect systems as \Metpt does.

\paragraph{Effectful Contextual Modal Type Theory}
\citet{ZyuzinN21} propose effectful contextual modal type theory
(ECMTT) which extends the \emph{contextual necessity
modality}~\citep{NanevskiPP08}
to track contexts of effectful operations.
Similar to MTBE, ECMTT also lacks the notion of ambient effect
contexts and is thus less expressive and flexible than \Metpt.
Moreover, ECMTT does not support dynamic generation of fresh effect
names and thus cannot express named handlers as in \Effekt.

\paragraph{Call-By-Push-Value}
Attempts to decouple programming language features have frequently
born fruit.
For instance, call-by-push-value (CBPV)~\citep{Levy2004} subsumes both
call-by-value (CBV) and call-by-name (CBN) by decoupling thunking and
forcing from function abstraction and application.
Our work is in a similar vein.
More interestingly, our encodings of \Feps and \SystemC possess
certain similarities with Levy's encodings of CBV and CBN into CBPV,
respectively.
In our encoding of \Feps, each function is wrapped in an absolute
modality, reminiscent of the CBV-to-CBPV encoding where each function
is thunked.
In our encoding of \SystemC, we only wrap a block in an absolute
modality when passing it as an argument, reminiscent of the
CBN-to-CBPV encoding, in which thunking of a function is deferred
until passing it as an argument.
We are interested in further exploring these similarities.

\paragraph{Expressive Power of Effect Handlers}

\citet{ForsterKLP19} compare the expressive power of effect handlers,
monadic reflection, and delimited control in a simply-typed setting
and show that delimited control cannot encode effect handlers in a
type-preserving way.
\citet{PirogPS19} extend the comparison between effect handlers and
delimited control to a polymorphic setting and show their equivalence.
\citet{IkemoriCM23} further show the typed equivalence between named
handlers and multi-prompt delimited control.
In contrast to these works, which compare effect handlers with
other programming abstractions, we compare different effect systems
for effect handlers.

\paragraph{Future Work}

In addition to the ideas already discussed above and in
\Cref{sec:future-extensions}, other directions for future work
include: exploring inverse encodings (from instantiations of \Metpt
into other calculi); studying parametricity and abstraction
safety~\citep{BiernackiPPS20,ZhangM19} for \Metpt; and further
developing \Metpt as a uniform intermediate language for type- and
effect-directed optimisation.

\begin{acks}
We thank Jonathan Immanuel Brachthäuser, Anton Lorenzen, Orpheas van
Rooij, Jesse Sigal, and the anonymous reviewers of ICFP 2025 and POPL
2026 for feedback.
Sam Lindley was supported by UKRI Future Leaders Fellowship ``Effect
Handler Oriented Programming'' (MR/T043830/1 and MR/Z000351/1) and by
the Huawei Edinburgh Joint Lab project “EPOCH: Effectful programming
on capability hardware”.
\end{acks}

\FloatBarrier
\bibliography{reference}

\ifnoappendix
\else
  \appendix
  \section{Formal Definitions of Effect Mode Theories}
\label{app:effect-mode-theories}

We provide the formal definitions of effect structures $\mcS$,
$\emtSimp$, and $\emtScp$ as introduced in
\Cref{sec:effect-mode-theory}.

\begin{figure}[htbp] \rulesize

\raggedright
\boxed{\Gamma\vdash D : \Effect}
\hfill
\begin{mathpar}
\inferrule*
{ }
{\Gamma \vdash \cdot : \Effect }

\inferrule*
{
  \Sigma,\Gamma \ni \ell : A\sto B \\
  \Gamma \vdash D : \Effect \\
}
{\Gamma \vdash \ell,D : \Effect}

\inferrule*
{
  \Gamma \vdash \evar : \Effect \\
  \Gamma \vdash D : \Effect \\
}
{\Gamma \vdash \evar,D : \Effect}
\end{mathpar}

\raggedright
\boxed{\Gamma \vdash D\equiv D'}
\hfill
\begin{mathpar}
\inferrule*
{ }
{D \equiv D}

\inferrule*
{
   D_1 \equiv D_2 \\ D_2 \equiv D_3
}
{D_1 \equiv D_3}

\inferrule*
{
  D \equiv D'
}
{\ell,D \equiv \ell,D'}

\inferrule*
{
  D \equiv D'
}
{\evar,D \equiv \evar,D'}

\inferrule*
{
}
{\ell,\ell',D \equiv \ell',\ell,D}

\inferrule*
{ }
{\ell,\evar,D \equiv \evar,\ell,D}

\inferrule*
{ }
{\evar,\evar',D \equiv \evar',\evar,D}

\inferrule*
{ }
{\ell,\ell,D \equiv \ell,D}

\inferrule*
{ }
{\evar,\evar,D \equiv \evar,D}
\end{mathpar}

\caption{The effect structure $\mcS$ (sets).}
\label{fig:emt-systemC}
\end{figure}

\begin{figure}[htbp] \rulesize

\raggedright
\boxed{\Gamma\vdash D : \Effect}
\hfill
\begin{mathpar}
\inferrule*
{ }
{\Gamma \vdash \cdot : \KindEffect}

\inferrule*
{
  \Sigma,\Gamma \ni \ell : A\sto B \\
  \Gamma \vdash D : \KindEffect \\
}
{\Gamma \vdash \ell,D : \KindEffect}
\end{mathpar}

\raggedright
\boxed{\Gamma \vdash D\equiv D'}
\hfill
\begin{mathpar}
\inferrule*
{ }
{D \equiv D}

\inferrule*
{
   D_1 \equiv D_2 \\ D_2 \equiv D_3
}
{D_1 \equiv D_3}

\inferrule*
{
  D \equiv D'
}
{\ell,D \equiv \ell,D'}

\inferrule*
{
  \ell \neq \ell'
}
{\ell,\ell',D \equiv \ell',\ell,D}
\end{mathpar}

\caption{The effect structure $\mathcal{R}_\textsf{scp}$ (scoped rows).}
\end{figure}

\begin{figure}[htbp] \rulesize
\raggedright
\boxed{\Gamma\vdash D : K}
\hfill
\begin{mathpar}
\inferrule*
{ }
{\Gamma \vdash \cdot : \KindEffect}

\inferrule*
{
  \Sigma,\Gamma \ni \ell : A\sto B \\
  \Gamma \vdash D : \KindEffect \\
}
{\Gamma \vdash \ell,D : \KindEffect}
\end{mathpar}

\raggedright
\boxed{\Gamma \vdash D\equiv D'}
\hfill
\begin{mathpar}
\inferrule*
{ }
{D \equiv D}

\inferrule*
{
   D_1 \equiv D_2 \\ D_2 \equiv D_3
}
{D_1 \equiv D_3}

\inferrule*
{
  D \equiv D'
}
{\ell,D \equiv \ell,D'}

\inferrule*
{
}
{\ell,\ell',D \equiv \ell',\ell,D}

\inferrule*
{ }
{\ell,\ell,D \equiv \ell,D}
\end{mathpar}

\caption{The effect structure $\mathcal{R}_\textsf{simp}$ (simple rows).}
\end{figure}

  \FloatBarrier
  \section{Omitted Rules of \Metp{\mcT}}
\label{app:full-spec}

We provide the kinding and type equivalence rules of \Metpt omitted in
\Cref{sec:core-calculus}.
We also provide some auxiliary definitions used by our proofs.

\subsection{Kinding and Well-Formedness}
\label{app:rules-metp}

The full kinding and well-formedness rules for \Metp{\mcT} are defined
in \Cref{fig:kinding-metp}.
For the global label context $\Sigma$ we require the kinding judgement
$\cdot\vdash A\sto B$ to hold for every $(\ell:A\sto B)\in \Sigma$.

\begin{figure}[htbp]\small
\raggedright
\boxed{\Gamma\vdash A : K\vphantom{\mu}}
\hfill
\begin{mathpar}
\inferrule*
{
  \Gamma \vdash A : \Pure
}
{\Gamma \vdash A : \Any}

\inferrule*
{
  \Gamma \ni \alpha : K
}
{\Gamma \vdash \alpha : K}

\inferrule*
{ }
{\Gamma \vdash \TUnit : \Pure}

\inferrule*
{
  \Gamma \vdash \aeq{E} \\
  \Gamma \vdash A : \Any \\
}
{\Gamma \vdash \boxwith{\aeq{E}} A : \Pure}

\inferrule*
{
  \Gamma \vdash \aex{D} \\
  \Gamma \vdash A : K \\
}
{\Gamma \vdash \boxwith{\aex{D}} A : K}

\inferrule*
{
  \Gamma \vdash A : \Any \\
  \Gamma \vdash B : \Any
}
{\Gamma \vdash A\to B : \Any}

\inferrule*
{
  \Gamma, \alpha:K \vdash A : K
}
{\Gamma \vdash \forall\alpha^K.A : K}
\end{mathpar}

\raggedright
\boxed{\Gamma\vdash \mu}
\hfill
\begin{mathpar}
\inferrule*
{
  \Gamma \vdash D : \Effect \\
}
{\Gamma \vdash \aex{D}}

\inferrule*
{
  \Gamma \vdash E : \Effect
}
{\Gamma \vdash \aeq{E}}
\end{mathpar}

\raggedright
\boxed{\Gamma\vdash A\sto B}
\hfill
\begin{mathpar}
\inferrule*
{
  \Gamma \vdash A : \Pure \\
  \Gamma \vdash B : \Pure
}
{\Gamma \vdash A \sto B}
\end{mathpar}

\raggedright
\boxed{\Gamma \atmode{E}}
\hfill
\begin{mathpar}
\inferrule*
{ }
{\cdot\atmode{E}}

\inferrule*
{
  \Gamma\atmode{F} \\
  \Gamma \vdash A : K \\
}
{\Gamma, x\varb{\mu_F}{A} \atmode{F}}

\inferrule*
{
  \Gamma\atmode{F} \\
  \mu(F) = E \\
}
{\Gamma, \lockwith{\mind{\mu}{F}} \atmode{E}}
\\

\inferrule*
{
  \Gamma\atmode{E} \\
}
{\Gamma, \alpha:K \atmode{E}}

\inferrule*
{
  \Gamma\atmode{E} \\
  \Gamma \vdash A \sto B \\
}
{\Gamma, \ell : A\sto B \atmode{E}}
\end{mathpar}
\caption{Kinding and well-formedness rules for \Metpt. Kinding rules of extensions and effect contexts are provided by the effect structure $\mcT$.}
\label{fig:kinding-metp}
\end{figure}

\subsection{Type Equivalence}
\label{app:equiv-and-order}

The type equivalence relation is defined in \Cref{fig:equiv-metn}.

\begin{figure}[htbp]\small
\raggedright
\boxed{\Gamma \vdash \mu\equiv \nu}
\hfill
\begin{mathpar}
\inferrule*
{
  E \equiv F
}
{\aeq{E} \equiv \aeq{F}}

\inferrule*
{
  D \equiv D'
}
{\aex{D} \equiv \aex{D'}}
\end{mathpar}
\raggedright
\boxed{\Gamma\vdash A\equiv B}
\hfill
\begin{mathpar}
\inferrule*
{
  \Gamma \vdash \alpha : K
}
{\Gamma \vdash \alpha \equiv \alpha}

\inferrule*
{ }
{\TUnit \equiv \TUnit}

\inferrule*
{
  \mu \equiv \nu \\
  A \equiv B
}
{\boxwith{\mu}A \equiv \boxwith{\nu}B}

\inferrule*
{
  A \equiv A' \\
  B \equiv B'
}
{A\to B \equiv A'\to B'}

\inferrule*
{
  \Gamma, \alpha : K \vdash A \equiv B
}
{\Gamma \vdash \forall\alpha^K.A \equiv \forall\alpha^K.B}
\end{mathpar}

\caption{Type equivalence rules for \Metpt. Type equivalence rules of extensions and effect contexts are provided by the effect structure $\mcT$.}
\label{fig:equiv-metn}
\end{figure}

\subsection{Mode Theory}
\label{app:mode-theory}

In the terminology of MTT~\citep{GratzerKNB20}, effect contexts are
\emph{modes}.
The structure of modes, modalities, and modality transformation
constitute the \emph{mode theory}.

To make the proofs easier, we frequently write modalities in the form
$\mu_F$ as introduced in \Cref{sec:modalities}.
Supposing $\mu(F) = E$, we read $\mu_F$ as a morphism $E\to F$ from
mode $E$ to mode $F$.
The reading of $\mu_F$ as a morphism between modes is consistent with
definition of modalities in MTT, while $\mu$ itself is actually an
indexed family of morphisms.
We call $\mu_F$ \emph{concrete modalities} since we have already
called $\mu$ modalities.
We repeat the definitions of modalities and modality composition using
syntax $\mu_F$ for easy reference. They are the same as those in
\Cref{sec:modalities}.
\[\ba{rcr@{\ \ }c@{\ \ }l}
\aeq{E}_F &:& E    &\to& F   \\
\aex{D}_F &:& D+F &\to& F \\
\ea\]
\[\ba{rclcll}
\aeq{E'}_F&\circ&\aeq{E}_{E'} &=& \aeq{E}_F
\\
\aex{D}_F&\circ&\aeq{E}_{D+F} &=& \aeq{E}_F
\\
\aeq{E}_F&\circ&\aex{D}_E &=& \aeq{D+E}_F
\\
\aex{D_1}_F&\circ&\aex{D_2}_{D_1+F} &=&
  \aex{D_2+D_1}_F
\\
\ea\]

We write $D+E$ for $D,E$ and $D+D'$ for $D,D'$ for notation
consistency with \Met.
We also write $\Gamma\vdash \mu_F\To\nu_F$ for
$\Gamma\vdash\mu\To\nu\atmode{F}$.
As we extend composition to concrete modalities, we also let the
operation $\locks{\Gamma}$ return a concrete modality as follows.
\[\bl
\locks{\cdot} = {\aid_F} \qquad
\locks{\Gamma,\lockwith{\mind{\mu}{F}}} = \locks{\Gamma}\circ {\mu_F} \qquad
\locks{\Gamma,x\varb{\mu_F}{A}} = \locks{\Gamma} \\
\el\]

  \FloatBarrier
  \section{Meta Theory and Proofs for \Metp{\mcT}}
\label{app:proofs-metp}

We provide meta theory and proofs for \Metp{\mcT} introduced in
\Cref{sec:core-calculus}.
The proofs are based on the proofs for \Met in \citet{TangWDHLL25} but
are parameterised over the effect structure $\mcT$.
We require the effect structure to satisfy the validity conditions in
\Cref{def:sanity-conditions}.\footnote{Our proofs (especially the
proof of progress in \Cref{app:progress}) only use the second validity
condition. The first condition is not necessary to show progress and
subject reduction but it is natural. We opt for keeping the first
condition.}

\subsection{Properties of the Mode Theory of \Metpt}
\label{app:proofs-mode-theory}

Our proofs for type soundness rely on some properties of the mode
theory of \Metpt.

First, the mode theory of \Metpt should form a double category.
The effect contexts and subeffecting (a preorder relation) form a
category generated by a poset.
The effect contexts (objects) and modalities (horizontal morphisms)
also form a category since modality composition possesses
associativity and identity.
We have the following lemma.

\begin{restatable}[Modes and modalities form a category]{lemma}{modCat}
  \label{lemma:mod-cat}
  Modes and modalities form a category with the identity morphisms
  $\one_E = \aid{}_E : E\to E$ and the morphism composition
  $\mu_F\circ\nu_{F'}$ such that
  \begin{enumerate}
    \item Identity: $\one_F\circ\mind{\mu}{F} = \mind{\mu}{F} =
    \mind{\mu}{F}\circ\one_E$ for $\mind{\mu}{F}:E\to F$.
    \item Associativity: $(\mu_{E_1}\circ\nu_{E_2})\circ\xi_{E_3} = \mu_{E_1}\circ(\nu_{E_2}\circ\xi_{E_3})$
    for $\mind{\mu}{E_1}:E_2\to E_1$, $\mind{\nu}{E_2}:E_3\to E_2$, and $\mind{\xi}{E_3}:E\to E_3$.
  \end{enumerate}
\end{restatable}
\begin{proof}
  By inlining the definitions of modalities and checking each case.
\end{proof}

As in \citet{TangWDHLL25}, we need to extend the modality
transformation relation a bit for meta theory and proofs.
We write $\Gamma\vdash \mu_F\To\nu_F$ if
$\Gamma\vdash\mu\To\nu\atmode{F}$.
We extend it to allow judgements of form
$\Gamma\vdash\mu_F\To\nu_{F'}$ where $F\subtype F'$ and add one new
rule \mtylab{Mono} as follows.
{
\begin{mathpar}
  {
  \inferrule*[Lab=\mtylab{Mono}]
  {
    \Gamma \vdash F \subtype F' \\
  }
  {\Gamma \vdash \mu_F \To \mu_{F'}}
  }
\end{mathpar}
}

Now we show that modality transformations are 2-cells in the double category.

\begin{restatable}[Modality transformations are 2-cells]{lemma}{modTransTwoCells}
  \label{lemma:modtrans-two-cells}
  If $\mu_F \To \nu_{F'}$, $\mu_F:E\to F$, and $\nu_{F'}:E'\to F'$,
  then $E\subtype E'$ and $F\subtype F'$.
  Moreover, the transformation relation is closed under vertical and
  horizontal composition as shown by the following admissible rules.
  \begin{mathpar}
  \inferrule*
  {
    \mu_{F_1} \To \nu_{F_2} \\
    \nu_{F_2} \To \xi_{F_3}
  }
  {\mu_{F_1} \To \xi_{F_3}}

  \inferrule*
  {
    \mu_F\To\mu'_{F'} \\
    \nu_E\To\nu'_{E'} \\
    \mu_F : E \to F \\
    \mu'_{F'} : E' \to F' \\
  }
  {\mu_F\circ\nu_E \To \mu'_{F'}\circ\nu'_{E'}}
  \end{mathpar}
\end{restatable}
\begin{proof}
  For the first part, $F\subtype F'$ is obvious from \mtylab{Mono}.
  $E\subtype E'$ follows from case analysis.
  \begin{description}
    \item[Case] $\mu = \aeq{E}$. Obvious from \mtylab{Abs} and \Cref{lemma:mono-mod}.
    \item[Case] $\mu = \aex{D_1}$ and $\nu = \aex{D_2}$. We need to
    show that $D_1,F \subtype D_2,F'$.
    By \mtylab{Extend} we have $D_1,F'\subtype D_2,F'$, which gives
    $D_1,F',F_1\equiv D_2,F'$ for some $F_1$.
    By $F\subtype F'$ we have $F,F_2 \equiv F'$ for some $F_2$. Then
    we have $D_1,F,F_2,F_1\equiv D_2,F'$.
    Finally we have $D_1,F\subtype D_2,F'$.
  \end{description}
  For the second part, vertical composition (the first rule) basically
  says that modality transformation is transitive. Easy to verify.
  Horizontal composition (the second rule) follows from a
  straightforward case analysis on shapes of modalities being
  composed.
  \begin{description}
    \item[Case] $\nu_E$ is an absolute modality. Suppose $\mu =
    \aeq{E_1}$. We have $(\mu\circ\nu)(F_1) = E_1$ for any $F'\subtype
    F_1$. By \Cref{lemma:mono-mod}, we have $E_1 \subtype (\mu'\circ\nu')(F_1)$.
    \item[Case] $\nu_E$ is an relative modality and $\mu_F$ is an
    absolute modality.
    Suppose $\mu = \aeq{E_1}$ and $\nu = \aex{D_1}$. We have
    $(\mu\circ\nu)(F_1) = D_1+E_1$ for any $F'\subtype F_1$. Similar
    to the above case, by \Cref{lemma:mono-mod}, we have $D_1+E_1
    \subtype (\mu'\circ\nu')(F_1)$.
    \item[Case] Both $\mu_F$ and $\nu_E$ are relative modalities. We
    also have that $\mu'_{F'}$ and $\nu'_{E'}$ are relative
    modalities.
    Suppose $\mu = \aex{D_1}, \nu = \aex{D_2}, \mu' = \aex{D_1'}, \nu' =
    \aex{D_2'}$.
    We have $E = D_1,F$ and $E' = D_1',F'$.
    By $\mu_F \To \mu'_{F'}$ and \mtylab{Extend}, we have
    \[ D_1,F_1 \subtype D_1',F_1 \]
    for all $F'\subtype F_1$.
    There exists $F_1'$ such that
    \[ \refa{D_1,F_1,F_1' \equiv D_1',F_1} \]
    By $\nu_E \To \nu'_{E'}$ and \mtylab{Extend}, we have
    \[ D_2,F_2 \subtype D_2',F_2 \]
    for all $E' = D_1',F' \subtype F_2$.
    There exists $F_2'$ such that
    \[ \refb{D_2,F_2,F_2' \equiv D_2',F_2} \]
    Given any $F'\subtype F_3$, by \refa{} we can find $F_{31}$ such that
    \[ D_1,F_3,F_{31} \equiv D_1',F_3 \]
    Then by $D_1',F' \subtype D_1',F_3$ and \refb{} we can find $F_{32}$ such that
    \[ D_2,D_1,F_3,F_{31},F_{32} \equiv D_2',D_1',F_3 \]
    Then we have
    \[ D_2,D_1,F_3 \subtype D_2',D_1',F_3 \]
    Finally by \mtylab{Extend} we have $\aex{D_2,D_1}_F \To \aex{D_2',D_1'}_{F'}$.
  \end{description}
\end{proof}

Beyond being a double category, we show some extra properties.
The most important one is that horizontal morphisms (sub-effecting)
act functorially on vertical ones (modalities). In other words, the
action of $\mu$ on effect contexts gives a total monotone function.

\begin{restatable}[Monotone modalities]{lemma}{monoModality}
  \label{lemma:mono-mod}
  If $\mind{\mu}{F}:E\to F$ and $F\subtype F'$, then
  $\mind{\mu}{F'}:E'\to F'$ with $E\subtype E'$.
\end{restatable}
\begin{proof}
  When $\mu$ is an absolute modality, obviously we have $E \equiv E'$.
  When $\mu$ is a relative modality $\aex{D}$, we need to show that
  $D+F\subtype D+F'$, which is obvious by $F\subtype F'$.
\end{proof}

\begin{restatable}[Soundness of modality transformation]{lemma}{soundModTransMetp}
  \label{lemma:sound-modtrans-metn}
  For modality transformation $\Gamma\vdash\mu\To\nu\atmode{F}$, we
  have $\act{\mu}{F'}\subtype\act{\nu}{F'}$ for all $F'$ with
  $F\subtype F'$.
\end{restatable}
\begin{proof}
  By case analysis on the two modality transformation rules.
  \begin{description}
    \item[Case] \mtylab{Abs}. Follow from \Cref{lemma:mono-mod}.
    \item[Case] \mtylab{Extend}. By definition.
  \end{description}
\end{proof}

We state some properties of the mode theory as the following lemmas
for easier references in proofs. Most of them directly follow from the
definition.

\begin{restatable}[Vertical composition]{lemma}{modtransVertical}
  \label{lemma:modtrans-vertical}
  If $\mu_{F_1}\To\nu_{F_2}$ and $\nu_{F_2}\To\xi_{F_3}$, then $\mu_{F_1}\To\xi_{F_3}$.
\end{restatable}
\begin{proof}
  Follow from \Cref{lemma:modtrans-two-cells}
\end{proof}

\begin{restatable}[Horizontal composition]{lemma}{modtransHorizontal}
  \label{lemma:modtrans-horizontal}
  If $\mind{\mu}{F}:E\to F$, $\mind{\mu'}{F'}:E'\to F'$, $\mind{\mu}{F}\To\mind{\mu'}{F'}$, and
  $\mind{\nu}{E}\To\nu'_{E'}$, then $\mind{\mu}{F}\circ\mind{\nu}{E}\To\mind{\mu'}{F'}\circ\mind{\nu'}{E'}$.
\end{restatable}
\begin{proof}
  Follow from \Cref{lemma:modtrans-two-cells}
\end{proof}

\begin{lemma}[Monotone modality transformation]
  \label{lemma:mono-modtrans}
  If $\mu_F \To \nu_F$ and $F\subtype F'$,
  then $\mu_{F'}\To \nu_{F'}$.
\end{lemma}
\begin{proof}
  By a case analysis.
  \begin{description}
    \item[Case] \mtylab{Abs}. Follow from \Cref{lemma:mono-mod}.
    \item[Case] \mtylab{Extend}. By definition.
  \end{description}
\end{proof}

\begin{lemma}[Asymmetric reflexivity of modality transformation]
  \label{lemma:self-modtrans}
  If $F\subtype F'$ and $\mu_F:E\to F$,
  then $\mu_F\To\mu_{F'}$.
\end{lemma}
\begin{proof}
  By \mtylab{Mono}.
\end{proof}

\subsection{Lemmas for the Calculus}
\label{app:CalcM-lemmas}

We prove structural and substitution lemmas for \Metpt as well as some
other auxiliary lemmas for proving type soundness.

\begin{lemma}[Canonical forms]~
  \label{lemma:canonical-forms}
  \begin{enumerate}[label=\arabic*.]
    \item If $\typm{\,}{U:\boxwith{\mu}A}{E}$, then $U$ is of shape $\Box_\mu\,U'$.
    \item If $\typm{\,}{U:A\to B}{E}$, then $U$ is of shape $\lambda x^A.M$.
    \item If $\typm{\,}{U:\forall\alpha^K. A}{E}$, then $U$ is of shape $\Lambda \alpha^K.V$.
    \item If $\typm{\,}{U:\TUnit}{E}$, then $U$ is $\Unit$.
  \end{enumerate}
\end{lemma}
\begin{proof}
  Directly follows from the typing rules.
\end{proof}

In order to define the lock weakening lemma, we first define a context
update operation $\updlock{\Gamma}{F'}$ which gives a new context
derived from updating the indexes of all locks and variable bindings
in $\Gamma$ such that $\locks{\updlock{\Gamma}{F'}} : E \to F'$ for
some $E$.
\[\ba{rcl}
\updlock{\cdot}{F} &=& \cdot \\
\updlock{\lockwith{\mind{\aeq{E}}{F'}},\Gamma'}{F} &=& \lockwith{\mind{\aeq{E}}{F}},\Gamma' \\
\updlock{\lockwith{\mind{\aex{D}}{F'}},\Gamma'}{F} &=& \lockwith{\mind{\aex{D}}{F}},\updlock{\Gamma'}{D+F} \\
\updlock{x\varb{\mu_{F'}}{A},\Gamma'}{F} &=& x\varb{\mu_F}{A},\updlock{\Gamma'}{F} \\
\updlock{\alpha:K,\Gamma'}{F} &=& \alpha:K,\updlock{\Gamma'}{F} \\
\updlock{\ell:A\sto B,\Gamma'}{F} &=& \ell:A\sto B,\updlock{\Gamma'}{F} \\
\ea\]

We have the following lemma showing that the index update operation
preserves the $\locks{-}$ operation except for updating the index.

\begin{lemma}[Index update preserves composition]
  \label{lemma:index-update}
  If $\mu_F = \locks{\Gamma} : E\to F$, $F \subtype F'$, and
  $\locks{\updlock{\Gamma}{F'}}:E'\to F'$, then $\locks{\updlock{\Gamma}{F'}} = \mu_{F'}$.
\end{lemma}
\begin{proof}
  By straightforward induction on the context and using the property
  that $(\mu\circ\nu)_F = \mu_F\circ\nu_E$ for $\mu_F : E\to F$.
\end{proof}

\begin{corollary}[Index update preserves transformation]
  \label{lemma:update-modtrans}
  If $\locks{\Gamma}:E\to F$, $F \subtype F'$, and
  $\locks{\updlock{\Gamma}{F'}}:E'\to F'$, then $\locks{\Gamma}\To
  \locks{\updlock{\Gamma}{F'}}$.
\end{corollary}
\begin{proof}
  Immediately follow from \Cref{lemma:index-update} and
  \Cref{lemma:self-modtrans}.
\end{proof}

We have the following structural lemmas.

\begin{restatable}[Structural rules]{lemma}{structuralRules} ~
  \label{lemma:structural-rules}
  The following structural rules are admissible.
  \begin{enumerate}[label=\arabic*.]
    \item Variable weakening.
    \begin{mathpar}
      \inferrule*
      {
        \typm{\Gamma,\Gamma'}{M:B}{E} \\
        \Gamma,x\varb{\mu_F}{A},\Gamma'\atmode{E}
      }
      {\typm{\Gamma,x\varb{\mu_F}{A},\Gamma'}{M:B}{E}}
    \end{mathpar}
    \item Variable swapping.
    \begin{mathpar}
      \inferrule*
      {
        \typm{\Gamma,x\varb{\mu_F}{A},y\varb{\nu_F}{B},\Gamma'}{M:A'}{E}
      }
      {\typm{\Gamma,y\varb{\nu_F}{B},x\varb{\mu_F}{A},\Gamma'}{M:A'}{E}}
    \end{mathpar}
    \item Lock weakening.
    \begin{mathpar}
      \inferrule*
      {
        \typm{\Gamma,\lockwith{\mind{\mu}{F}},\Gamma'}{M:A}{E} \\
        \mu_F \To \nu_F \\
        \nu_F : F'\to F \\
        \locks{\updlock{\Gamma'}{F'}}:E'\to F' \\
      }
      {\typm{\Gamma,\lockwith{\mind{\nu}{F}},\updlock{\Gamma'}{F'}}{M:A}{E'}}
    \end{mathpar}
    \item Type variable weakening.
    \begin{mathpar}
      \inferrule*
      {
        \typm{\Gamma,\Gamma'}{M:B}{E}
      }
      {\typm{\Gamma,\alpha:K,\Gamma'}{M:B}{E}}
    \end{mathpar}
    \item Type variable swapping.
    \begin{mathpar}
      \inferrule*
      {
        \typm{\Gamma_1,\Gamma_2,\alpha:K,\Gamma_3}{M:A}{E}
      }
      {\typm{\Gamma_1,\alpha:K,\Gamma_2,\Gamma_3}{M:A}{E}}

      \inferrule*
      {
        \alpha\notin\ftv{\Gamma_2} \\
        {\typm{\Gamma_1,\alpha:K,\Gamma_2,\Gamma_3}{M:A}{E}}
      }
      {\typm{\Gamma_1,\Gamma_2,\alpha:K,\Gamma_3}{M:A}{E}}
    \end{mathpar}
    \item Label weakening.
    \begin{mathpar}
      \inferrule*
      {
        \typm{\Gamma,\Gamma'}{M:B}{E}
      }
      {\typm{\Gamma,\ell:A\sto B,\Gamma'}{M:B}{E}}
    \end{mathpar}
    \item Lock swapping.
    \begin{mathpar}
      \inferrule*
      {
        {\typm{\Gamma,\lockwith{\mu},x\varb{\nu} A,\Gamma'}{M:B}{E}} \\
        \Gamma \vdash A : \Pure \text{ or } \nu = \aeq{F} \\
      }
      {\typm{\Gamma,x\varb{\nu} A,\lockwith{\mu},\Gamma'}{M:B}{E}}
    \end{mathpar}
  \end{enumerate}
\end{restatable}
\begin{proof}
1, 2, 4, 5, 6, 7 follow from straightforward induction on the typing
derivation.
Among them, 7 makes use of the transformation rule \mtylab{Abs}.
3 follows from a induction on the typing derivation.
The most interesting cases are \tylab{Do} and \tylab{Handle} due to
the abstraction over the effect mode theory.
We show the proof as follows.
\begin{description}
\item[Case]
  \begin{mathpar}
  \inferrule*[Lab=\tylab{Var}]
  {
    \nu'_{F_1}= \locks{\Gamma_2} : E\to F_1 \\
    \inferrule*
    {
    \refa{\mu'_{F_1}\To\nu'_{F_1}} \text{ or } \Gamma_1\vdash A:\Pure
    }
    {(\mu,A) \To \nu' \atmode{F_1}}
  }
  {\typm{\Gamma_1,x\varb{\mu'_{F_1}} A,\Gamma_2}{x:A}{E}}
  \end{mathpar}
  Trivial when $A$ is pure. Otherwise, case analysis on where the lock
  weakening happens.
  \begin{description}
    \item[Case] $\Gamma_1$. Supposing $\Gamma_1 =
    \Gamma,\lockwith{\mu_F},\Gamma_0$ and after lock weakening we have
    $\Gamma,\lockwith{\nu_F},\Gamma_0',x\varb{\mu'_{F_1'}},\Gamma_2'$
    where $\Gamma_2' = \updlock{\Gamma_2}{F_1'} : E'\to F_1'$ and
    $\Gamma_0' = \updlock{\Gamma_0}{F'}:F_1'\to F'$.
    By \Cref{lemma:index-update} on $\Gamma_0$, $F\subtype F'$, and
    \Cref{lemma:mono-mod}, we have $F_1\subtype F_1'$.
    Then by \refa{} and \Cref{lemma:mono-modtrans}, we have
    $\mu'_{F_1'}\To\nu'_{F_1'}$.
    Then by \Cref{lemma:index-update} we have $\nu'_{F_1'} =
    \locks{\Gamma_2'}$.
    Finally by \tylab{Var} we have
    \[
      {\typm{\Gamma,\lockwith{\nu_F},\Gamma_0',x\varb{\mu'_{F_1'}} A,\Gamma_2'}{x:A}{E'}}
    \]
    \item[Case] $\Gamma_2$. Suppose $\Gamma_2 =
    \Gamma_0,\lockwith{\mu_F},\Gamma'$ and after lock weakening
    $\Gamma_2$ is replaced by $\Gamma_2' = \Gamma_1,x\varb{\mu'_{F_1}}
    A,\Gamma_0,\lockwith{\nu_F},\updlock{\Gamma'}{F'}$.
    By \Cref{lemma:update-modtrans} we have
    $\locks{\Gamma'}\To\locks{\updlock{\Gamma'}{F'}}$.
    Then by \Cref{lemma:modtrans-horizontal} we have
    $\locks{\Gamma_2} \To \locks{\Gamma_2'}$.
    By \Cref{lemma:modtrans-vertical} and \refa{}, we have $\mu'_{F_1}
    \To \locks{\Gamma_2'}$.
    Finally by \tylab{Var} we have
    \[
      {\typm{\Gamma,x\varb{\mu'_{F_1}},\Gamma_2'}{x:A}{E'}}
    \]
  \end{description}
\item[Case]
  \begin{mathpar}
    \inferrule*[Lab=\tylab{Mod}]
    {
      \mind{\mu'}{E} : F_1 \to E \\
      \refa{\typm{\Gamma,\lockwith{\mind{\mu}{F}},\Gamma',\lockwith{\mind{\mu'}{E}}}{V:A}{F_1}}
    }
    {\typm{\Gamma,\lockwith{\mind{\mu}{F}},\Gamma'}{\Box_{\mu'}\,V : \boxwith{\mu'} A}{E}}
  \end{mathpar}
  We have
  \[
   \updlock{\Gamma',\lockwith{\mind{\mu'}{E}}}{F'}
  = \updlock{\Gamma'}{F'}, \updlock{\lockwith{\mind{\mu'}{E}}}{E'}
  = \updlock{\Gamma'}{F'}, \lockwith{\mind{\mu'}{E'}}.
  \]
  Supposing $\mind{\mu'}{E'}:F_1' \to E'$, by
  $\locks{\updlock{\Gamma'}{F'}, \lockwith{\mind{\mu'}{E'}}}:F_1'\to F'$ and
  IH on \refa{}, we have
  \[
    \typm{\Gamma,\lockwith{\mind{\nu}{F}},\updlock{\Gamma'}{F'}, \lockwith{\mind{\mu'}{E'}}}{V:A}{F_1'}.
  \]
  Then by \tylab{Mod} we have
  \[
  \typm{\Gamma,\lockwith{\mind{\nu}{F}},\updlock{\Gamma'}{F'}}{\Box_{\mu'}\,V : \boxwith{\mu'} A}{E'}.
  \]
  \item[Case]
  \begin{mathpar}
    \inferrule*[Lab=\tylab{Letmod}]
    {
      \nu'_E : F_1\to E \\
      \refa{\typm{\Gamma,\lockwith{\mu_F},\Gamma',\lockwith{\nu'_E}}{V : \boxwith{\mu'} A}{F_1}} \\
      \refb{\typm{\Gamma,\lockwith{\mu_F},\Gamma',x\varb{\nu'_E\circ\mu'_{F_1}}{A}}{M:B}{E}}
    }
    {\typm{\Gamma,\lockwith{\mu_F},\Gamma'}{\Letm{\nu'}{\mu'} x = V \In M : B}{E}}
  \end{mathpar}
  By IH on \refa{}, we have
  \[
    \typm{\Gamma,\lockwith{\nu_F},\updlock{\Gamma'}{F'},\lockwith{\nu'_{E'}}}{V : \boxwith{\mu'} A}{F_1'}
  \]
  where $\nu'_{E'} : F_1' \to E'$.
  By IH on \refb{}, we have
  \[
    \typm{\Gamma,\lockwith{\nu_F},\updlock{\Gamma'}{F'},x\varb{\nu'_{E'}\circ\mu'_{F_1'}}{A}}{M:B}{E'}.
  \]
  Then by \tylab{Letmod}, we have
  \[
    {\typm{\Gamma,\lockwith{\mu_F},\updlock{\Gamma'}{F'}}{\Letm{\nu'}{\mu'} x = V \In M : B}{E'}}
  \]
\item[Case]
  \begin{mathpar}
    \inferrule*[Lab=\tylab{Do}]
    {
      \Sigma,\Gamma,\Gamma' \ni \ell : A\sto B \\
      \refa{\typm{\Gamma,\lockwith{\mu_F},\Gamma'}{N : A}{\ell, E}} \\
    }
    {\typm{\Gamma,\lockwith{\mu_F},\Gamma'}{\Do\ell\; N : B}{\ell, E}}
  \end{mathpar}
  Suppose $\locks{\Gamma'} = \xi_{\mu(F)}$.
  We have $\ell \subtype \xi(\mu(F))$.
  By \Cref{lemma:index-update} we have $\locks{\updlock{\Gamma'}{\nu(F)}} = \xi_{\nu(F)}$
  By IH on \refa{} we have
  \[
    \typm{\Gamma,\lockwith{\nu_F},\updlock{\Gamma'}{\nu(F)}}{N : A}{\xi(\mu(F))}
  \]
  By \Cref{lemma:modtrans-two-cells} and $\mu_F\To\nu_F$ we have
  $\mu(F)\subtype\nu(F)$.
  By \Cref{lemma:mono-mod} we have $\xi(\mu(F))\subtype\xi(\nu(F))$.
  By transitivity of subeffecting we have $\ell \subtype \xi(\nu(F))$.
  Finally our goal follows from reapplying \tylab{Do}.
\item[Case]
  \begin{mathpar}
    \inferrule*[Lab=\tylab{Handle}]
    {
      \mu(F) = E \\
      \Gamma, \lockwith{\xi_{F_1}}, \Gamma' \vdash \mu \To \aid \atmode{F} \\
      \Gamma, \lockwith{\xi_{F_1}}, \Gamma' \vdash \mu \To \mu\circ\mu \atmode{F} \\
      \refa{\typm{\Gamma, \lockwith{\xi_{F_1}}, \Gamma', \lockwith{\mu_F}, \lockwith{\aex{\ell}_{E}}}{M : A}{\ell,E}} \\
      \Sigma,\Gamma,\Gamma' \ni \ell : A'\sto B' \\
      \refb{\typm{\Gamma, \lockwith{\xi_{F_1}}, \Gamma', \lockwith{\mu_F}, x : \boxwith{(\mu\circ\aex{\ell})} A}{N : B}{E}} \\
      \refc{\typm{\Gamma, \lockwith{\xi_{F_1}}, \Gamma', \lockwith{\mu_F}, p : A', r: \boxwith{\mu}(B' \to B)}{N' : B}{E}} \\
    }
    {\typm{\Gamma, \lockwith{\xi_{F_1}}, \Gamma'}{\Handle^\mu\;M\With
      \{\Ret x \mapsto N, \ell\;p\;r \mapsto N' \}
    : B}{F}}
  \end{mathpar}
  Suppose we want to do the lock weakening $\xi_{F_1} \To \nu_{F_1}$.
  Our goal follows from IHs on \refa{}, \refb{}, \refc{} and
  reapplying \tylab{Handle}.
  As in the case of \tylab{Do}, we need to make sure that after IH on
  \refa{} the label $\ell$ is still in the effect context.
  This is guaranteed by the lock $\lockwith{\aex{\ell}}$.
\item[Case] \tylab{TAbs}, \tylab{Abs}, \tylab{TApp}, \tylab{App}.
Follow from IH. Similar to other cases we have shown.
\end{description}
\end{proof}

The following lemma reflects the intuition that pure values can be
used in any effect context.

\begin{lemma}[Pure Promotion]
  \label{lemma:pure-promotion}
  The following promotion rule is admissible.
  \begin{mathpar}
    \inferrule*
    {
      \typm{\Gamma_1,\Gamma}{V:A}{E} \\
      \Gamma_1\vdash A : \Pure \\
      \Gamma_1,\Gamma'\atmode{E'} \\
      \fv{V} \cap \dom{\Gamma} = \emptyset \\
    }
    {\typm{\Gamma_1,\Gamma'}{V:A}{E'}}
  \end{mathpar}
\end{lemma}
\begin{proof}
By induction on the typing derivation of $V$.
\begin{description}
\item[Case] \tylab{Var}. Trivial.
\item[Case]
\begin{mathpar}
\inferrule*[Lab=\tylab{Mod}]
{
  \mind{\mu}{E} : F_1 \to E \\
  \refa{\typm{\Gamma_1,\Gamma,\lockwith{\mind{\mu}{E}}}{V:A}{F_1}}
}
{\typm{\Gamma_1,\Gamma}{\Box_\mu\,V : \boxwith{\mu} A}{E}}
\end{mathpar}
Case analysis on the shape of $\mu$.
\begin{description}
  \item[Case] $\mu$ is relative. $A$ must have kind $\Pure$.
  By IH on \refa{}, we have
  \[
    \typm{\Gamma_1,\Gamma',\lockwith{\mind{\mu}{E'}}}{V:A}{F_1'}
  \]
  where $\mu_{E'}:F_1'\to E'$.
  Then by \tylab{Mod} we have
  \[
    \typm{\Gamma_1,\Gamma'}{\Box_\mu\,V:\boxwith{\mu} A}{E'}
  \]
  \item[Case] $\mu$ is absolute. We have $\mu = \aeq{F_1}$ and
  $\locks{\Gamma',\lockwith{\mu_{E'}}} = \aeq{F_1}_{F} =
  \locks{\Gamma,\lockwith{\mu_E}}$.
  Thus, replacing the context $(\Gamma,\lockwith{\mu_E})$ with
  $(\Gamma',\lockwith{\mu_{E'}})$ in \refa{} does not influence all
  usages of \tylab{Var} in the derivation tree of \refa{}. We have
  \[
    {\typm{\Gamma_1,\Gamma',\lockwith{\mind{\mu}{E'}}}{V:A}{F_1}}
  \]
  Then by \tylab{Mod} we have
  \[
    \typm{\Gamma_1,\Gamma'}{\Box_\mu\,V:\boxwith{\mu} A}{E'}
  \]
\end{description}
\item[Case] \tylab{TAbs}. Follow from IH and
\Cref{lemma:structural-rules}.5.
\item[Case] \tylab{Abs}. Impossible since function types are impure.
\end{description}
\end{proof}

\begin{lemma}[Generalised Pure Promotion]
  \label{lemma:pure-promotion-vars}
  Given $\typm{\Gamma_1,\Gamma}{M:A}{E}$, if for any $x\in\ftv{M}$, we
  have $\Gamma_1 \ni x\varb{\mu} B$ such that either $\Gamma_1\vdash B
  : \Pure$ or $\mu$ is an absolute modality, then
  $\typm{\Gamma_1,\Gamma'}{M:A}{E'}$ for $E\subtype E'$ and $\Gamma_1,\Gamma'\atmode{E'}$.
\end{lemma}
\begin{proof}
  By straightforward induction on typing judgements in \Metpt.
  The most non-trivial case is to show the accessibility of each
  variable $x\in\ftv{M}$.
  For variables with types of kind $\Pure$, we can always access them.
  For variables annotated with an absolute modality, the modality
  transformation relation \mtylab{Abs} still holds because $E\subtype E'$.
\end{proof}

\begin{restatable}[Substitution]{lemma}{substitution} ~
  \label{lemma:substitution}
  The following substitution rules are admissible.
  \begin{enumerate}[label=\arabic*.]
    \item Preservation of kinds under type substitution.
    \begin{mathpar}
      \inferrule*
      {
        \Gamma\vdash A : K \\
        \Gamma,\alpha:K,\Gamma'\vdash B : K
      }
      {\Gamma,\Gamma'\vdash B[A/\alpha] : K}
    \end{mathpar}
    \item Preservation of types under type substitution.
    \begin{mathpar}
      \inferrule*
      {
        \Gamma\vdash A : K \\
        \typm{\Gamma,\alpha:K,\Gamma'}{M:B}{F} \\
      }
      {\typm{\Gamma,\Gamma'}{M[A/\alpha] : B[A/\alpha]}{F}}
    \end{mathpar}
    \item Preservation of types under value substitution.
    \begin{mathpar}
      \inferrule*
      {
        \typm{\Gamma,\lockwith{\mu_F}}{V:A}{F'} \\
        \typm{\Gamma,x\varb{\mu_F}{A},\Gamma'}{M:B}{E} \\
      }
      {\typm{\Gamma,\Gamma'}{M[V/x]:B}{E}}
    \end{mathpar}
    \item Preservation of types under label substitution.
    \begin{mathpar}
      \inferrule*
      {
        \typm{\Instctx\mid\Gamma,\ell:A'\sto B',\Gamma'}{M:A}{E} \\
        \ell' \in \Instctx
      }
      {\typm{\Instctx\mid\Gamma,\Gamma'[\ell'/\ell]}{M[\ell'/\ell]:A[\ell'/\ell]}{E[\ell'/\ell]}}
    \end{mathpar}
  \end{enumerate}
\end{restatable}
\begin{proof} ~ \\
\noindent 1,2,4. Follow from straightforward induction.

\noindent 3. By induction on the typing derivation of $M$.
Trivial when variable $x$ is not used. In the following induction we
always assume $x$ is used.
\begin{description}
\item[Case]
\begin{mathpar}
\inferrule*[Lab=\tylab{Var}]
{
  \mind{\nu}{F}= \locks{\Gamma'} : E\to F \\
  \inferrule*
  {
  \refa{\mind{\mu}{F}\To\nu_F} \text{ or } \Gamma\vdash A:\Pure
  }
  {(\mu,A) \To \nu \atmode{F}}
}
{\typm{\Gamma,x\varb{\mu_F}{A},\Gamma'}{x:A}{E}}
\end{mathpar}
Case analysis on the kind of $A$
\begin{description}
  \item[Case] $A$ does not have kind $\Pure$. By $\typm{\Gamma,\lockwith{\mu_F}}{V:A}{F'}$,
  \refa{}, and \Cref{lemma:structural-rules}.3, we have
  \[
    \typm{\Gamma,\lockwith{\nu_F}}{V:A}{E}.
  \]
  Then, by context equivalence, \Cref{lemma:structural-rules}.1, and
  \Cref{lemma:structural-rules}.4, we have
  \[
    \typm{\Gamma,\Gamma'}{V:A}{E}.
  \]
  \item[Case] $A$ has kind $\Pure$. By
  $\typm{\Gamma,\lockwith{\mu_F}}{V:A}{F'}$ and
  \Cref{lemma:pure-promotion}, we have
  \[
    \typm{\Gamma,\Gamma'}{V:A}{E}.
  \]
\end{description}
\item[Case]
\begin{mathpar}
\inferrule*[Lab=\tylab{Mod}]
{
  \mind{\mu'}{E} : F_1 \to E \\
  \refa{\typm{\Gamma,x\varb{\mu_F}{A},\Gamma',\lockwith{\mind{\mu'}{E}}}{W:B}{F_1}}
}
{\typm{\Gamma,x\varb{\mu_F}{A},\Gamma'}{\Box_{\mu'}\,W : \boxwith{\mu'} B}{E}}
\end{mathpar}
By IH on \refa{} we have
\[
  \typm{\Gamma,\Gamma',\lockwith{\mind{\mu'}{E}}}{W[V/x]:B}{F_1}.
\]
Then by \tylab{Mod} we have
\[
  \typm{\Gamma,\Gamma'}{(\Box_{\mu'}\,W)[V/x] : \boxwith{\mu'} B}{E}
\]
\item[Case] \tylab{Letmod}, \tylab{TAbs}, \tylab{TApp}, \tylab{Abs},
\tylab{App}, \tylab{Do}, \tylab{LocalEffect}, \tylab{Handle}.
By IHs.
\end{description}
\end{proof}

\subsection{Progress}
\label{app:progress}

\progress*
\begin{proof}
By induction on the typing derivation
$\typm{\Instctx\mid\cdot}{M:A}{E}$.
\begin{description}
\item[Case] $M$ is in a value normal form $U$. Trivial. Base case.
\item[Case] \tylab{Mod}. $\Box_\mu\, V$. By IH on $V$.
\item[Case] \tylab{Letmod}. $\Letm{\nu}{\mu} x = V \In N$. By IH on
$V$, if $V$ is reducible then $M$ is reducible; otherwise, $V$ is in a
value normal form, then by \Cref{lemma:canonical-forms} we have that
$M$ is reducible by \semlab{Letmod}.
\item[Case] \tylab{TApp}. $M\,A$. Similarly by IH on $M$, \Cref{lemma:canonical-forms}, and \semlab{TApp}.
\item[Case] \tylab{App}. $M\,N$. Similarly by IH on $M$ and $N$, \Cref{lemma:canonical-forms}, and \semlab{App}.
\item[Case] \tylab{Do}. $\Do\ell\;M$. We have $\ell\subtype E$. Either
$M$ is reducible or the whole term is in a normal form with respect to
$E$.
\item[Case] \tylab{LocalEffect}. Reducible by \semlab{Gen}.
\item[Case] \tylab{Handle}. $\Instctx\mid\cdot \vdash \Handle^\mu\;M\With \{\Ret x \mapsto N,\ell\;p\;r\mapsto N'\} \atmode{E}$.
\begin{description}
  \item[Case] $M$ is reducible. Trivial.
  \item[Case] $M$ is a value. By \semlab{Ret}.
  \item[Case] $M = \EC[\Do\ell'\;U]$.
  Since $M$ is not reducible itself, there is no handler for $\ell'$
  in $\EC$.
  If $\ell = \ell'$, by \semlab{Op}.
  Otherwise, since $M$ is at effect context $\ell,E$, by a
  straightforward induction on the evaluation context $\EC$ we know
  that $\Do\ell'\;U$ has the effect context $D,\ell,E'$ where
  $\ell'\notin D$ and $E'\subtype E$.
  The condition $\ell'\notin D$ is guaranteed by the fact that all
  labels in $D$ are introduced by handlers in $\EC$ and there is no
  handler for $\ell'$ in $\EC$.
  The condition $E'\subtype E$ is guaranteed by the fact that
  modality-parameterised handlers may only shrink the effect context,
  reading bottom-up, with their modality annotations.
  Then by inversion on the typing judgement of $\Do\ell'\;U$ we have
  $\ell'\subtype D,\ell,E'$.
  By \Cref{def:sanity-conditions}, $\ell'\notin D$ and $\ell'\neq\ell$
  we have $\ell'\subtype E$. Thus the whole term is in a normal form
  with respect to $E$.
\end{description}
\end{description}
\end{proof}

\subsection{Subject Reduction}
\label{app:subject-reduction}

\subjectReduction*
\begin{proof}
By induction on the typing derivation
$\typm{\Instctx\mid\Gamma}{M:A}{E}$.
We write out the context $\Instctx$ for runtime labels when necessary.
\begin{description}
\item[Case] \tylab{Var}. Impossible as there is no further reduction.
\item[Case]
\begin{mathpar}
  \inferrule*[Lab=\tylab{Mod}]
  {
    \mind{\mu}{F} : E \to F \\
    \refa{\typm{\Gamma,\lockwith{\mind{\mu}{F}}}{V:A}{E}}
  }
  {\typm{\Gamma}{\Box_\mu\,V : \boxwith{\mu} A}{F}}
\end{mathpar}
The only way to reduce is by \semlab{Lift} and $V\reducesto W$. IH on
\refa{} gives
\[
  {\typm{\Gamma,\lockwith{\mu_F}}{W:A}{E}}.
\]
Then by \tylab{Mod} we have
\[
  {\typm{\Gamma}{\Box_\mu\,W : \boxwith{\mu} A}{F}}.
\]
\item[Case]
\begin{mathpar}
\inferrule*[Lab=\tylab{Letmod}]
{
  \nu_F : E\to F \\
  \refa{\typm{\Gamma,\lockwith{\nu_F}}{V : \boxwith{\mu} A}{E}} \\
  \refb{\typm{\Gamma,x\varb{\nu_F\circ\mu_E}{A}}{M:B}{F}}
}
{\typm{\Gamma}{\Letm{\nu}{\mu} x = V \In M : B}{F}}
\end{mathpar}
By case analysis on the reduction.
\begin{description}
  \item[Case] \semlab{Lift} with $V\reducesto W$. By IH on \refa{} and
  reapplying \tylab{Letmod}.
  \item[Case] \semlab{Letmod}. We have $V = \Box_\mu\,U$ and
  \[
    \Letm{\nu}{\mu} x = \Box_\mu\,U \In M \reducesto M[U/x].
  \]
  Inversion on \refa{} gives
  \[
    {\typm{\Gamma,\lockwith{\nu_F},\lockwith{\mu_E}}{U : A}{E'}}.
  \]
  where $\mu_E : E' \to E$.
  By context equivalence, we have
  \[
    {\typm{\Gamma,\lockwith{\nu_F\circ\mu_E}}{U : A}{E'}}
  \]
  where $\nu_F\circ\mu_E : E'\to F$. By \Cref{lemma:substitution}.3 and \refb{}, we have
  \[
    {\typm{\Gamma}{M[U/x] : B}{F}}.
  \]
\end{description}
\item[Case] \tylab{TAbs},\tylab{Abs}. Impossible as there is no further reduction.
\item[Case]
\begin{mathpar}
\inferrule*[Lab=\tylab{TApp}]
{
\refa{\typm{\Gamma}{M:\forall\alpha^K.B}{E}} \\
\refb{\Gamma\vdash A : K} \\
}
{\typm{\Gamma}{M\,A : B[A/\alpha]}{E}}
\end{mathpar}
By case analysis on the reduction.
\begin{description}
  \item[Case] \semlab{Lift} with $M\reducesto N$. By IH on \refa{}
  and reapplying \tylab{TApp}.
  \item[Case] \semlab{TApp}. We have $M = \Lambda\alpha^K.V$ and
  \[
    (\Lambda\alpha^K.V)\,A\reducesto V[A/\alpha].
  \]
  Inversion on \refa{} gives
  \[
    \typm{\Gamma,\alpha:K}{V:B}{E}.
  \]
  Then by \Cref{lemma:substitution}.2 on \refb{}, we have
  \[
    \typm{\Gamma}{V[A/\alpha]:B[A/\alpha]}{E}.
  \]
\end{description}
\item[Case]
\begin{mathpar}
\inferrule*[Lab=\tylab{App}]
{
  \refa{\typm{\Gamma}{M : A \to B}{E}} \\
  \refb{\typm{\Gamma}{N : A}{E}}
}
{\typm{\Gamma}{M\; N: B}{E}}
\end{mathpar}
By case analysis on the reduction.
\begin{description}
  \item[Case] \semlab{Lift} with $M\reducesto M'$. By IH on \refa{}
  and reapplying \tylab{App}.
  \item[Case] \semlab{Lift} with $N\reducesto N'$. By IH on \refb{}
  and reapplying \tylab{App}.
  \item[Case] \semlab{App}. We have $M = \lambda x^A.M'$, $N=U$, and
  \[
    M\,N\reducesto M'[U/x].
  \]
  Inversion on \refa{} gives
  \[
    \typm{\Gamma,x:A}{M':B}{E}.
  \]
  Then by \Cref{lemma:substitution}.3 we have
  \[
    \typm{\Gamma}{M'[U/x]:B}{E}.
  \]
\end{description}
\item[Case] \tylab{Do}. The only way to reduce is by \semlab{Lift}.
Follow from IH and reapplying \tylab{Do}.
\item[Case]
\begin{mathpar}
\inferrule*[Lab=\tylab{LocalEffect}]
{
  \refa{\typm{\Instctx \mid \Gamma,\ell:A\sto B}{M : A'}{E}} \\
}
{\typm{\Instctx \mid \Gamma}{\Localeffect{\ell:A\sto B} M : A'}{E}}
\end{mathpar}
The only way to reduce is by \semlab{Gen}.
By \refa{}, \Cref{lemma:substitution}.4, and context weakening we have
\[
  \typm{\Instctx, \ell':A\sto B \mid \Gamma}{M[\ell'/\ell] : A'}{E}
\]
Note that we do not need to substitute the label $\ell$ in $A'$ and
$E$ since $\ell$ cannot appear in them.
Note that $\ell$ cannot appear in $A'$ and $E$.
\item[Case]
\begin{mathpar}
{
\inferrule*[Lab=\tylab{Handle}]
{
  H = \{\Ret x \mapsto N, \ell\;p\;r \mapsto N' \} \\
  \Sigma,\Instctx \ni \ell : A'\sto B' \\
  \mu(F) = E \\
  \Instctx\mid\Gamma \vdash \mu \To \aid \atmode{F} \\
  \Instctx\mid\Gamma \vdash \mu \To \mu\circ\mu \atmode{F} \\
  \refa{\typm{\Instctx\mid\Gamma, \lockwith{\mu_F}, \lockwith{\aex{\ell}_{E}}}{M : A}{\ell,E}} \\\\
  \refb{\typm{\Instctx\mid\Gamma, \lockwith{\mu_F}, x : \boxwith{(\mu\circ\aex{\ell})} A}{N : B}{E}} \\
  \refc{\typm{\Instctx\mid\Gamma, \lockwith{\mu_F}, p : A', r: \boxwith{\mu}(B' \to B)}{N' : B}{E}} \\
}
{\typm{\Instctx\mid\Gamma}{\Handle^\mu\;M\With H : B}{F}}
}
\end{mathpar}
By case analysis on the reduction.
\begin{description}
  \item[Case] \semlab{Lift} with $M\reducesto M'$. By IHs and
  reapplying \tylab{Handle}.
  \item[Case] \semlab{Ret}. We have $M=U$ and
  \[
    \Handle\;U \With H \reducesto N[(\Mod_{\mu\circ\aex{\ell}}\;U)/x].
  \]
  By \refa{} and \tylab{Mod}, we have
  \[
  \typm{\Gamma}{\Box_{\mu\circ\aex{\ell}}\;U:\boxwith{(\mu\circ\aex{\ell})}A}{F}
  \]
  By \refb{}, $\mu\To\aid\atmode{F}$, and \Cref{lemma:structural-rules}.3, we have
  \[
    {\typm{\Gamma, x : \boxwith{(\mu\circ\aex{\ell})} A}{N : B}{F}} \\
  \]
  Then by \Cref{lemma:substitution}.3 we have
  \[
    \typm{\Gamma}{N[(\Box_{\mu\circ\aex{\ell}}\;U)/x]:B}{F}
  \]
  \item[Case] \semlab{Op}.
  We have $M = \EC[\Do\ell\;U]$ and
  \[
    \Handle^\mu\;M\With H \reducesto
    N'[U/p, (\Box_\mu\,(\lambda y.\Handle^\mu\; \EC[y] \With H))/r]
  \]
  By \refc{}, $\mu\To\aid\atmode{F}$, and \Cref{lemma:structural-rules}.3, we have
  \[
    \refd{\typm{\Gamma, p : A', r: \boxwith{\mu}(B' \to B)}{N' : B}{F}}
  \]
  The label $\ell$ must be either in the global context $\Sigma$ or a
  runtime label in $\Instctx$ introduced by \semlab{Gen} before.
  We know that $U$ must have type $A'$ since there is a unique entry
  $\ell : A'\sto B'$ in $\Sigma,\Instctx$.
  By \refa{}, a straightforward induction on $\EC$, and inversion on
  the typing judgement of $\Do\ell\;U$, we have
  \[
    \typm{\Gamma,\lockwith{\mu_F}, \lockwith{\aex{\ell}_E}}{U:A'}{E'}
  \]
  for some $E'$.
  By $A' : \Pure$ and \Cref{lemma:pure-promotion}, we have
  \[
    \refe{\typm{\Gamma}{U:A'}{F}}.
  \]
  By \refa{}, $\mu\To\mu\circ\mu\atmode{F}$, and
  \Cref{lemma:structural-rules}.3, supposing $\mu(E) = E'$, we have
  \[
    \reff{\typm{\Gamma, \lockwith{\mu_F}, \lockwith{\mu_E}, \lockwith{\aex{\ell}_{E'}}}{M : A}{\ell,E'}}
  \]
  Observe that $B':\Pure$ allows $y$ to be accessed in any context.
  By \reff{} and a straightforward induction on $\EC$ we have
  \[
    \typm{\Gamma, \lockwith{\mu_F}, y:B',\lockwith{\mu_E}, \lockwith{\aex{\ell}_{E'}}}{
      \EC[y] : A}{\ell,E'}
  \]
  Then by \tylab{Handle} we have
  \[
    \typm{\Gamma, \lockwith{\mu_F}, y : B'}{
      \Handle^\mu\;\EC[y] \With H : A}{E}
  \]
  Note that we need to use $\mu\To\mu\circ\mu\atmode{F}$ and
  \Cref{lemma:structural-rules}.3 to duplicate the lock
  $\lockwith{\mu_F}$ in \refb{} and \refc{} for the handler $H$.
  Then by \tylab{Abs} and \tylab{Mod} we have
  \[
    \refg{\typm{\Gamma}{
      \Mod_\mu\;(\lambda y^{B'} . \Handle^\mu\;\EC[y] \With H) : \boxwith{\mu}(B'\to A)}{F}}
  \]
  By \refd{}, \refe{}, \refg{}, and \Cref{lemma:substitution}.3 we have
  \[
    \typm{\Gamma}{
      N'[U/p, (\Box_\mu\,(\lambda y^{B'}.\Handle^\mu\; \EC[y] \With H))/r]:B}{F}
  \]
\end{description}
\end{description}
\end{proof}

  \FloatBarrier
  \section{Source Calculi and Encodings}
\label{app:full-spec-others}

In this section, we provide the typing rules and operational semantics
of \Feps and \SystemC that are omitted in \Cref{sec:feps,sec:systemc}.
We also provide the translations of runtime constructs used in their
operational semantics.
Furthermore, we provide the specification of
\SystemXi~\citep{BrachthauserSO20} and its encoding in \Metp{\mcS},
and the specification of \Fepssn~\citep{XieCIL22} and its encoding in
\Metp{\emtScp}.

\subsection{Typing Rules and Operational Semantics of \Feps}
\label{app:semantics-feps}

\Cref{fig:typing-feps} gives the full typing rules of \Feps.
\Cref{fig:semantics-feps} gives the operational semantics of \Feps
including the definitions of runtime constructs and evaluation
contexts.
As we have mentioned in \Cref{sec:feps}, they are pretty standard.

Typing rules for runtime constructs are as follows.
\begin{mathpar}
\inferrule*[Lab=\tylab{Handle}]
{
  \Sigma \ni \ell : A'\sto B' \\\\
  \typf{\Gamma
  }{M}{A}{\ell,E} \\
  \typf{\Gamma, p : A', r: \earr{B'}{A}{E}}{N}{A}{E}  \\
}
{\typf{\Gamma}{\Handle\;M\With \{\ell\; p\;r \mapsto N \}}{A}{E}}
\end{mathpar}

Translations of runtime constructs are as follows.
They are used in the proof of semantics preservation in
\Cref{app:proof-feps}.
\[\ba{rcl}
  \transl{-} &:& \text{Computation} \to \text{Term} \\
  \tr{\Handle\;M\With H : A\mid E} &=& \Handle^{\aeq{\tr{E}}}\; \tr{M} \With \tr{H}
\ea\]
Translations of evaluation contexts are analogous to the translations
of their corresponding terms.

\begin{figure}[tb]\small
\raggedright
\boxed{\typf{\Gamma}{V}{A}{}}
\hfill
\begin{mathpar}
\inferrule*[Lab=\tylab{Unit}]
{ }
{\typf{\Gamma}{\Unit}{\TUnit}{}}

\inferrule*[Lab=\tylab{Var}]
{
  \Gamma \ni x : A
}
{\typf{\Gamma}{x}{A}{}}

\inferrule*[Lab=\tylab{Abs}]
{
  \typf{\Gamma,
  x:A}{M}{B}{E}
}
{\typf{\Gamma}{\lambda^E x^A.M}{\earr{A}{B}{E}}{}}

\inferrule*[Lab=\tylab{TAbs}]
{
\typf{\Gamma,\alpha:K}{V}{A}{}
}
{\typf{\Gamma}{\Lambda\alpha^K.V}{A}{}}

\inferrule*[Lab=\tylab{Handler}]
{
  \Sigma \ni \ell : A'\sto B' \\
  \typf{\Gamma,
  p : A', r: \earr{B'}{A}{E}}{N}{A}{E} \\
}
{\typf{\Gamma}{\Handler\; \{ \ell\;p\;r \mapsto N \}}{\earr{(\earr{\TUnit}{A}{\ell,E})}{A}{E}}{}}
\end{mathpar}
\raggedright
\boxed{\typf{\Gamma}{M}{A}{E}}
\hfill
\begin{mathpar}
\inferrule*[Lab=\tylab{Value}]
{
  \typf{\Gamma}{V}{A}{}
}
{\typf{\Gamma}{\Ret V}{A}{E}}

\inferrule*[Lab=\tylab{Let}]
{
  \typf{\Gamma}{M}{A}{E} \\
  \typf{\Gamma,x:A}{N}{B}{E}
}
{\typf{\Gamma}{\Let x = M \In N}{B}{E}}

\inferrule*[Lab=\tylab{App}]
{
  \typf{\Gamma}{V}{\earr{A}{B}{E}}{} \\
  \typf{\Gamma}{W}{A}{}
}
{\typf{\Gamma}{V\;W}{B}{E}}

\inferrule*[Lab=\tylab{TApp}]
{
  \typf{\Gamma}{V}{\forall\alpha^K.B}{} \\
  \Gamma\vdash A : K
}
{\typf{\Gamma}{V\;A}{B[A/\alpha]}{E}}

\inferrule*[Lab=\tylab{Do}]
{
  \Sigma \ni \ell : A\sto B \\
  \typf{\Gamma}{V}{A}{} \\
}
{\typf{\Gamma}{\Do \ell\;V}{B}{\ell,E}}
\end{mathpar}
\caption{Typing rules of \Feps.}
\label{fig:typing-feps}
\end{figure}

\begin{figure}[tb]\small
\begin{syntax}
\slab{Computations} & M & ::=
  & \cdots \mid \Handle\; M \With H \\
\slab{Evaluation Contexts} &  \EC &::= & [~]
  \mid \Let x = \EC \In N \mid \Handle\; \EC \With H
\end{syntax}

\begin{nreductions}
\semlab{TApp}   & (\Lambda \alpha^K.V)\,T &\reducesto& V[T/\alpha] \\
\semlab{App}   & (\lambda x^A.M)\,V &\reducesto& M[V/x] \\
\semlab{Handler}\hspace{-3em} &\Handler\;H\;V
  &\reducesto&
  \Handle\; V\;\Unit \With H,\\
\semlab{Ret} &
  \Handle\; (\Ret V) \With H &\reducesto& \Ret V
\\
\semlab{Op} &
  \Handle\; \EC[\Do\ell \; V] \With H
    &\reducesto& N[V/p, (\lambda y.\Handle\; \EC[\Ret y] \With H)/r],\\
\multicolumn{4}{@{}r@{}}{
      \text{ where } \ell\notin\BL{\EC} \text{ and }  H \ni (\ell \; p \; r \mapsto N)
} \\
\semlab{Lift} &
  \EC[M] &\reducesto& \EC[N],  \hfill\text{if } M \reducesto N \\
\end{nreductions}
\caption{Operational semantics of \Feps.}
\label{fig:semantics-feps}
\end{figure}

\subsection{Operational Semantics of \SystemC}
\label{app:semantics-cap}

\begin{figure}[tb]\small
\begin{syntax}
\slab{Runtime Labels} & \ell \\
\slab{Runtime Contexts} & \Instctx &::=& \cdot \mid \Instctx,\ell:\carrsingle{A}{}{B} \\
\slab{Capability Sets} & C &::=& \cdot \mid \{\ell\} \mid \{f\} \mid C\cup C' \\
\slab{Blocks} & P,Q & ::= & \cdots \mid \tmcap{\ell} \\
\slab{Computations} & M & ::=
  & \cdots \mid \Try_\ell\; M \With H
  \\
\slab{Evaluation Contexts} &  \EC &::= & [~]
  \mid \Let x = \EC \In N
  \mid \Def f = \EC \In N
  \mid \Try_\ell\; \EC \With H
  \\
\end{syntax}

\begin{nreductions}
\semlab{Box}  & \CUnbox (\CBox P) &\reducesto & P \\
\semlab{Let}  & \Let x = \Ret V \In N &\reducesto & N[V/x] \\
\semlab{Def}  & \Def f = P \In N &\reducesto & N[P/f] \\
\semlab{Call} & \block{\ol{x:A}}{\ol{f:T}}{M}(\ol{V},\ol{Q}) &\reducesto &
                M[\ol{V/x}, \ol{Q/f}, \ol{C/f}] \quad\text{where } \ol{\cdot\vdash Q:T\mid C}
                \\[1ex]
\semlab{Gen} &\Try\; \{f^{\carrsingle{A}{}{B}}\To M\} \With H \mid \Instctx
  &\reducesto &
  \Try_\ell\; M[\tmcap{\ell}/f, \{\ell\}/f] \With H \mid\Instctx, \ell : \carrsingle{A}{}{B} \\
\multicolumn{4}{@{}r@{}}{
  \text{where } \ell\text{ fresh}
} \\
\semlab{Ret} &
  \Try_\ell\; (\Ret V) \With H &\reducesto& \Ret V
\\
\semlab{Op} &
  \Try_\ell\; \EC[\tmcap{\ell}(V)] \With H
    &\reducesto& N[V/p,  \blocksingle{y}{}{\Try_\ell\; \EC[\Ret y] \With H}/r],\\
\multicolumn{4}{@{}r@{}}{
      \text{ where } \ell\notin\BL{\EC} \text{ and } H = \{p \; r \mapsto N\}
} \\
\semlab{Lift} &
  \EC[M] &\reducesto& \EC[N],  \hfill\text{if } M \reducesto N \\
\end{nreductions}
  \caption{Operational semantics and runtime constructs for \SystemC.}
  \label{fig:semantics-systemc}
\end{figure}

\Cref{fig:semantics-systemc} defines the operational semantics and
syntax of runtime constructs for \SystemC.
Reduction in \SystemC is defined not only on terms but also blocks
since we have $\CUnbox V$ which can reduce.
Since \SystemC uses block variables $f$ as both term-level and
type-level variables, we need to substitute them separately.
We follow \citet{BrachthauserSLB22} to overload the notion of
substitution.
We write $C/f$ for substituting in types and $P/f$ for substituting in
terms.

Typing rules for runtime constructs are as follows.
\begin{mathpar}\small
\inferrule*[Lab=\tylab{Cap}]
{
  \Instctx \ni \ell : \carrsingle{A}{}{B}
}
{\typf{\Instctx\mid\Gamma}{\tmcap{\ell}}{\carrsingle{A}{}{B}}{\{\ell\}}}

\inferrule*[Lab=\tylab{Handle}]
{
  \Instctx \ni \ell : \carrsingle{A'}{}{B'} \\
  \typxi{\Gamma
  }{M}{A}{C\cup\{\ell\}} \\\\
  \typxi{\Gamma,p:A',r:^C\carr{B'}{}{A}}{N}{A}{C} \\
}
{\typxi{\Instctx\mid\Gamma}{\Try_\ell\; M \With \{p\;r\mapsto N\}}{A}{C}}
\end{mathpar}

Translations of runtime syntax are as follows.
They are used for semantics preservation.
\begin{prog}
  \ba{rcl}
  \transl{-} &:& \text{Runtime Context} \to \text{Runtime Context}\\
  \tr{\cdot} &=& \cdot \\
  \tr{\Instctx, \ell:\carrsingle{A}{}{B}} &=& \tr{\Instctx}, \ell:\tr{A}\sto\tr{B} \\[2ex]
  \transl{-} &:& \text{Term} \to \text{Term}\\[.5ex]
  \tr{\tmcap{\ell}} &=& \Modhl{\aid}\;(\lambda x^{\tr{A}} . \Do \ell\;x)
    \quad\text{where } \Instctx\ni \ell : \carrsingle{A}{}{B} \\[.5ex]
  \tr{\Try_\ell\;M\With H {:A\mid C}} &=&
  \Handle^{\aeq{\tr{C}}}\;\tr{M}\With\tr{H^{{\ell,C}}} \\[.5ex]
  \transl{\{p\;r\mapsto N\}^{{\ell,C}}} &=&
    \{\bl
    \Ret x \mapsto {\Letmhl{}{\aeq{\ell,\tr{C}}} x'=x \In} x', \\[.5ex]
    \ell\;p\;r\mapsto {\Letmhl{}{\aeq{\transl{C}}} \hat{r} = r \In} \transl{N}
    \}\el
\ea
\end{prog}
Translations of evaluation contexts are analogous to the translations
of their corresponding terms.

\subsection{\SystemXi and its Encoding in \Metp{\mcS}}
\label{app:systemxi}

\SystemXi~\citep{BrachthauserSO20} is a fragment of \SystemC without
boxing and unboxing.
As a result, capabilities never appear in types.
\SystemXi actually does not even track capabilities in the typing
judgements.
There is no danger of capability leakage since capabilities are
second-class.
\Cref{fig:typing-systemxi} gives the syntax and typing rules of
\SystemXi.

\begin{figure}[tb] \small
\[\bs
  \slab{Value Types}\hspace{-1em} &A,B  &::= & \TUnit \\
  \slab{Block Types}\hspace{-1em} &T &::= & \carr{\ol{A}}{\ol{T}}{B} \\
  \slab{Contexts} &\Gamma  &::= & \cdot
                  \mid \Gamma, x: A
                  \mid \Gamma, f: T
                  \\
  \slab{Values}   &V,W  &::= & x \mid \Unit \\
  \slab{Blocks} &P,Q &::= & f \mid \block{\ol{x:A}}{\ol{f:T}}{M} \\
\es
\hfill
\bs
  \slab{Computations}\hspace{-.5em} &M,N  &::= & \Ret V \mid
                              \Let x = M \In N \\
                  & & \mid & \Def f = P \In N \mid P(\ol{V},\ol{Q})  \\
                  &     &\mid& \Try\;\{f^{\carrsingle{A}{}{B}}\To M\}\With H \\
  \slab{Handlers} &H &::= & \{p\;r\mapsto N\} \\
\es\]

\raggedright
\boxed{\typxi{\Gamma}{V}{A}{}}
\boxed{\typxi{\Gamma}{P}{T}{}}
\boxed{\typxi{\Gamma}{M}{A}{}}
\hfill
\begin{mathpar}
\inferrule*[Lab=\tylab{Unit}]
{ }
{\typxi{\Gamma}{\Unit}{\TUnit}{}}

\inferrule*[Lab=\tylab{Var}]
{
  \Gamma \ni x : A
}
{\typxi{\Gamma}{x}{A}{}}

\inferrule*[Lab=\tylab{BlockVar}]
{
  \Gamma \ni f : T
}
{\typxi{\Gamma}{f}{T}{}}

\inferrule*[Lab=\tylab{Block}]
{
  \typxi{\Gamma,\ol{x:A},\ol{f:T}}{M}{B}{}
}
{\typxi{\Gamma}{\block{\ol{x:A}}{\ol{f:T}}{M}}{\carr{\ol{A}}{\ol{T}}{B}}{}}

\inferrule*[Lab=\tylab{Value}]
{
  \typxi{\Gamma}{V}{A}{}
}
{\typxi{\Gamma}{\Ret V}{A}{}}

\inferrule*[Lab=\tylab{Let}]
{
  \typxi{\Gamma}{M}{A}{} \\
  \typxi{\Gamma,x:A}{N}{B}{} \\
}
{\typxi{\Gamma}{\Let x = M \In N}{B}{}}

\inferrule*[Lab=\tylab{Def}]
{
  \typxi{\Gamma}{P}{T}{} \\
  \typxi{\Gamma,f : T}{N}{B}{} \\
}
{\typxi{\Gamma}{\Def f = P \In N}{B}{}}

\inferrule*[Lab=\tylab{Call}]
{
  \typxi{\Gamma}{P}{\carr{\ol{A}_i}{\ol{T}_j}{B}}{} \\\\
  \ol{\typxi{\Gamma}{V_i}{A_i}{}} \\
  \ol{\typxi{\Gamma}{Q_j}{T_j}{}}
}
{\typxi{\Gamma}{P(\ol{V}_i,\ol{Q}_j)}{B}{}}

\inferrule*[Lab=\tylab{Handle}]
{
  \typxi{\Gamma,f:\carrsingle{A'}{}{B'}}{M}{A}{} \\
  \typxi{\Gamma,p:A',r:\carrsingle{B'}{}{A}}{N}{A}{} \\
}
{\typxi{\Gamma}{\Try\; \{f^{\carrsingle{A'}{}{B'}}\To M\} \With \{p\;r\mapsto N\}}{A}{}}
\end{mathpar}
\caption{Syntax and typing rules for \SystemXi.}
\label{fig:typing-systemxi}
\end{figure}

The operational semantics of \SystemXi is almost identical to that of
\SystemC except for removing the \semlab{Box} rule and substitutions
of capability sets $C/f$.

\Cref{fig:encoding-systemxi} gives the encoding of \SystemXi into
\Metp{\mcS}.
This encoding is straightforward and does not even use any modalities.
We mostly just translate the syntax of second-class blocks in
\SystemXi to first-class functions in \Metp{\mcS}.
For named handlers we introduce local labels.
We use the syntactic sugar in \Cref{sec:typing-metp} for $\Handle$
with no modality annotation.

\begin{figure}[htb]\small
\[\ba[t]{r@{\ \ }c@{\ \ }l}
  \transl{-} &:& \text{Type} \to \text{Type}\\
  \transl{\TUnit} &=& \TUnit \\
  \transl{\carr{\ol{A}}{\ol{T}}{B}} &=&
    \ol{\transl{A}}\to\ol{\transl{T}}\to\transl{B}
    \\[2ex]
  \transl{-} &:& \text{Context} \to \text{Context} \\
  \transl{\cdot} &=& \cdot \\
  \transl{\Gamma,x:A} &=& \transl{\Gamma},x:\transl{A} \\
  \transl{\Gamma,f:T} &=& \transl{\Gamma},f:\transl{T} \\[2ex]
\ea
\ba[t]{r@{\ \ }c@{\ \ }l}
  \transl{-} &:& \text{Value} \to \text{Term}\\
  \transl{x} &=& x \\
  \transl{\Unit} &=& \Unit \\[2ex]
  \transl{-} &:& \text{Block} \to \text{Term}\\
  \transl{f} &=& f \\
  \transl{\block{\ol{x:A}}{\ol{f:T}}{M}} &=&
    \lambda \ol{x^{\transl{A}}}\,\ol{f^{\transl{T}}} . \transl{M}
    \\[2ex]
\ea\]
\[\ba[t]{r@{\ \ }c@{\ \ }l}
  \transl{-} &:& \text{Computation} \to \text{Term}\\
  \transl{\Ret V} &=& \transl{V} \\
  \transl{\Let x=M \In N} &=& \Let x = \transl{M} \In \transl{N} \\
  \transl{\Def f = P \In N} &=&
    \Let f = \transl{P} \In \transl{N}
    \\
  \transl{P(\ol{V},\ol{Q})} &=&
    \transl{P}\;\ol{\transl{V}}\;\ol{\transl{Q}}
    \\
  \transl{\Try\;\{f^{\carrsingle{A}{}{B}}\To M\} \With H} &=&
    \bl
    \Localeffect{\ell_f : \tr{A}\to \tr{B}} \\
    \Handle\; (\lambda f . \transl{M})\; (\lambda x^{\tr{A}} . \Do\ell_f\;x) \With \tr{H^{\gray{f}}}
    \el \\
  \transl{\{p\;r\mapsto N\}^{\gray{f}}} &=&
    \{\Ret x \mapsto \Letm{}{\aex{\ell_f}} x' = x \In x', \ell_f\;p\;r\mapsto \transl{N}\}
    \\
\ea\]
\caption{An encoding of \SystemXi in \Metp{\mcS}.}
\label{fig:encoding-systemxi}
\end{figure}

We translate runtime constructs used in the operational semantics as follows.
\begin{prog}
  \ba{rcl}
  \tr{\tmcap{\ell}} &=& \lambda x^{\tr{A}} . \Do\ell\;x
    \quad\text{where } \Instctx\ni \ell : \carrsingle{A}{}{B} \\[.5ex]
  \tr{\Try_\ell\;M\With H} &=& \Handle\;\tr{M}\With\tr{H^\ell} \\[.5ex]
  \transl{\{p\;r\mapsto N\}^\ell} &=&
    \{\Ret x \mapsto \Letm{}{\aex{\ell}} x' = x \In x', \ell\;p\;r\mapsto \transl{N}\}
\ea
\end{prog}

We have the following theorems which we prove in \Cref{app:proof-systemxi}.
\begin{restatable}[Type Preservation]{theorem}{SystemXiToMetp}
  \label{lemma:type-preservation-systemxi-to-metn}
  If $\,\Gamma \vdash M : A $ in \SystemXi, then $\transl{\Gamma} \vdash
  \transl{M} : \transl{A} \atmode{\cdot}$ in \Metp{\mcS}.
  Similarly for values and blocks.
\end{restatable}
\vspace{-.5\baselineskip}
\begin{restatable}[Semantics Preservation]{theorem}{SystemXiToMetpSemantics}
  \label{lemma:sematics-preservation-systemxi-to-metn}
  If $M$ is well-typed and $M\mid \Instctx\reducesto N\mid {\Instctx'}$ in
  \SystemXi, then $\transl{M}\mid {\transl{\Instctx}}
  \reducesto^\ast \transl{N}\mid {\transl{\Instctx'}}$
  in \Metp{\mcS}.
\end{restatable}

\subsection{\Fepssn and its Encoding in \Metp{\mcS}}
\label{app:fepssn}

\Cref{fig:syntax-fepssn} gives the syntax of \Fepssn. Our presentation
of \Fepssn is fine-grain call-by-value. We highlight new parts
compared to \Feps.
Each effect label $\ell$ is annotated with a scope variable $a$.
As before, we omit kinds when obvious from alphabets.

\begin{figure}[htbp]\small
\begin{syntax}
  \slab{Value Types} &A,B  &::= & \TUnit \mid \alpha
                    \mid \earr{A}{B}{E} \mid
                    \forall \alpha^K.A \mid \hl{\tyev{\ell}{a}} \\
  \slab{Kind}     &K &::= & \Value \mid \Effect \mid \hl{\Scope(\ell)} \\
  \slab{Effect Rows}  &E &::= & \cdot \mid \evar \mid \hl{\ell^a}, E \\
  \slab{Contexts} &\Gamma  &::= & \cdot
                  \mid \Gamma,x : A
                  \mid \Gamma,\alpha : K
                  \\
  \slab{Values}   &V,W  &::= & \Unit \mid x \mid \elambda{E} x^A . M
                        \mid \Lambda \alpha^K . V
                        \mid \hl{\NHandler\;H}
                        \\
  \slab{Computations} &M,N  &::= &
                        \Ret V
                        \mid V\;W \mid V\;A
                      \mid  \Let x = M \In N
                      \mid \hl{\Do V\;W} \\
  \slab{Handlers} &H &::= & %
                            \{\ell\;p\;r \mapsto M\} \\
  \slab{Label Contexts} &\Sigma &::= & \cdot \mid \Sigma,\ell:A\sto B \\
\end{syntax}
\caption{Syntax of \Fepssn.}
\label{fig:syntax-fepssn}
\end{figure}

Different from \citet{XieCIL22}, the kind $\Scope(\ell)$ of scope
variables is annotated with an effect label $\ell$.
For $a : \Scope(\ell)$, each appearance of $a$ in effect rows must be
associated with this label $\ell$ as $\ell^a$.
This annotation is important to rule out well-typed but meaningless
terms in \Fepssn.
Not every well-typed term in \Fepssn is meaningful in the sense that we
can find an appropriate evaluation context to fully apply its
abstractions and handle its effects.
For example, a function of type $\forall a . \tyev{\ell_1}{a} \times
\tyev{\ell_2}{a}\to^{\ell_1^a,\ell_2^a}\TUnit$ cannot be applied and
handled when $\ell_1 \neq \ell_2$, because named handlers cannot
provide two evidences values of types $\tyev{\ell_1}{a}$ and
$\tyev{\ell_2}{a}$ with the same scope variable $a$ but different
operations $\ell_1$ and $\ell_2$. (This type becomes meaningful with
umbrella effects in \citet{XieCIL22}.)
Each handler introduces its own scope variable $a$ with some fixed
operation label $\ell$.
Annotating the kind $\Scope$ with an operation label $\ell$ solves the
problem of attaching the same scope variable $a$ to different
operation labels.

\Cref{fig:typing-fepssn} gives the typing rules.
Rule \tylab{NamedHandler} introduces a named handler as a function,
whose argument takes a handler name of the evidence type
$\tyev{\ell}{a}$.
An evidence type specifies an effect $\ell$ and a scope variable $a$.
The use of rank-2 polymorphism guarantees that the handler name of
type $\tyev{\ell}{a}$ cannot escape the scope of the handler.
Rule \tylab{DoName} invokes an operation via a handler name $M$ of the
evidence type $\tyev{\ell}{a}$.

\begin{figure}[htbp]\small
\raggedright
\boxed{\typf{\Gamma}{V}{A}{}}
\hfill
\begin{mathpar}
\inferrule*[Lab=\tylab{Unit}]
{ }
{\typf{\Gamma}{\Unit}{\TUnit}{}}

\inferrule*[Lab=\tylab{Var}]
{
  \Gamma \ni x : A
}
{\typf{\Gamma}{x}{A}{}}

\inferrule*[Lab=\tylab{Abs}]
{
  \typf{\Gamma,
  x:A}{M}{B}{E}
}
{\typf{\Gamma}{\lambda^E x^A.M}{\earr{A}{B}{E}}{}}

\inferrule*[Lab=\tylab{TAbs}]
{
\typf{\Gamma,\alpha:K}{V}{A}{}
}
{\typf{\Gamma}{\Lambda\alpha^K.V}{A}{}}

\inferrule*[Lab=\tylab{NamedHandler}]
{
  \Sigma \ni \ell : A'\sto B' \\
  \typf{\Gamma,
  p : A', r: \earr{B'}{A}{E}}{N}{A}{E} \\
}
{\typf{\Gamma}{\text{$\NHandler$}\; \{ \ell\;p\;r \mapsto N \}}{\earr{(\forall a^{\Scope(\ell)} . \earr{\tyev{\ell}{a}}{A}{\ell^a,E})}{A}{E}}{}}
\end{mathpar}
\raggedright
\boxed{\typf{\Gamma}{M}{A}{E}}
\hfill
\begin{mathpar}
\inferrule*[Lab=\tylab{Value}]
{
  \typf{\Gamma}{V}{A}{}
}
{\typf{\Gamma}{\Ret V}{A}{E}}

\inferrule*[Lab=\tylab{App}]
{
  \typf{\Gamma}{V}{\earr{A}{B}{E}}{} \\
  \typf{\Gamma}{W}{A}{}
}
{\typf{\Gamma}{V\;W}{B}{E}}

\inferrule*[Lab=\tylab{TApp}]
{
  \typf{\Gamma}{V}{\forall\alpha^K.B}{} \\
  \Gamma\vdash A : K
}
{\typf{\Gamma}{V\;A}{B[A/\alpha]}{E}}

\inferrule*[Lab=\tylab{Let}]
{
  \typf{\Gamma}{M}{A}{E} \\
  \typf{\Gamma,x:A}{N}{B}{E}
}
{\typf{\Gamma}{\Let x = M \In N}{B}{E}}

\inferrule*[Lab=\tylab{DoName}]
{
  \Sigma \ni \ell : A\sto B \\
  \typf{\Gamma}{V}{\tyev{\ell}{a}}{} \\
  \typf{\Gamma}{W}{A}{} \\
}
{\typf{\Gamma}{\Do V\;W}{B}{\ell^a, E}}
\end{mathpar}
\caption{Typing rules for \Fepssn.}
\label{fig:typing-fepssn}
\end{figure}

\Cref{fig:semantics-fepssn} gives the operational semantics of \Fepssn
including definitions of runtime constructs and evaluation contexts.
Similar to the operational semantics of \Metpt and \SystemC, reduction
rules in \Fepssn are also of form $M\mid\Instctx \reducesto
N\mid\Instctx'$.
The most interesting rule is \semlab{Gen} which reduces a handler
application to a runtime $\Handle$ construct and passes a runtime
evidence value $\tmev{h}$ to the argument.
This rule generates a fresh marker $h$ and a fresh scope variable $a$.
The runtime context $\Instctx$ associates dynamically generated
markers $h$ to their operation label $\ell$ and dynamically generated
scope variable $a$.

Different from \citet{XieCIL22}, our \semlab{Op} rule for \Fepssn has
the condition $h\notin\meta{BH}(\EC)$ which makes sure there is no
other handler in $\EC$ with the marker $h$.
This condition is necessary to guarantee that the current handler is
the nearest one for the marker $h$, because even for named handlers it
is still possible to duplicate the same handler during the runtime as
observed by \citet{BiernackiPPS20}.

\begin{figure}[htbp]\small
\begin{syntax}
\slab{Runtime Markers} & h \\
\slab{Runtime Contexts} & \Instctx &::=& \cdot \mid \Instctx,h:\ell^a \\
\slab{Computations} & M & ::=
  & \cdots \mid \Handle_h\; M \With H
  \\
\slab{Values} & V & ::= & \cdots \mid \tmev{h} \\
\slab{Evaluation Contexts} &  \EC &::= & [~]
  \mid \Let x = \EC \In N \mid \Handle_h\; \EC \With H
  \\
\end{syntax}

\begin{nreductions}
\semlab{TApp}   & (\Lambda \alpha^K.V)\,A &\reducesto& V[A/\alpha] \\
\semlab{App}   & (\lambda x^A.M)\,V &\reducesto& M[V/x] \\
\semlab{Gen} &\NHandler\;H\;V \mid \Instctx
  &\reducesto &
  \Handle_h\; (\Let x = V\,a\In x\;\tmev{h}) \With H \mid \Instctx, h:\ell^a\\
\multicolumn{4}{@{}r@{}}{
  \text{where } a,h\text{ fresh and } H \ni (\ell \; p \; r \mapsto N)
} \\
\semlab{NRet} &
  \Handle_h\; (\Ret V) \With H &\reducesto& \Ret V
\\
\semlab{NOp}\hspace{-1em}  &
  \Handle_h\; \EC[\Do \tmev{h} \; V] \With H
    &\reducesto& N[V/p, (\lambda y.\Handle_h\; \EC[\Ret y] \With H)/r],\\
\multicolumn{4}{@{}r@{}}{
      \text{ where } \ell\notin\meta{BH}(\EC) \text{ and }  H \ni (\ell \; p \; r \mapsto N)
} \\
\semlab{Lift} &
  \EC[M] &\reducesto& \EC[N],  \hfill\text{if } M \reducesto N \\
\end{nreductions}
\caption{Operational semantics for \Fepssn.}
\label{fig:semantics-fepssn}
\end{figure}

Typing rules for runtime constructs are as follows.
{\rulesize
\begin{mathpar}
\inferrule*[Lab=\tylab{Evidence}]
{
  \Instctx \ni h : \ell^a
}
{\typf{\Instctx\mid\Gamma}{\tmev{h}}{\tyev{\ell}{a}}{}}

\inferrule*[Lab=\tylab{HandleName}]
{
  \Instctx \ni h : \ell^a \\
  \Sigma \ni \ell : A'\sto B' \\\\
  \typf{\Gamma
  }{M}{A}{\ell^a,E} \\
  \typf{\Gamma, p : A', r: \earr{B'}{A}{E}}{N}{A}{E}  \\
}
{\typf{\Instctx\mid\Gamma}{\Handle_{h}\;M\With \{\ell\; p\;r \mapsto N \}}{A}{E}}
\end{mathpar}
}

\Cref{fig:fepssn-to-metp} gives the translation of \Fepssn into \Metp{\mcS}.
The evidence type $\tyev{\ell}{a}$ with $\Sigma\ni\ell:A\sto B$ is
translated to a function type $\boxwith{\aeq{a}}(\tr{A}\to\tr{B})$.
Correspondingly, an operation invocation $\Do V\;W$ for
$V:\tyev{\ell}{a}$ is translated similarly to a function application.
A named handler $\NHandler\;H$ is translated to a function that takes
a function argument $f$ and handles it with a handler in \Metp{\mcS}.
We pass the term $\Mod_{\aeq{a}}\;(\lambda x.\Doy_a\;x)$ to the
argument $f$ to simulate the handler name of type $\tyev{\ell}{a}$ in
\Fepssn.

\begin{figure}[tb]
  \small
  \renewcommand{\arraystretch}{1.1}
\[\ba[t]{r@{\ \ }c@{\ \ }l}
  \transl{-} &:& \text{Kind} \to \text{Kind}\\
  \transl{\Value} &=& \Pure \\
  \transl{\Effect} &=& \Effect \\
  \transl{\Scope(\ell)} &=& \Effect \\[2ex]
  \transl{-} &:& \text{Type} \to \text{Type}\\
  \transl{\TUnit} &=& \TUnit \\
  \transl{\alpha} &=& \alpha \\
  \transl{\earr{A}{B}{E}} &=& \boxwith{\aeq{\transl{E}}}(\transl{A} \to \transl{B}) \\
  \transl{\forall\alpha^K.A} &=& \forall\alpha^{\tr{K}} . \transl{A} \\
  \transl{\tyev{\ell}{a}} &=& \boxwith{\aeq{a}}(\transl{A}\to\transl{B})
  \\
  & &\text{where } \Sigma \ni \ell : A\sto B
  \\
\ea
\ba[t]{r@{\ \ }c@{\ \ }l}
  \transl{-} &:& \text{Effect Row} \to \text{Effect Context}\\
  \transl{\cdot} &=& \cdot \\
  \transl{\evar} &=& \evar \\
  \transl{\ell^a,E} &=& a,\transl{E}
  \\[2ex]
  \stransl{-}{} &:& \text{Context} \to \text{Context}\\
  \stransl{\cdot}{} &=& \cdot \\
  \stransl{\Gamma,x:A}{} &=& \stransl{\Gamma}{},x:\transl{A} \\
  \stransl{\Gamma,\alpha:K}{} &=& \stransl{\Gamma}{},\alpha:\tr{K}
  \\[2ex]
  \transl{-} &:& \text{Label Context} \to \text{Label Context} \\
  \transl{\cdot} &=& \cdot \\
  \transl{\Sigma,\ell:A\sto B} &=& \transl{\Sigma}, \ell:\tr{A}\sto\tr{B}
\ea
\]

\[
\ba[t]{r@{\ \ }c@{\ \ }l}
  \transl{-} &:& \text{Value / Computation} \to \text{Term} \\
  \transl{\Unit} &=& \Unit \\
  \transl{x} &=& x \\
  \transl{\Lambda\alpha^K.V} &=& \Lambda\alpha^{\tr{K}}.\transl{V} \\
  \transl{\Ret V} &=& \tr{V} \\
  \transl{\Let x = M \In N} &=& \Let x = \tr{M} \In \tr{N} \\
  \transl{V\;A} &=& \transl{V}\;\transl{A} \\
  \transl{V^{A\to^E B}\;W} &=&
    \Letm{}{\aeq{\transl{E}}} x = \transl{V} \In x\;\transl{W} \\
  \transl{\Do V^{\tyev{\ell}{a}}\; W} &=&
    \Letm{}{\aeq{a}} x = \transl{V} \In x\;\transl{W}\\
  \Bigl\llbracket
    {\vcenter{\bl\NHandler\; H\\ %
    :\earr{(\forall a^\ell . \earr{\tyev{\ell}{a}}{A}{\ell^a,E})}{A}{E}\el}}
  \Bigr\rrbracket
  &=&
  \Mod_{\aeq{\transl{E}}}\;(
  \lambda f .
    \bl
    \Localeffect{\ell_a : \tr{A'}\sto \tr{B'}}\\
    \Handle^{\aeq{\tr{E}}}\;
      (\bl\Letm{}{\aeq{\ell_a,\transl{E}}} f' = f\;\ell_a \In \\
      f'\;(\Mod_{\aeq{\ell_a}}\;(\lambda x^{\tr{A'}} . \Do\ell_a\;x)))\el \\
    \!\!\With \tr{H^{\gray{a,E}}})\el \\
  \tr{H^{\gray{a,E}}} &=&
    \{\bl
      \Ret x \mapsto \Letm{}{\aeq{\ell_a,\transl{E}}} x = x \In x,\\
      \ell\;p\;r\mapsto \transl{N}
    \}\el
    \\
\ea\]
  \caption{An encoding of \Fepssn in \Metp{\mcS}.}
\label{fig:fepssn-to-metp}
\end{figure}

The following theorems show that our encoding preserves types and
operational semantics.

\begin{restatable}[Type Preservation]{theorem}{FepssnToMetp}
  \label{lemma:type-preservation-fepssn-to-metn}
  If $\,\typf{\Gamma}{M}{A}{E}$ in \Fepssn, then $\transl{\Gamma}
  \vdash \transl{M} : \transl{A} \atmode{\transl{E}}$ in
  \Metp{\mcS}.
  Similarly for typing judgements of values.
\end{restatable}

\begin{restatable}[Semantics Preservation]{lemma}{FepssnToMetpSemantics}
  \label{lemma:sematics-preservation-fepssn-to-metn}
  If $M$ is well-typed and $M \reducesto N$ in \Fepssn, then $\transl{M}
  \reducesto^\ast \transl{N}$ in \Metp{\mcS} where
  $\reducesto^\ast$ denotes the transitive closure of $\reducesto$.
\end{restatable}

Translations of runtime constructs for \Fepssn are as follows.
Note that for dynamically generated scope variables $a$, we translate
to their corresponding dynamically generated labels $\ell_a$ where
$\Instctx \ni h : \ell^a$.
The translations of terms do not use $h$ as the scope variable $a$ is
also dynamically generated and is enough to uniquely assign operations
to handlers.
The translations are used in the proof of semantics preservation in
\Cref{app:proof-fepssn}.
\begin{prog}
  \ba{rcl}
  \transl{-} &:& \text{Runtime Context} \to \text{Runtime Context}\\
  \tr{\cdot} &=& \cdot \\
  \tr{\Instctx, h:\ell^a} &=&
    \tr{\Instctx}, \ell_a:\tr{A}\sto\tr{B}
    \hfill\text{ where } \Sigma \ni \ell : A\sto B
     \\[2ex]
  \transl{-} &:& \text{Term} \to \text{Term}\\[.5ex]
  \tr{V\;a} &=& \tr{V}\;\ell_a \quad\text{where } \Instctx\ni h:\ell^a
    \\[.5ex]
  \tr{\tmev{h}} &=& \Mod_{\aeq{\ell_a}}\;(\lambda x^{\tr{A}} . \Do \ell_a\;x)
    \\
    & &\quad\text{where } \Instctx\ni h:\ell^a \text{ and }
    \Sigma\ni \ell : A\sto B
    \\[.5ex]
  \tr{\Handle_h\;M\With H \;\gray{:A\mid E}} &=&
    \Handle^{\aeq{\tr{E}}}\; \tr{M} \With \tr{H}
    \\
    & &\quad\text{where } \Instctx\ni h:\ell^a \text{ and }
    \Sigma \ni \ell : A\sto B\\
\ea
\end{prog}
Translations of evaluation contexts are analogous to the translations
of their corresponding terms.

  \FloatBarrier
  \section{Proofs of Encodings}
\label{app:proof-encoding}

We prove the type preservation and semantics preservation theorems in
\Cref{sec:encoding-caps,sec:encoding-rows}.
We also prove the type preservation and semantics preservation
theorems for some other encodings.

\subsection{Proofs of Encoding of \Feps in \Metp{\emtScp}}
\label{app:proof-feps}

\FepsToMetp*
\begin{proof}
By induction on typing judgements $\Gamma\vdash M : A\mid E$ in
\Feps.
Most cases follow from using IH trivially. We elaborate interesting cases.
When referring to the name of a rule, we sometimes also mention the
calculus name to disambiguate. For instance, \tylabi{Var}{\Feps}
refers to the rule \tylab{Var} of \Feps.
\begin{description}
  \item[Case] \focus{\Unit} By \tylabi{Unit}{\Feps} and \tylabi{Unit}{\Metp{\emtScp}}.
  \item[Case] \focus{x} All translated types have kind $\Pure$ in
    \Metp{\emtScp} and thus can always be accessed by \tylabi{Var}{\Metp{\emtScp}}.
  \item[Case] \focus{\Lambda\alpha^K.V} By IH, \tylabi{TAbs}{\Feps} and
  \tylabi{TAbs}{\Metp{\emtScp}}.
  \item[Case] \focus{V\;A} By IH, \tylabi{TApp}{\Feps} and \tylabi{TApp}{\Metp{\emtScp}}.
  \item[Case] \focus{\elambda{E} x^A.M}
    \begin{mathpar}
      \inferrule*[Lab=\tylab{Abs}]
      {
        \refa{\typf{\Gamma, x:A}{M}{B}{E}}
      }
      {\typf{\Gamma}{\lambda^E x^A.M}{\earr{A}{B}{E}}{}}
    \end{mathpar}
    By IH on \refa{}, \Cref{lemma:structural-rules},
    \Cref{lemma:pure-promotion-vars}, we have
    \[
      \typmet{\stransl{\Gamma}{}, \lockwith{\aeq{\transl{E}}}, x:\transl{A}}{\transl{M}}{\transl{B}}{\transl{E}}
    \]
    By \tylabi{Abs}{\Metp{\emtScp}} and \tylabi{Mod}{\Metp{\emtScp}}, we have
    \[
      \typmet{\stransl{\Gamma}{}}{\Mod_{\aeq{\transl{E}}}\;(\lambda x^{\transl{A}}.\transl{M})}{
        \boxwith{\aeq{\transl{E}}}(\transl{A}\to\transl{B})}{{\cdot}}
    \]
  \item[Case] \focus{\Ret V} By \Cref{lemma:feps-subeffecting}.
  \item[Case] \focus{\Let x = M \In N} By IH, syntactic sugar, and \tylabi{App}{\Metp{\emtScp}}.
  \item[Case] \focus{V\;W}
    \begin{mathpar}
      \inferrule*[Lab=\tylab{App}]
      {
        \refa{\typf{\Gamma}{V}{\earr{A}{B}{E}}{}} \\
        \refb{\typf{\Gamma}{W}{A}{}}
      }
      {\typf{\Gamma}{V\;W}{B}{E}}
    \end{mathpar}
    By IH on \refa{} and \Cref{lemma:feps-subeffecting}, we have
    \[
      \typmet{\transl{\Gamma}}{\transl{V}}{\boxwith{\aeq{\transl{E}}}(\transl{A}\to\transl{B})}{\transl{E}}
    \]
    By IH on \refb{} and \Cref{lemma:feps-subeffecting}, we have
    \[
      \typmet{\transl{\Gamma}}{\transl{W}}{\transl{A}}{\transl{E}}
    \]
    By \tylabi{Letmod}{\Metp{\emtScp}} and \tylabi{App}{\Metp{\emtScp}}, we have
    \[
      \typmet{\transl{\Gamma}}{\Letm{}{\aeq{\transl{E}}} x = \transl{V}\In x\;\transl{W}}{
        \transl{B}}{\transl{E}}
    \]
  \item[Case] \focus{\Do\ell\;V} By IH, \Cref{lemma:feps-subeffecting}, \tylabi{Do}{\Feps} and \tylabi{Do}{\Metp{\emtScp}}.
  \item[Case] \focus{\Handler\;\{\ell\;p\;r\mapsto N\}}
    \begin{mathpar}
      \inferrule*[Lab=\tylab{Handler}]
      {
        \Sigma \ni \ell : A'\sto B' \\
        \refa{\typf{\Gamma, p : A', r: \earr{B'}{A}{E}}{N}{A}{E}} \\
      }
      {\typf{\Gamma}{\Handler\; \{\ell\;p\;r \mapsto N \}}{\earr{(\earr{\TUnit}{A}{\ell,E})}{A}{E}}{}}
    \end{mathpar}
    By IH on \refa{}, \Cref{lemma:pure-promotion-vars}, and
    \Cref{lemma:structural-rules}, we have
    \[
      \refb{\typmet{\str{\Gamma}{},
      \lockwith{\aeq{\tr{E}}_{\cdot}},
      \lockwith{\aeq{\tr{E}}_{\tr{E}}}, p : \tr{A'}, r: \boxwith{\aeq{\tr{E}}}({\tr{B'}}\to {\tr{A}})}{\tr{N}}{\tr{A}}{\tr{E}}} \\
    \]

    By \tylabi{Letmod}{\Metp{\emtScp}}, \tylabi{Var}{\Metp{\emtScp}},
    and $\tr{A} : \Pure$, we have
    \[
      \refc{\typmet{
        \str{\Gamma}{},
        \lockwith{\aeq{\tr{E}}_{\cdot}},
        \lockwith{\aeq{\tr{E}}_{\tr{E}}},
        x:\boxwith{\aeq{\tr{\ell,E}}}\tr{A}}{
        \Letm{}{\aeq{\tr{\ell,E}}} x = x \In x
      }{\tr{A}}{\tr{E}}}
    \]
    By \tylabi{Var}{\Metp{\emtScp}}, \tylabi{Letmod}{\Metp{\emtScp}}, and \tylabi{App}{\Metp{\emtScp}}, we have
    \[
      \refd{\typmet{\str{\Gamma}{},
      \lockwith{\aeq{\tr{E}}}, f : \boxwith{\aeq{\tr{\ell,E}}}(\TUnit\to\tr{A}),
          \lockwith{\aeq{\tr{E}}},
          \lockwith{\aex{\ell}}
          }{
        \Letm{}{\aeq{\tr{\ell,E}}} f' = f \In f'\;\Unit}{\tr{A}}{\tr{\ell,E}}}
    \]
    By \tylabi{Handle}{\Metp{\emtScp}}, \refb{}, \refc{}, and \refd{}, we have
    \[\bl
      {\str{\Gamma}{},
      \lockwith{\aeq{\tr{E}}_{\cdot}}, f : \boxwith{\aeq{\tr{\ell,E}}}(\TUnit\to\tr{A})
      }\vdash \\[1ex]
      \quad {
        \Handle^{\aeq{\tr{E}}}\; (\Letm{}{\aeq{\tr{\ell,E}}} f' = f \In f'\;\Unit) \With \tr{H}
      }:{\tr{A}}\atmode{\tr{E}}
    \el\]
    Finally our final goal follows from
    \tylabi{Abs}{\Metp{\emtScp}} and \tylabi{Mod}{\Metp{\emtScp}}.
\end{description}
\end{proof}

The proof relies on the following lemma.

\begin{restatable}[Pure Values]{lemma}{FepsSubeffecting}
  \label{lemma:feps-subeffecting}
  Given a typing judgement $\typf{\Gamma}{V}{A}{}$ in \Feps,
  if $\typm{\stransl{\Gamma}{}}{\transl{V} : \transl{A}}{\cdot}$
  then $\typm{\stransl{\Gamma}{}}{\transl{V} :
  \transl{A}}{\transl{E}}$ for any $E$.
\end{restatable}
\begin{proof}
  By straightforward induction on typing judgements of values in
  \Feps.
\end{proof}

\FepsToMetpSemantics*
\begin{proof}
  By induction on $M$ and case analysis on the next reduction rule.
  Note that values in \Feps are translated to value normal forms in \Metp{\emtScp}.
  \begin{description}
    \item[Case] \focus{\semlab{App}} We have
    \[
    \tr{(\lambda^E x^A.M)\;V} = \Letm{}{\aeq{\tr{E}}} x =
    \Mod_{\aeq{\tr{E}}}\;(\lambda x^{\tr{A}}.\tr{M}) \In x\;\tr{V}
    \]
    Our goal follows from \semlab{Letmod} and \semlab{App} in \Metp{\emtScp}.
    It is obvious that translation preserves value substitution.
    \item[Case] \focus{\semlab{TApp}} By \semlab{TApp} in
    \Metp{\emtScp}. It is obvious that translation preserves type
    substitution.
    \item[Case] \focus{\semlab{Let}} By syntactic sugar and
    \semlab{App} in {\Metp{\emtScp}}.
    \item[Case] \focus{\semlab{Handler}}
    Suppose the effect row of the whole term is $E$.
    \[\ba{rcl}
    \Handler\;H\;V
      &\reducesto&
      \Handle\; V\;\Unit \With H
    \ea\]
    We have
    \[\ba{rcl}
    \tr{\LHS} &=&
    \Letm{}{\aeq{\tr{E}}} x = \tr{\Handler\;H} \In x\;\tr{V} \\
    \tr{\Handler\;H} &=&
    {\Mod_{\aeq{\transl{E}}}}\;(
    \lambda f .
      \Handle^{\aeq{\tr{E}}}\;
      (\Let {\Mod_{\aeq{\transl{\ell,E}}}}\; f' = f \In
        f'\;\Unit) \With \tr{H^{\gray{E}}})
    \ea\]
    By \semlab{Letmod} and \semlab{App} in \Metp{\emtScp}, $\tr{\LHS}$ reduces to
    \[
      \Handle^{\aeq{\tr{E}}}\;
      (\Let {\Mod_{\aeq{\transl{\ell,E}}}}\; f' = \tr{V} \In
        f'\;\Unit) \With \tr{H^{\gray{E}}}
    \]
    which is equal to $\tr{\RHS}$ of the above reduction step.
    \item[Case] \focus{\semlab{Ret}} By \semlab{Ret} and \semlab{Letmod} in \Metp{\emtScp}.
    \item[Case] \focus{\semlab{Op}}
    Suppose the effect row of the whole term is $E$.
    \[\ba{rcl}
      \Handle\; \EC[\Do \ell \; V] \With H
      &\reducesto& N[V/p, (\lambda y.\Handle\; \EC[\Ret y] \With H)/r]
    \ea\]
    where $\ell\notin\BL{\EC}$ and $H\ni \ell\;p\;r\mapsto N$.
    We have
    \[\ba{rcl}
    \tr{\LHS} &=& \Handle^{\aeq{\tr{E}}}\; \tr{\EC[\Do \ell\; V]} \With \tr{H}
    \ea\]
    By \Cref{lemma:ectrans-feffn}, we have
    \[\ba{rcl}
      \tr{\EC[\Do \ell \; V]} &=& \tr{\EC}[\Do\ell\;\tr{V}] \\
    \ea\]
    Then by \semlabi{Op}{\Metp{\emtScp}} and translation preserving substitution, $\tr{\LHS}$ reduces to
    \[
      \tr{N}[\tr{V}/p, (\Mod_{\aeq{\tr{E}}}\; (\lambda y . \Handle^{\aeq{\tr{E}}}\; \tr{\EC[y]} \With \tr{H})) / r]
    \]
    which is equal to $\tr{RHS}$ of the above reduction step.
  \item[Case] \focus{\semlab{Lift}} By IH and \Cref{lemma:ectrans-feffn}.
  \end{description}
\end{proof}

The proof of semantics preservation relies on the following lemma.

\begin{lemma}[Translation of Evaluation Contexts]
  \label{lemma:ectrans-feffn}
  For the translation $\tr{-}$ from \Feps to \Metp{\emtScp}, we have
  $\tr{\EC[M]} = \tr{\EC}[\tr{M}]$ for any
  evaluation context $\EC$ and term $M$.
\end{lemma}
\begin{proof}
  By straightforward case analysis on evaluation contexts of \Feps.
\end{proof}

\subsection{Proofs of Encoding \SystemC in \Metp{\mcS}}
\label{app:proof-systemc}

\SystemCToMetp*
\begin{proof}
  By induction on typing judgements in \SystemC.
  As a visual aid, for each non-trivial case we repeat its typing rule
  in \SystemC. We replace each premise by the \Metp{\mcS} judgement of
  the translated premise implied by the induction hypothesis. We
  replace the conclusion by the \Metp{\mcS} judgement of the
  translated conclusion that we need to prove.
  When referring to the name of a rule, we sometimes also mention the
  calculus name to disambiguate. For instance, \tylabi{Var}{\SystemC}
  refers to the rule \tylab{Var} of \SystemC.
  \begin{description}
    \item[Case] \focus{\Unit} By \tylabi{Unit}{\SystemC} and \tylabi{Unit}{\Metp{\mcS}}.
    \item[Case] \focus{x} By \tylabi{Var}{\SystemC} and
      \tylabi{Var}{\Metp{\mcS}}. Variables are always accessible after
      translation as translations of value types always have kind $\Pure$.
    \item[Case] \focus{\CBox P}
      \begin{mathpar}
        \inferrule*
        {
          \refa{\typmet{\str{\Gamma}{}}{\tr{P}}{\tr{T}}{\tr{C}}}
        }
        {\typmet{\str{\Gamma}{}}{
          \Mod_{\tr{C}}\;\tr{P}
        }{\tr{T\At C}}{\cdot}}
      \end{mathpar}
      We have $\tr{T\At C} = \boxwith{\aeq{\tr{C}}}\tr{T}$.
      By \refa{} and \Cref{lemma:pure-promotion-vars}, we have
      \[
        {\typmet{\str{\Gamma}{}, \lockwith{\aeq{\tr{C}}}}{\tr{P}}{\tr{T}}{\tr{C}}}
      \]
      Our goal follows from \tylabi{Mod}{\Metp{\mcS}}.
    \item[Case] \focus{f~\text{transparent}}
      \begin{mathpar}
        \inferrule*
        {
          \Gamma \ni f :^C T
        }
        {\typmet{\str{\Gamma}{}}{\hat{f}}{\tr{T}}{\tr{C}}}
      \end{mathpar}
      Suppose $\Gamma = \Gamma_1,f:^C T,\Gamma_2$.
      We have \[\str{\Gamma_1,f:^C T,\Gamma_2}{C}
        = \str{\Gamma_1}{}, f:\boxwith{\aeq{\tr{C}}}\tr{T}
        , \hat{f}\varb{\aeq{\tr{C}}}{\tr{T}}, \str{\Gamma_2}{}\]
      Our goal follows from \tylabi{Var}{\Metp{\mcS}} and \mtylab{Abs} ($\aeq{\tr{C}}\To\aid\atmode{\tr{C}}$).
    \item[Case] \focus{f~\text{tracked}}
      \begin{mathpar}
        \inferrule*
        {
          \Gamma \ni f :^\ast T
        }
        {\typmet{\str{\Gamma}{}}{\hat{f}}{\tr{T}}{\shat{f}}}
      \end{mathpar}
      Suppose $\Gamma = \Gamma_1,f:^\ast T,\Gamma_2$.
      We have \[
        \str{\Gamma_1, f:^\ast T, \Gamma_2}{} =
        \str{\Gamma_1}{},
          \shat{f},
          f:\boxwith{\aeq{\shat{f}}}\tr{T},
          \hat{f}\varb{\aeq{\shat{f}}}{\tr{T}},
          \str{\Gamma_2}{}
      \]
      Our goal follows from \tylabi{Var}{\Metp{\mcS}} and \mtylab{Abs} ($\aeq{\shat{f}}\To\aid\atmode{\shat{f}}$).
    \item[Case] \focus{\block{\ol{x:A}}{\ol{f:T}}{M}}
      \begin{mathpar}
        \inferrule*
        {
          \refa{\typmet{\str{\Gamma,{\ol{x:A},\ol{f:^\ast T}}}{}}{
            \tr{M}
          }{\tr{B}}{\tr{C \cup \{\ol{f}\} }}}
        }
        {\typmet{\str{\Gamma}{}}{
          \Lambda\ol{\shat{f}} .
          \Mod_{\aex{\ol{\shat{f}}}}(
          \lambda\ol{x^{\transl{A}}}\,\ol{f^{\boxwith{\aeq{\shat{f}}}\transl{T}}} .
          \ol{\Letm{}{\aeq{\shat{f}}}\hat{f} = f \In}
          \transl{M})
        }{\tr{\carr{\ol{A}}{\ol{f:T}}{B}}}{\tr{C}}}
      \end{mathpar}
      By \refa{}, \Cref{lemma:structural-rules}, and
      \Cref{lemma:pure-promotion-vars}, we have
      \[
        \typmet{\tr{\Gamma}, \ol{x:\tr{A}},
        \ol{\shat{f}},
        \lockwith{\aex{\ol{\shat{f}}}},
        \ol{f:\boxwith{\aeq{\shat{f}}}\tr{T}}, \ol{\hat{f}\varb{\aeq{\shat{f}}}{\tr{T}}}}{
          \tr{M}
        }{\tr{B}}{\tr{C},\ol{\shat{f}}}
      \]
      Our goal follows from \tylabi{Abs}{\Metp{\mcS}}, \tylabi{Letmod}{\Metp{\mcS}}, and \tylabi{Mod}{\Metp{\mcS}}.
    \item[Case] \focus{\CUnbox V}
      \begin{mathpar}
        \inferrule*
        {
          \refa{\typmet{\str{\Gamma}{}}{\tr{V}}{\tr{T\At C}}{\cdot}}
        }
        {\typmet{\str{\Gamma}{}}{\Letm{}{\aeq{\transl{C}}} x = \,\transl{V} \In x}{\tr{T}}{\tr{C}}}
      \end{mathpar}
      We have $\tr{T\At C} = \boxwith{\aeq{\tr{C}}}\tr{T}$.
      By \refa{} and \Cref{lemma:systemc-subeffecting}, we have
      \[
        {\typmet{\str{\Gamma}{}}{\tr{V}}{\boxwith{\aeq{\tr{C}}}\tr{T}}{\tr{C}}}
      \]
      Our goal follows from \tylabi{Letmod}{\Metp{\mcS}} and \mtylab{Abs} ($\aeq{\tr{C}}\To\aid\atmode{\tr{C}}$).
    \item[Case] \focus{\Let x = M \In N} By IH,
    \Cref{lemma:systemc-subeffecting}, \tylabi{Let}{\SystemC}, and
    \tylabi{Let}{\Metp{\mcS}}.
    \item[Case] \focus{\Def f = P \In M}
      \begin{mathpar}
        \inferrule*
        {
          \refa{\typmet{\str{\Gamma}{}}{\tr{P}}{\tr{T}}{\tr{C'}}} \\
          {\typmet{\str{\Gamma,f :^{C'} T}{}}{\tr{M}}{\tr{A}}{\tr{C}}} \\
        }
        {\typmet{\str{\Gamma}{}}{
          \Let f = \Mod_{\aeq{\transl{C'}}}\,\transl{P} \In
          \Letm{}{\aeq{\transl{C'}}} \hat{f} = f \In
          \transl{M}
        }{\tr{A}}{\tr{C}}}
      \end{mathpar}
      By \refa{} and \Cref{lemma:pure-promotion-vars}, we have
      \[
      \typmet{\tr{\Gamma}, \lockwith{\aeq{\tr{C'}}}}{\tr{P}}{\tr{T}}{\tr{C'}}
      \]
      Our goal follows from \tylabi{Mod}{\Metp{\mcS}},
      \tylabi{Let}{\Metp{\mcS}}, and \tylabi{Letmod}{\Metp{\mcS}}.
  \end{description}
  \item[Case] \focus{P(\ol{V_i}, \ol{Q_j})}
    \begin{mathpar}
      \inferrule*
      {
        \refa{\ol{\typmet{\str{\Gamma}{}}{\tr{V_i}}{\tr{A_i}}{\cdot}}} \\
        \refb{\ol{\typmet{\str{\Gamma}{}}{\tr{Q_j}}{\tr{T_j}}{\tr{C_j}}}} \\\\
        \typmet{\str{\Gamma}{}}{\tr{P}}{\tr{\carr{\ol{A_i}}{\ol{f_j:T_j}}{B}}}{\tr{C}} \\
        C' \coloneq C \cup \ol{C_j}
      }
      {\typmet{\str{\Gamma}{}}{
        \Letm{}{\aex{\ol{\transl{C_j}}}} x = \transl{P}\;\ol{\transl{C_j}}\In
        x\;\ol{\transl{V_i}}\;\ol{(\Mod_{\aeq{\transl{C_j}}}\,\transl{Q_j})}
      }{\tr{B}[\ol{\tr{C_j}/\shat{f}_j}]}{\tr{C'}}}
    \end{mathpar}
    We have $\tr{\carr{\ol{A_i}}{\ol{f_j:T_j}}{B}} =
    \forall\ol{\shat{f}_j} . \boxwith{\aex{\ol{\shat{f}_j}}}(
    \ol{\transl{A_i}}\to \ol{\boxwith{\aeq{\shat{f}_j}}\transl{T_j}}\to
    \transl{B})$.
    By \refa{} and \Cref{lemma:systemc-subeffecting} we have
    \[\bl
        \refc{\ol{\typmet{\str{\Gamma}{}}{\tr{V_i}}{\tr{A_i}}{C'}}} \\
    \el\]
    By \refb{} and \Cref{lemma:pure-promotion-vars}, we have
    \[\bl
        {\ol{\typmet{\str{\Gamma}{}, \lockwith{\aeq{\tr{C_j}}}}{
          \tr{Q_j}}{\tr{T_j}}{\tr{C_j}}}}
    \el\]
    Then by \tylabi{Mod}{\Metp{\mcS}} we have
    \[\bl
        \refd{\ol{\typmet{\str{\Gamma}{}}{\Mod_{\aeq{\tr{C_j}}}\;\tr{Q_j}}{\boxwith{\aeq{\tr{C_j}}}\tr{T_j}}{\tr{C'}}}}
    \el\]
    Also note that translation preserves type substitution of capability variables, which gives
    \[\refe{\tr{B}[\ol{\tr{C_j}/\shat{f}_j}] = \tr{B[\ol{C_j/f_j}]}}.\]
    Finally our goal follows from \refc{}, \refd{}, \refe{},
    \tylabi{Letmod}{\Metp{\mcS}}, \tylabi{App}{\Metp{\mcS}}, and \tylabi{TApp}{\Metp{\mcS}}.
  \item[Case] \focus{\Ret V} By IH.
  \item[Case] \focus{\text{subtyping of blocks and computations}} By IH and \Cref{lemma:systemc-subeffecting}.
  \item[Case] \focus{\Try\;\{f^{\carrsingle{A'}{}{B'}}\To M\}\With\{p\;r\mapsto N\}}
  \begin{mathpar}
      \inferrule*
      {
        \refa{\typmet{\str{
          \Gamma,
          {f:^\ast\carr{A'}{}{B'}}
          }{}}{
          \tr{M}}{\tr{A}}{\tr{C\cup\{f\}}}} \\
        \refb{\typmet{\str{\Gamma,p:A',r:^C\carr{B'}{}{A}}{C}}{\tr{N}}{\tr{A}}{\tr{C}}} \\
      }
      {\typmet{\str{\Gamma}{C}}{
        {\bl
        {\Localeffect{\ell_f : \tr{A'}\sto\tr{B'}}}
        \Letm{}{\aex{\ell_f}} g\; = M_1 \In \\[.5ex]
        \Handle^{\aeq{\tr{C}}}\; g\;M_2 \With \{\Ret x \mapsto N_1, \ell_f\;p\;r\mapsto N_2\} \el}
      }{\tr{A}}{\tr{C}}}
  \end{mathpar}
  where
  \[\ba{r@{~}c@{~}l}
  M_1 &=& ({\Lambda\shat{f}}.\Mod_{\aex{\shat{f}}}\;(
          \lambda f^{\boxwith{\aeq{\shat{f}}}\boxwith{\aid}(\tr{A'}\to\tr{B'})}.
          \Letm{}{\aeq{\shat{f}}}\hat{f} = f \In \transl{M}))\;\ell_f \\[.5ex]
  M_2 &=& \Mod_{\aeq{\ell_f}}\;(\Mod_{\aid}\;(\lambda x^{\tr{A'}} . \Do \ell_f\;x)) \\[.5ex]
  N_1 &=& {\Letm{}{\aeq{\ell_f,\tr{C}}} x'=x \In} x' \\[.5ex]
  N_2 &=& {\Letm{}{\aeq{\transl{C}}} \hat{r} = r \In} \transl{N}
  \ea\]
  For the translations of contexts in \refa{} and \refb{}, we have
  \[\ba{l@{~}c@{~}l}
    {\str{\Gamma, {f:^\ast\carr{A'}{}{B'}}}{}}
    &=& \bl
      \str{\Gamma}{},
        \shat{f},
      f:\boxwith{\aeq{\shat{f}}}(\tr{A'}\to\tr{B'}),
      \hat{f}\varb{\aeq{\shat{f}}}{\tr{A'}\to\tr{B'}} \\[1ex] \el \\
    {\str{\Gamma,p:A',r:^C\carr{B'}{}{A}}{}}
    &=& \bl \str{\Gamma}{}, p:\tr{A'},
      r:\boxwith{\aeq{\transl{C}}}(\tr{B'}\to\tr{A}),
      \hat{r}\varb{\aeq{\transl{C}}}{\tr{B'}\to\tr{A}} \el
  \el\]
  Then by \refa{}, \Cref{lemma:structural-rules},
  \Cref{lemma:pure-promotion-vars}, and several typing rules in
  \Metp{\mcS}, we have
  \[
    \str{\Gamma}{}, \ell_f:\tr{A'}\sto\tr{B'}
    \vdash M_1 :
    \boxwith{\aex{\ell_f}}(\boxwith{\aeq{\ell_f}}(\boxwith{\aid}(\tr{A'}\to\tr{B'}))\to\tr{A})
    \mid \tr{C}
  \]
  which gives the binding of $g$:
  \[
    g \varb{\aex{\ell_f}} \boxwith{\aeq{\ell_f}}\boxwith{\aid}(\tr{A'}\to\tr{B'})\to\tr{A}
  \]
  Then we have
  \[
    \str{\Gamma}{}, \ell_f:\tr{A'}\sto\tr{B'}
    \vdash g\;M_2
    : \tr{A} \atmode{\ell_f,\tr{C}}
  \]
  By \Cref{lemma:pure-promotion-vars}, we have
  \[
    \refc{\str{\Gamma}{}, \ell_f:\tr{A'}\sto\tr{B'}, \lockwith{\aeq{\tr{C}}}, \lockwith{\aex{\ell_f}}
    \vdash g\;M_2 : \ell_f,\tr{C}}
  \]
  By \refb{}, \Cref{lemma:structural-rules}.6, and \tylabi{Letmod}{\Metp{\mcS}}, we have
  \[
    {\typmet{\str{\Gamma}{}, \ell_f:\tr{A'}\sto\tr{B'},
     p:\tr{A'}, r:\boxwith{\aeq{\transl{C}}}(\tr{B'}\to\tr{A})
    }{N_2}{\tr{A}}{\tr{C}}}
  \]
  Again by \Cref{lemma:pure-promotion-vars}, we have
  \[
    \refd{\typmet{\str{\Gamma}{}, \ell_f:\tr{A'}\sto\tr{B'}, \lockwith{\aeq{\tr{C}}}, p:\tr{A'},
      r:\boxwith{\aeq{\transl{C}}}(\tr{B'}\to\tr{A})
      }{N_2}{\tr{A}}{\tr{C}}}
  \]
  For the translated return clause, by \tylabi{Letmod}{\Metp{\mcS}} and $\tr{A}:\Pure$, we have
  \[
    \refe{\typmet{\str{\Gamma}{}, \ell_f:\tr{A'}\sto\tr{B'}, \lockwith{\aeq{\tr{C}}},
     x:\boxwith{\aeq{\ell_f,\tr{C}}}\tr{A}
    }{N_1}{\tr{A}}{\tr{C}}}
  \]
  Our goal follows from \refc{}, \refd{}, \refe{}, \tylabi{Handle}{\Metp{\mcS}}.
\end{proof}

The proof relies on the following lemma.

\begin{restatable}[Subeffecting]{lemma}{SystemCSubeffecting}
  \label{lemma:systemc-subeffecting}
  Given a typing judgement $\Gamma \vdash M : A \mid C$ in \SystemC,
  if $\typm{\stransl{\Gamma}{}}{\transl{M} : \transl{A}}{\transl{C}}$
  and $C \subseteq C'$ then $\typm{\stransl{\Gamma}{}}{\transl{M} :
  \transl{A}}{\transl{C'}}$. Similarly for blocks.
\end{restatable}
\begin{proof}
  By straightforward induction on typing judgements in \SystemC.
\end{proof}

\SystemCToMetpSemantics*
\begin{proof}
  By induction on $M$ and case analysis on the next reduction rule.
  Values in \SystemC are translated to values in \Metp{\mcS}.
  Not all values in \SystemC are translated to value normal forms in
  \Metp{\mcS}, but we can always further reduce them to value normal forms
  in \Metp{\mcS}.
  The theorems allows us to have more steps of reduction in \Metp{\mcS}.
  We do not explicitly mention that we reduce translations of values
  to value normal forms in the following proof.
  \begin{description}
    \item[Case] \focus{\semlab{Box}}
    We have
    \[\ba{rcl}
      \tr{\CUnbox (\CBox P)} &=& \Letm{}{\aeq{\tr{C}}} x = \Mod_{\aeq{\tr{C}}}\;\tr{P} \In x
    \ea\]
    By \semlabi{Letmod}{\Metp{\mcS}}.
    \item[Case] \focus{\semlab{Let}} We have
    \[\ba{rcl}
      \tr{\Let x = \Ret V \In N} &=& \Let x = \tr{V} \In \tr{N}
    \ea\]
    LHS reduces to $N[V/x]$ and RHS reduces to $\tr{N}[\tr{V}/x]$.
    It is easy to show that translation preserves value substitution.
    \item[Case] \focus{\semlab{Def}}
    \[\ba{rcl}
         \Def f = P \In N &\reducesto & N[P/f] \\
    \ea\]
    By translation preserving substitution, we have
    \[\ba{rcl}
    \tr{\LHS} &=&
      \Let f = {\Mod_{\aeq{\transl{C}}}}\,\transl{P} \In
      \Let {\Mod_{\aeq{\transl{C}}}\;} \hat{f} = f \In
      \transl{N} \\[1ex]
    \tr{\RHS} &=& \tr{N}[\tr{P}/\hat{f}]
    \ea\]
    Our goal follows from \semlabi{App}{\Metp{\mcS}} and \semlabi{Letmod}{\Metp{\mcS}}.
    \item[Case] \focus{\semlab{Call}}
    \[
      \block{\ol{x:A}}{\ol{f:T}}{M}(\ol{V}, \ol{Q}) \reducesto M[\ol{V/x}, \ol{Q/f},
      \ol{C/f}]
    \]
    Let $P = {\block{\ol{x:A}}{\ol{f:T}}{M}}$, we have
    \[\ba{rcl}
    \tr{P} &=&
      \Lambda\ol{\shat{f}} .
      \Mod_{\aex{\ol{\shat{f}}}}(
      \lambda\ol{x^{\transl{A}}}\,\ol{f^{\boxwith{\aeq{\shat{f}}}\transl{T}}} .
      \ol{\Letm{}{\aeq{\shat{f}}}\hat{f} = f \In}
      \transl{M}) \\[1ex]
    \tr{P(\ol{V_i},\ol{Q_j})} &=& \Letm{}{\aex{\ol{\transl{C_j}}}} x = \transl{P}\;\ol{\transl{C_j}}\In
      x\;\ol{\transl{V_i}}\;\ol{(\Mod_{\aeq{\transl{C_j}}}\,\transl{Q_j})}
    \ea\]
    Our goal follows from \semlab{Letmod}, \semlab{App} and
    \semlab{TApp} in \Metp{\mcS}, as well as the fact that translation
    preserves value substitution and type substitution.
    \item[Case] \focus{\semlab{Gen}}
    \[\ba{rcl}
    \Try\; \{f^{\carrsingle{A'}{}{B'}}\To M\} \With H \mid \Instctx
    &\reducesto &
    \Try_\ell\; M[\tmcap{\ell}/f, \{\ell\}/f] \With H \mid\Instctx, \ell : \carrsingle{A'}{}{B'}
    \ea\]
    where $\ell$ fresh.
    We have
    \[\ba{rcl}
    \tr{\Try_f \; M \With H} &=&
    \bl
    {\Localeffect{\ell_f : \tr{A'}\sto\tr{B'}}}
    {\Letm{}{\aex{\ell_f}}} g\; = \\[.5ex]
      ({\Lambda\shat{f}}.\Mod_{\aex{\shat{f}}}\;(\lambda
        f^{\boxwith{\aeq{\shat{f}}}\boxwith{\aid}(\tr{A'}\to\tr{B'})}.
      {\Letm{}{\aeq{\shat{f}}}\hat{f} = f \In} \transl{M}) {)\;\ell_f} \\[1ex]
    \!\!\In\Handle^{\aeq{\tr{C}}}\;
      (g\; (\Mod_{\aeq{\ell_f}}\;(\Mod_\aid\;(\lambda x^{\tr{A'}} . \Do \ell_f\;x))))
    \With \tr{H} \el
    \ea\]
  Taking the same $\ell$ as in the reduction of \SystemC, the
  translated term reduces to
  \[
    \Handle^{\aeq{\tr{C}}}\;
    \transl{M}[(\Mod_\aid\;(\lambda x^{\tr{A'}} . \Do\ell\;x))/\hat{f},\ell/\shat{f}] \With \tr{H}
  \]
  which is equal to $\tr{\Try_\ell\; M[\tmcap{\ell}/f, \{\ell\}/f]
  \With H}$ (recall that the translation of runtime capability value
  is $\tr{\tmcap{\ell}} = \Mod_\aid\;(\lambda x^{\tr{A'}} . \Do \ell\;x)$).
  \item[Case] \focus{\semlab{Ret}} By \semlab{Ret} and \semlab{Letmod} in {\Metp{\mcS}}.
  \item[Case] \focus{\semlab{Op}}
  \[\ba{rcl}
    \Try_\ell\; \EC[\tmcap{\ell}(V)] \With H \mid \Instctx
      &\reducesto& N[V/p,  \block{y}{}{\Try_\ell\; \EC[\Ret y] \With H}/r] \mid \Instctx
  \ea\]
  where $\Instctx \ni \ell : \carrsingle{A'}{}{B'}$. We have
  \[\ba{rcl}
  \tr{\Try_\ell\; \EC[\tmcap{\ell}(V)] \With H \gray{\;: A\mid C}}
  &=& \Handle^{\aeq{\tr{C}}}\; \tr{\EC[\tmcap{\ell}(V)]} \With \tr{H^\ell} \\
  \ea\]
  By \Cref{lemma:ectrans-systemc}, we have $\tr{\EC[\tmcap{\ell}(V)]} =
  \tr{\EC}[\Letm{}{\aid} g = \Mod_{\aid}\;(\lambda x^{\tr{A'}} . \Do\ell\;x)\In g\;\tr{V}]$.
  We can reduce $\tr{V}$ to a value normal form in \Metp{\mcS}.
  Our goal follows from \semlab{Letmod}, \semlab{App}, and \semlab{Op} in \Metp{\mcS}.
  Note that the RHS of the above reduction step is translated to
  \[
  \tr{N}[\tr{V}/p, \tr{\block{y}{}{\Try_\ell\; \EC[\Ret y] \With H}}/\hat{r}].
  \]
  \item[Case] \focus{\semlab{Lift}} By IH and \Cref{lemma:ectrans-systemc}.
  \end{description}
\end{proof}

The proof of semantics preservation relies on the following lemma.

\begin{lemma}[Translation of Evaluation Contexts]
  \label{lemma:ectrans-systemc}
  For the translation $\tr{-}$ from \SystemC to \Metp{\mcS}, we have
  $\tr{\EC[M]} = \tr{\EC}[\tr{M}]$ for any
  evaluation context $\EC$ and term $M$.
\end{lemma}
\begin{proof}
  By straightforward induction on evaluation contexts of \SystemC.
  For the case of $\Def f = \EC \In N$, note that $\Mod_\mu\;\tr{\EC}$ is a
   valid evaluation context in \Metp{\mcS}.
\end{proof}

\subsection{Proofs of Encoding \SystemXi in \Metp{\mcS}}
\label{app:proof-systemxi}

\SystemXiToMetp*
\begin{proof}
  By induction on typing judgements in \SystemXi.
  We prove a stronger version which says that
  $\tr{\Gamma}\vdash\tr{M}:\tr{A}\atmode{E}$ for any well-scoped $E$
  in \Metp{\mcS}.
  We need this stronger version to prove the case of named handlers.
  As a visual aid, for each non-trivial case we repeat its typing rule
  in \SystemXi. We replace each premise by the \Metp{\mcS} judgement of
  the translated premise implied by the induction hypothesis. We
  replace the conclusion by the \Metp{\mcS} judgement of the
  translated conclusion that we need to prove.
  When referring to the name of a rule, we sometimes also mention the
  calculus name to disambiguate. For instance, \tylabi{Var}{\SystemXi}
  refers to the rule \tylab{Var} of \SystemXi.
  \begin{description}
    \item[Case] \focus{\Unit} By \tylabi{Unit}{\SystemXi} and \tylabi{Unit}{\Metp{\mcS}}.
    \item[Case] \focus{x} By \tylabi{Var}{\SystemXi} and
      \tylabi{Var}{\Metp{\mcS}}. Variables are always accessible after
      translation as there is no lock in translated contexts at all.
    \item[Case] \focus{f} By \tylabi{BlockVar}{\SystemXi} and
    \tylabi{Var}{\Metp{\mcS}}. Block variables are always accessible after
    translation as there is no lock in translated contexts at all.
    \item[Case] \focus{\block{\ol{x:A}}{\ol{f:T}}{M}}
    \begin{mathpar}
      \inferrule*
      {
        \typmet{\tr{\Gamma},\ol{x:\tr{A}},\ol{f:\tr{T}}}{\tr{M}}{\tr{B}}{E}
      }
      {\typmet{\tr{\Gamma}}{\lambda\ol{x^{\tr{A}}}\,\ol{f^{\tr{T}}}.\tr{M}}{\tr{\carr{\ol{A}}{\ol{T}}{B}}}{E}}
    \end{mathpar}
    We have ${\tr{\carr{\ol{A}}{\ol{T}}{B}}} = \ol{\tr{A}} \to \ol{\tr{T}} \to \tr{B}$.
    Our goal follows from \tylabi{Abs}{\Metp{\mcS}}.
    \item[Case] \focus{\Ret V} By IH.
    \item[Case] \focus{\Let x = M \In N} By IH, \tylabi{Let}{\SystemXi}, and \tylabi{Let}{\Metp{\mcS}}.
    \item[Case] \focus{\Def f = P \In N} By IH, \tylabi{Def}{\SystemXi}, and \tylabi{Let}{\Metp{\mcS}}.
    \item[Case] \focus{P(\ol{V_i}, \ol{Q_j})}
    \begin{mathpar}
      \inferrule*
      {
        \typmet{\tr{\Gamma}}{\tr{P}}{\tr{\carr{\ol{A_i}}{\ol{T_j}}{B}}}{E} \\
        \ol{\typmet{\tr{\Gamma}}{\tr{V_i}}{\tr{A_i}}{E}} \\
        \ol{\typmet{\tr{\Gamma}}{\tr{Q_j}}{\tr{T_j}}{E}}
      }
      {\typmet{\tr{\Gamma}}{\tr{P}\;\ol{\tr{V_i}}\;\ol{\tr{Q_j}}}{\tr{B}}{E}}
    \end{mathpar}
    We have ${\tr{\carr{\ol{A_i}}{\ol{T_j}}{B}}} = \ol{\tr{A_i}} \to \ol{\tr{T_j}} \to \tr{B}$.
    Our goal follows from \tylabi{App}{\Metp{\mcS}}.
    \item[Case] \focus{\Try\;\{f^{\carrsingle{A'}{}{B'}}\To M\}\With \{p\;r\mapsto N\}}
    \begin{mathpar}
      \inferrule*
      {
        \refa{\typmet{\tr{\Gamma},f:\tr{\carrsingle{A'}{}{B'}}}{\tr{M}}{\tr{A}}{E_1}} \text{ for any } E_1 \\
        \refb{\typmet{\tr{\Gamma},p:\tr{A'},r:\tr{\carrsingle{B'}{}{A}}}{\tr{N}}{\tr{A}}{E}} \\
      }
      {\typmet{\tr{\Gamma}}{
        {\bl
        \Localeffect{\ell_f : \tr{A'}\sto \tr{B'}} \\
        \Handle\; (\lambda f . \transl{M})\; (\lambda x^{\tr{A'}} . \Do\ell_f\;x) \With \tr{H^{\gray{f}}}
        \el}
      }{\tr{A}}{E}}
    \end{mathpar}
    By \refa{}, $\tr{\carrsingle{A'}{}{B'}} = \tr{A'}\to\tr{B'}$,
    \mtylab{Extend}, and \Cref{lemma:structural-rules}.3, we have
    \[
      {\typmet{\tr{\Gamma}, \ell_f:\tr{A'}\sto\tr{B'}, \lockwith{\aex{\ell_f}},
              f:\tr{A'}\to\tr{B'}}{
        \tr{M}}{\tr{A}}{\ell_f,E}}
    \]
    which further gives
    \[
      \refc{\typmet{\tr{\Gamma}, \ell_f:\tr{A'}\sto\tr{B'}, \lockwith{\aex{\ell_f}}}{
         (\lambda f . \transl{M})\; (\lambda x^{\tr{A'}} . \Do\ell_f\;x)
      }{\tr{A}}{\ell_f,E}}
    \]
    Our goal follows from \refb{}, $\tr{\carrsingle{B'}{}{A}} =
    \tr{B'}\to\tr{A}$, \refc{}, $\tr{A}:\Pure$, and
    \tylabi{Handle}{\Metp{\mcS}}.
    Recall that our syntactic sugar for $\Handle$ with no modality
    annotation defined in \Cref{sec:typing-metp} allows us to directly
    give type $\tr{B'}\to\tr{A}$ to the continuation $r$ with no
    modality.
  \end{description}
\end{proof}

\SystemXiToMetpSemantics*
\begin{proof}
  By induction on $M$ and case analysis on the next reduction rule.
  Note that values in \SystemXi are translated to value normal forms
  in \Metp{\mcS}.
  \begin{description}
    \item[Case] \focus{\semlab{Let}} Let-binding in \Metp{\mcS} is syntactic
    sugar of lambda application. By \semlabi{App}{\Metp{\mcS}}.
    \item[Case] \focus{\semlab{Def}} Similar to the above case.
    \item[Case] \focus{\semlab{Call}}
    \[\ba{rcl}
      \block{\ol{x:A}}{\ol{f:T}}{M}(\ol{V},\ol{Q}) &\reducesto & M[\ol{V/x}, \ol{Q/f}]
    \ea\]
    We have
    \[\ba{rcl}
      \transl{\block{\ol{x:A}}{\ol{f:T}}{M}(\ol{V},\ol{Q})} &=&
        (\lambda \ol{x^{\transl{A}}}\,\ol{f^{\transl{T}}} . \transl{M})\;\ol{\tr{V}}\;\ol{\tr{Q}}\\
    \ea\]
    By multiple usages of \semlabi{App}{\Metp{\mcS}}.
    \item[Case] \focus{\semlab{Gen}}
    \[\ba{rcl}
    \Try\; \{f^{\carrsingle{A'}{}{B'}}\To M\} \With H \mid {\Instctx}
    &\reducesto&
    \Try_\ell\; M[\tmcap{\ell}/f] \With H \mid {\Instctx,\ell:\carrsingle{A'}{}{B'}}
    \ea\]
    We have
    \[\ba{rcl}
        \transl{\LHS} &=& \Handle\; (\lambda f . \transl{M})\;(\lambda x . \Do\ell\;x)
        \With \tr{H}\\
        \transl{\RHS} &=& \Handle\; \transl{M[\tmcap{\ell}/f]}
        \With \tr{H}\\
        \tr{\tmcap{\ell}} &=& \lambda x . \Do\ell\;x \ea\]
    Our goal follows from \tylabi{Gen}{\Metp{\mcS}} and the fact that
    translation preserves substitution.
    \item[Case] \focus{\semlab{Ret}} By \semlab{Ret} and \semlab{Letmod} in {\Metp{\mcS}}.
    \item[Case] \focus{\semlab{Op}}
    \[\ba{rcl}
    \Try_\ell\; \EC[\tmcap{\ell}(V)] \With H \mid \Instctx
    &\reducesto& N[V/p, \block{y}{}{\Try_\ell\; \EC[\Ret y] \With H}/r] \mid \Instctx
    \ea\]
    where $\Instctx \ni \ell : \carrsingle{A'}{}{B'}$. We have
    \[\ba{rcl}
    \tr{\LHS} &=& \Handle\; \tr{\EC[\tmcap{\ell}(V)]}
        \With \tr{H^\ell}\\
    \ea\]
    By \Cref{lemma:ectrans-systemxi}, we have
    \[
      \tr{\EC[\tmcap{h}(V)]} =
        \tr{\EC}[\tr{\tmcap{\ell}(V)}] = \tr{\EC}[(\lambda x.\Do\ell\;x)\;\tr{V}]
    \]
    Our goal follows from \semlabi{App}{\Metp{\mcS}} and \semlabi{Op}{\Metp{\mcS}}
    \item[Case] \focus{\semlab{Lift}} Follow from IH and \Cref{lemma:ectrans-systemxi}
  \end{description}
\end{proof}

The proof of semantics preservation relies on the following lemma.

\begin{lemma}[Translation of Evaluation Contexts]
  \label{lemma:ectrans-systemxi}
  For the translation $\tr{-}$ from \SystemXi to \Metp{\mcS}, we have
  $\tr{\EC[M]} = \tr{\EC}[\tr{M}]$ for any
  evaluation context $\EC$ and term $M$.
\end{lemma}
\begin{proof}
  By straightforward induction on evaluation contexts of \SystemXi.
\end{proof}

\subsection{Proofs of Encoding \Fepssn in \Metp{\mcS}}
\label{app:proof-fepssn}

\FepssnToMetp*
\begin{proof}
By induction on typing judgements $\Gamma\vdash M : A\mid E$ in
\Fepssn.
Most cases are similar to those in the proof of encoding \Feps in
\Cref{app:proof-feps}.
We elaborate cases relevant to named handlers.
When referring to the name of a rule, we sometimes also mention the
calculus name to disambiguate. For instance, \tylabi{Var}{\Feps}
refers to the rule \tylab{Var} of \Feps.
\begin{description}
  \item[Case] \focus{\Do V\;W}
    \begin{mathpar}
      \inferrule*[Lab=\tylab{DoName}]
      {
        \Sigma \ni \ell : A\sto B \\
        \refa{\typf{\Gamma}{V}{\tyev{\ell}{a}}{}} \\
        \refb{\typf{\Gamma}{W}{A}{}} \\
      }
      {\typf{\Gamma}{\Do V\;W}{B}{\ell^a, E}}
    \end{mathpar}
    By IH on \refa{} and \refb{} and \Cref{lemma:fepssn-subeffecting}, we have
    \[\bl
      \typmet{\transl{\Gamma}}{\transl{V}}{\boxwith{\aeq{a}}(\transl{A}\to\transl{B})}{\transl{E}} \\
      \typmet{\transl{\Gamma}}{\transl{W}}{\transl{A}}{\transl{E}} \\
    \el\]
    By \tylabi{Letmod}{\Metp{\mcS}} and \tylabi{App}{\Metp{\mcS}}, we have
    \[
      \typmet{\transl{\Gamma}}{\Letm{}{\aeq{a}} x = \transl{V} \In x\;\transl{W}}{\transl{B}}{\transl{E}}
    \]
  \item[Case] \focus{\NHandler\;\{\ell\;p\;r\mapsto N\}}
    \begin{mathpar}
      \inferrule*[Lab=\tylab{NamedHandler}]
      {
        \Sigma \ni \ell : A'\sto B' \\
        \refa{\typf{\Gamma, p : A', r: \earr{B'}{A}{E}}{N}{A}{E}} \\
      }
      {\typf{\Gamma}{\NHandler\; \{\ell\;p\;r \mapsto N \}}{
        \earr{(\forall a^{\Scope(\ell)} . \earr{\tyev{\ell}{a}}{A}{\ell^a,E})}{A}{E}}{}}
    \end{mathpar}
    By IH on \refa{} and \Cref{lemma:pure-promotion-vars}, we have
    \[
      \refb{\typmet{\str{\Gamma}{},
      \lockwith{\aeq{\tr{E}}_{\cdot}},
      \lockwith{\aeq{\tr{E}}_{\tr{E}}}, p : \tr{A'}, r: \boxwith{\aeq{\tr{E}}}({\tr{B'}}\to {\tr{A}})}{\tr{N}}{\tr{A}}{\tr{E}}} \\
    \]

    By \tylabi{Letmod}{\Metp{\mcS}}, \tylabi{Var}{\Metp{\mcS}},
    and $\tr{A} : \Pure$, we have
    \[
      \refc{\typmet{
        \str{\Gamma}{},
        \lockwith{\aeq{\tr{E}}_{\cdot}},
        \lockwith{\aeq{\tr{E}}_{\tr{E}}},
        x:\boxwith{\aeq{\ell_a,\tr{E}}}\tr{A}}{
        \Letm{}{\aeq{\ell_a,\tr{E}}} x = x \In x
      }{\tr{A}}{\tr{E}}}
    \]
    By \tylabi{Var}{\Metp{\mcS}}, \tylabi{Mod}{\Metp{\mcS}}, and \tylabi{App}{\Metp{\mcS}}, we have
    \[\refd{\ba{r}
      {\str{\Gamma}{},
      \lockwith{\aeq{\tr{E}}},
          f' \varb{\aeq{\ell_a,\tr{E}}} \boxwith{\aeq{\ell_a}}(\tr{A'}\sto\tr{B'}) \to\tr{A},
          \lockwith{\aeq{\tr{E}}},
          \lockwith{\aex{\ell_a}}
          }\vdash\qquad\qquad\qquad \\[.5ex] {
        f'\;(\Mod_{\aeq{\ell_a}}\;(\lambda x . \Do\ell_a\;x))}
        : {\tr{A}} \atmode{\ell_a,\tr{E}}
    \ea}\]
    By \tylabi{Handle}{\Metp{\mcS}}, \refb{}, \refc{}, \refd{}, and
    \Cref{lemma:pure-promotion-vars}, we have
    \[\bl
      {\str{\Gamma}{},
      \lockwith{\aeq{\tr{E}}}, f : \boxwith{\aeq{a,\tr{E}}}(\boxwith{\aeq{a}}(\tr{A'}\sto\tr{B'}) \to\tr{A})
      }\vdash
        \Localeffect{\ell_a : \tr{A'}\sto\tr{B'}} \\[.5ex]
        \quad \Handle^{\aeq{\tr{E}}}\;
        (\Letm{}{\aeq{\ell_a,\tr{E}}} f' = f\;\ell_a \In
        f'\;(\Mod_{\aeq{\ell_a}}\;(\lambda x . \Do\ell_a\;x)))
        \With \tr{H}
      :{\tr{A}}\atmode{\tr{E}}
    \el\]
    Finally our final goal follows from
    \tylabi{Abs}{\Metp{\mcS}} and \tylabi{Mod}{\Metp{\mcS}}.
\end{description}
\end{proof}

The proof of type preservation relies on the following lemma.

\begin{restatable}[Pure Values]{lemma}{FepssnSubeffecting}
  \label{lemma:fepssn-subeffecting}
  Given a typing judgement $\typf{\Gamma}{V}{A}{}$ in \Fepssn,
  if $\typm{\stransl{\Gamma}{\cdot}}{\transl{V} : \transl{A}}{\cdot}$
  then $\typm{\stransl{\Gamma}{E}}{\transl{V} :
  \transl{A}}{\transl{E}}$ for any $E$.
\end{restatable}
\begin{proof}
  By straightforward induction on typing judgements of values in
  \Fepssn.
  The most non-trivial case is to show the accessibility of variables.
  Observe that the change from $\str{\Gamma}{\cdot}$ to
  $\str{\Gamma}{E}$ only changes the translations of locks.
  After translation, all variables in the context have types of
  kind $\Pure$.
  Their accessibility follows from \Cref{lemma:pure-promotion-vars}.
\end{proof}

\FepssnToMetpSemantics*
\begin{proof}
  By induction on $M$ and case analysis on the next reduction rule.
  Note that values in \Feps are translated to value normal forms in \Metp{\emtScp}.
  Most cases are similar to those in the proof of encoding \Feps in
  \Cref{app:proof-feps}.
  We show new cases relevant to named handlers.
  \begin{description}
    \item[Case] \focus{\semlab{Gen}} Suppose the effect row of the whole term is $E$.
    \[\ba{rcl}
    \NHandler\;H\;V \mid \Instctx
      &\reducesto&
      \Handle_h\; (\Let x = V\;b \In x\;\tmev{h}) \With H \mid \Instctx,h:\ell^b
    \ea\]
    where $b,h$ fresh and $\Sigma \ni \ell : A'\sto B'$.
    We have
    \[\ba{rcl}
    \tr{\LHS} &=& \Letm{}{\aeq{\tr{E}}} g = \Mod_{\aeq{\tr{E}}}\;\tr{\NHandler\;H} \In g\;\tr{V} \\[1ex]
    \tr{\NHandler\;H} &=&
    \bl
    \Mod_{\aeq{\transl{E}}}\;(
    \lambda f .
      \Localeffect{\ell_a : \tr{A'}\sto \tr{B'}} \\
      \quad\Handle^{\aeq{\tr{E}}}\;
        (\Letm{}{\aeq{\ell_a,\transl{E}}} f' = f\;\ell_a \In
          f'\;(\Mod_{\aeq{\ell_a}}\;(\lambda x . \Do\ell_a\;x))) \\
      \quad\!\!\With \tr{H})\el \\
    \ea\]
    By \semlab{Letmod}, \semlab{App}, and \semlab{Gen} (use the runtime label $\ell_b$)
    in \Metp{\mcS}, $\tr{\LHS}$ reduces to
    \[
        \Handle^{\aeq{\tr{E}}}\;
        \bl(\Letm{}{\aeq{\ell_b,\transl{E}}} f' = f\;\ell_b \In \\
          f'\;(\Mod_{\aeq{\ell_b}}\;(\lambda x . \Do \ell_b\;x))) \With \tr{H}\el
    \]
    which is equal to $\tr{\RHS}$ of the above reduction step.
    Note that the runtime generated scope variable $b$ is translated to $\ell_b$.
    \item[Case] \focus{\semlab{NRet}} By \semlab{Ret} and \semlab{Letmod} in \Metp{\mcS}.
    \item[Case] \focus{\semlab{NOp}} Suppose the effect row of the whole term is $E$.
    \[\ba{rcl}
      \Handle_h\; \EC[\Do \tmev{h} \; V] \With H
      &\reducesto& N[V/p, (\lambda y.\Handle_h\; \EC[\Ret y] \With H)/r]
    \ea\]
    where $\Instctx \ni h:\ell^b$. We have
    \[\ba{rcl}
    \tr{\LHS} &=& \Handle^{\aeq{\tr{E}}}\; \tr{\EC[\Do \tmev{h} \; V]} \With \tr{H}
    \ea\]
    By \Cref{lemma:ectrans-fepssn}, we have
    \[\ba{rcl}
      \tr{\EC[\Do \tmev{h} \; V]} &=& \tr{\EC}[\Letm{}{\aeq{\ell_b}} f = \tr{\tmev{h}} \In f\;\tr{V}] \\
      &=&  \tr{\EC}[\Letm{}{\aeq{\ell_b}} f = \Mod_{\aeq{\ell_b}}\;(\lambda x .\Do\ell_b\;x) \In f\;\tr{V}]
    \ea\]
    Then by \semlab{Letmod} and \semlab{App} in \Metp{\mcS}, $\tr{\LHS}$ reduces to
    \[
      \Handle^{\aeq{\tr{E}}}\; \tr{\EC}[\Do\ell_b\;\tr{V}] \With \tr{H}
    \]
    Our goal follows from \semlab{Op} in \Metp{\mcS}.
  \end{description}
\end{proof}

The proof of semantics preservation relies on the following lemma.

\begin{lemma}[Translation of Evaluation Contexts]
  \label{lemma:ectrans-fepssn}
  For the translation $\tr{-}$ from \Fepssn to \Metp{\mcS}, we have
  $\tr{\EC[M]} = \tr{\EC}[\tr{M}]$ for any
  evaluation context $\EC$ and term $M$.
\end{lemma}
\begin{proof}
  By straightforward induction on evaluation contexts of \Fepssn.
\end{proof}
\fi

\end{document}